\documentclass[10pt,twoside,twocolumn]{IEEEtran}
\usepackage[dvipsnames,table]{xcolor}
\usepackage{algorithmic}
\usepackage{graphicx}
\usepackage{textcomp}
\definecolor{subsectioncolor}{rgb}{0,0.541,0.855}
\usepackage{amsmath, amssymb, amsfonts, enumerate, bbold}
\usepackage{amsthm}
\usepackage[font=small,labelfont=bf]{caption}
\usepackage{tikz}
\usetikzlibrary{arrows.meta,positioning,calc,decorations.pathmorphing,intersections,fit,patterns,backgrounds}
\usepackage{tcolorbox}
\usepackage{hyperref}
\hypersetup{
colorlinks = true,
citecolor = NavyBlue!80!Black,
urlcolor = NavyBlue!80!Black,
linkcolor = NavyBlue!80!Black
}
\usepackage{wrapfig,graphicx,lipsum}
\usepackage{subcaption}
\usepackage{bm}
\usepackage{bbm}
\usepackage{mathtools}  
\usepackage{booktabs}
\usepackage{tabularx}
\usepackage{multirow}
\usepackage[numbers,sort&compress]{natbib}
\makeatletter
\let\journalname\relax
\makeatother

\providecommand{\journalname}{IEEE Transactions}
\renewcommand{\journalname}{IEEE Transactions}
\newtheorem{definition}{Definition}
\newtheorem{lemma}{Lemma}
\newtheorem{remark}{Remark}
\newtheorem{proposition}{Proposition}
\newtheorem{theorem}{Theorem}
\newtheorem{corollary}{Corollary}
\newtheorem{example}{Example}
\def\nr{n_\mathrm{R}}
\def\np{n_\mathrm{P}}

\def\mt{m_\mathrm{T}}
\def\mp{m_\mathrm{P}}
\def\nc{n_\mathrm{C}}
\def\lp{\ell_\mathrm{P}}
\def\lt{\ell_\mathrm{T}}
\def\Alr{\mathbf A^\mathrm{LR}}
\def\Acb{\mathbf A^\mathrm{CB}}
\def\Acr{\mathbf A^\mathrm{CR}}

\newcommand{\mh}[1]{{\color{black}{#1}}}
\newcommand{\new}[1]{{\color{black}{#1}}}
\newcommand{\rev}[1]{{\color{black}{#1}}}
\newcommand{\condense}[1]{{\color{black}{#1}}}
\newcommand{\red}[1]{{\color{black}{#1}}}
\newcommand{\jqe}[1]{{\color{black} {#1}}}
\usepackage{etoolbox}

\def\BibTeX{{\rm B\kern-.05em{\sc i\kern-.025em b}\kern-.08em
    T\kern-.1667em\lower.7ex\hbox{E}\kern-.125emX}}
\begin{document}
\title{Braess' Paradoxes in Coupled Power and Transportation Systems}
\author{Minghao Mou and Junjie Qin\thanks{M.M. and J.Q. are with the Elmore Family School of Electrical and Computer Engineering at Purdue University, West Lafayette, IN, USA. Emails: \texttt{\{mmou,jq\}@purdue.edu}.}
}

\maketitle

\begin{abstract}
Transportation electrification introduces strong coupling between the power and transportation systems. In this paper, we generalize the classical notion of Braess' paradox to coupled power and transportation systems, and examine how the cross-system coupling induces new types of Braess' paradoxes. To this end, we model the power and transportation networks as graphs, coupled with charging points connecting to nodes in both graphs. The power system operation is characterized by the economic dispatch optimization, while the transportation system user equilibrium models travelers' route and charging choices. By analyzing simple coupled systems, we demonstrate that capacity expansion in either the transportation or the power system can deteriorate the performance of both systems, and uncover the fundamental mechanisms for such new Braess' paradoxes to occur. We also provide necessary and sufficient conditions for the occurrence of Braess' paradoxes in general coupled systems, leading to managerial insights for infrastructure planners. \rev{For general networks, by characterizing generalized user equilibria of coupled systems, we develop charging-pricing policies to mitigate Braess' paradoxes.} 
\end{abstract}


\section{Introduction}
The global trend of transportation electrification calls for massive upgrades of supporting infrastructure. In addition to the rapidly expanding charging infrastructure \citep{afdc2024trends}, accommodating the increasing charging loads requires additional power grid capacities, which are usually achieved by costly and time-consuming grid asset upgrades.
 Independently, like many countries in the world, the U.S. transportation infrastructure requires continuous investments to address aging roads/structures and accommodate increasing usage. Given budgetary constraints, it is therefore critical to strategically plan and invest in infrastructure upgrades to streamline transportation electrification while improving the performance of the power and transportation systems. 

Meanwhile, transportation electrification deepens the \emph{coupling} between the transportation and power systems. Electric vehicle (EV) routing and charging choices determine the spatial charging load profile, which affects dispatch and locational marginal prices (LMPs) in the grid; these prices in turn enter drivers' decision-making and reshape equilibrium travel/\rev{charging} patterns. This feedback motivates planning and operational questions that cannot be addressed by studying the two infrastructures in isolation.

A key concern in infrastructure planning is whether capacity expansion necessarily improves performance. In transportation networks, Braess' paradox (BP) shows that under selfish routing, adding capacity can \emph{increase} equilibrium travel time \cite{braess2005paradox}. Similar paradoxical effects appear in power systems  \cite{schafer2022understanding,wu1996folk}. In coupled power and transportation systems, capacity changes in either system can propagate through the feedback loop, potentially creating \emph{cross-network} 
\jqe{paradoxical outcomes}. This paper asks: \emph{Can the closed-loop coupling between power and transportation systems manifest new mechanisms and types of \mh{BPs}? If yes, how to screen and mitigate them?}

\subsection{\jqe{Contributions and Paper Organization}}
We study a coupled system where transportation and power networks are interconnected through charging locations. EV drivers are modeled as a nonatomic population that selects routes and charging points under a transportation user equilibrium (UE). The power network is modeled by an economic dispatch problem that maps induced charging loads to LMPs and generation schedules. \new{The equilibrium of the coupled system is characterized by the fixed point of the mutual dependence between traffic flows and LMPs: given LMPs, EV travelers are at UE, and given the induced charging demand, LMPs satisfy economic dispatch. We refer to such an equilibrium as a \emph{generalized user equilibrium} (GUE).} We analyze GUE sensitivities to road and transmission expansions, derive BP conditions, and study mitigation via charging pricing. Our contributions are: 

\new{
\emph{1) Generalized BP.} We formalize BP in coupled power and transportation systems induced by transportation or power expansions, and distinguish whether the BP appears in the same network or in the other network through the coupling. Specifically, we study Transportation/Power expansion induced BP in Transportation/Power systems, abbreviated as type T-T, T-P, P-T, and P-P.
}

\emph{2) BP \new{m}echanisms and characterization.} We identify \new{the underlying} mechanisms of each BP type, and derive necessary and sufficient conditions for BP occurrences across increasingly general \jqe{power network }congestion regimes, including uncongested networks \rev{(Theorem~\ref{thm:uncongested})} and the fully congested radial power network case \rev{(Theorem~\ref{thm:fully_congested})}, with arbitrary radial congestion handled through aggregation.


\emph{\rev{3)} BP \new{m}itigation.} We propose alternative charging pricing policies, including system-optimal adaptive prices that steer the GUE toward minimizing transportation, power, or joint costs \rev{(Theorem~\ref{thm:sysoptPi})}, and a convex formulation that identifies (when feasible) static prices that preclude all BP types \rev{(Theorem~\ref{thm:BP_convex})}.

\condense{The rest of the paper is organized as follows. Section~\ref{sec:literature} reviews related work; Section~\ref{sec:model} introduces the model, GUE, and generalized BP; Sections~\ref{sec:examples}--\ref{sec:mitigation} study BP mechanisms, conditions, and mitigation; Section~\ref{sec:numerical} gives numerical results; and Section~\ref{sec:Conclusion} concludes.}

\subsection{Related Literature} \label{sec:literature}
Our work builds on a growing literature on coupled power and transportation systems, with a particular focus on network equilibria and inefficiencies arising from network expansion. The related work can be grouped into three streams: classical \mh{BP}, coupled system modeling, and optimal pricing.

\subsubsection{Classical Braess' Paradox}
Braess' paradox was first identified in transportation networks and has since been extensively studied \cite{braess2005paradox,steinberg1983prevalence,pas1997braess}, with empirical evidence reported in \cite{fisk1981empirical}. Related paradoxical effects have also been observed in power systems, where power line expansion may worsen equilibrium by increasing generation cost or creating unintended inefficiencies \cite{wu1996folk,schafer2022understanding}. Several \new{studies} characterize \jqe{the conditions under which} BP arises in transportation networks. For example, \new{Pas and Principio} \cite{pas1997braess} shows that paradoxical outcomes depend critically on both travel demand and congestion characteristics, while \new{Milchtaich} \cite{milchtaich2006network} proves that under separable, strictly increasing travel costs, the occurrence of BP depends only on network topology. \new{Acemoglu et al. \cite{acemoglu2018informational} extend BP to traffic networks with heterogeneous information, where travelers know different subsets of routes (information sets). They show that enlarging the information set of a population, thereby giving its users access to additional routes, can paradoxically increase that population's equilibrium travel cost.} In contrast to single-network settings, we study BP induced by the \emph{coupling} between transportation and power systems.

\subsubsection{Coupled Power and Transportation System Modeling}
A growing body of work develops mathematical models for electrified transportation. These studies mostly consider two operating paradigms: \emph{decentralized operation} \cite{wei2017network,he2013integrated,alizadeh2016optimal}, in which travelers make self-interested routing and charging decisions while the power grid is operated independently; and \emph{centralized or coordinated operation} \cite{wei2016optimal}, in which a planner jointly manages both infrastructures. Existing work largely focuses on equilibrium characterization and computation. Consistent with much of this literature, we adopt a decentralized equilibrium framework, but address a different question: how \emph{capacity perturbations} in either infrastructure propagate through and under what conditions they induce BP.

\subsubsection{Optimal Pricing in Coupled Power and Transportation Systems}
Pricing is a key mechanism for aligning individual decisions with system-level objectives. Electricity tariffs or road tolls designed in isolation generally fail to capture cross-network externalities. Optimal pricing frameworks jointly design tariffs and tolls to reflect marginal costs across the systems and reduce the inefficiency of decentralized behavior. This literature can be broadly divided into \textit{\mh{centralized or cooperative pricing}} \cite{alizadeh2016optimal,cui2021optimal}, where one planner or multiple coordinated operators jointly set prices and tolls to minimize total social cost, and \textit{\mh{non-cooperative pricing}} \cite{lu2024optimal,li2022strategic}, where multiple operators act strategically and electricity prices emerge from game-theoretic interactions. \rev{Although tolling and pricing remedies for classical BP have been studied \cite{lawphongpanich2010solving}, our focus is different: we study charging-pricing mitigation for BPs induced by road- and transmission-capacity perturbations in coupled power and transportation systems, where pricing changes coupled system equilibria.}
 

\section{Model}\label{sec:model}
\noindent \textbf{Notation:} We denote $[n]$ as the set of positive integers no more than $n$, and $\mathbbm{1}\{\cdot\}$ as the indicator function. We also adopt the convention that $\mathbf{a}$ stands for a vector with $a_i$ denoting its $i$-th element, and $\mathbf{A}$ for a matrix. For an $\mathbb R^n$-valued differentiable function $\mathbf h(x)$, we use $\partial \mathbf h/\partial x$ to denote the $n$-dimensional vector whose $i$-th element is $\partial h_i /\partial x$.

\subsection{Coupled System Model}
\subsubsection{Power Network Model}
We model the power grid as a  graph $\mathcal{G}_\mathrm{P} = (\mathcal{V}_\mathrm{P},\mathcal{E}_\mathrm{P})$ with nodes modeling buses and edges modeling the power lines.  Denote the number of buses by $n_\mathrm{P}$ 
and power injection by $\mathbf p \in \mathbb R^{n_\mathrm{P}}$. 
The (linearized) power flow constraints for the power network can be 
written as 
\begin{equation}
	 \jqe{\mathbf 1^\top \mathbf p = 0 \quad \mbox{and} \quad \mathbf H \mathbf p \leq \bar {\mathbf f},}
\end{equation}
where the equality constraint enforces network-wide power balance, and the inequality constraint imposes the line flow limit with $m_\mathrm{P}$ being the number of flow constraints and $\new{\mathbf{\bar f} > \mathbf{0}}$ modeling the flow capacities. \rev{Here matrix $\mathbf H\in \mathbb{R}^{m_\mathrm{P} \times n_\mathrm{P}}$, commonly referred to as the \emph{shift factor matrix}, maps injections to line flows. We compute $\mathbf H$ by selecting bus $1$ as the slack/reference bus, setting its voltage-angle deviation to zero, and eliminating that bus when solving the DC power-flow equations. Thus, $\mathbf H$ is the bus-$1$-slack shift-factor matrix. Choosing a different reference bus changes the matrix representation, but not the physical line flows $\mathbf H\mathbf p$ for balanced injections satisfying $\mathbf 1^\top \mathbf p=0$. In all numerical studies (Section~\ref{sec:numerical}), we use bus $1$ as the reference bus.} See \cite{qin2018submodularity} for details on deriving these constraints. 

\subsubsection{Transportation Network Model}
\condense{We model the transportation system as a graph $\mathcal G_\mathrm{T}=(\mathcal{V}_\mathrm{T}, \mathcal{E}_\mathrm{T})$, whose nodes are locations and links are roads. Denote the number of links by $m_\mathrm{T}$. We consider EV travelers sharing the same \emph{origin} and \emph{destination}}\footnote{We allow the possibility of \emph{substitutable destinations}, in which case there are multiple nodes $\mathcal V_\mathrm{T}^\mathrm{dest} \subset \mathcal V_\mathrm{T}$ and it is indifferent for each traveler to arrive at any $v^\mathrm{dest}_\mathrm{T}\in \mathcal V_\mathrm{T}^\mathrm{dest}$. This can be used to model, e.g., grocery trips to similar nearby grocery stores. \mh{Results in this paper can be easily generalized to the case with multiple origin-destination pairs. See details in \hyperref[sec:generalization_to_multiple_O-D_pairs]{Appendix~B}.}}\condense{, both in $\mathcal V_\mathrm{T}$, and model them as a continuum $\mathcal J\triangleq[0,1]$. Let $n_\mathrm{R}$ be the number of origin-destination \emph{routes}. Each route $r=1,\dots,n_\mathrm{R}$ traverses a finite subset of links. Route-link membership is summarized by the \emph{link-route incidence matrix} $\mathbf A^\mathrm{LR} \in \mathbb R^{m_\mathrm{T}\times n_\mathrm{R}}$, whose $(\ell,r)$-entry is $\mathbbm{1}\{\mbox{if route $r$ contains link $\ell$}\}$.}
 
 \subsubsection{Charging Points} We assume each EV $j\in \mathcal J$ must charge at one of $n_\mathrm{C}$ charging points located at the nodes en route (including the origin and destination). Furthermore, we associate each route with exactly one charging location\footnote{This can be done without loss of generality. In the event that a route spans $\mh{Z}>1$ charging points, we can create $\mh{Z}$ copies of the same route and associate the $\mh{z}$-th copy with the $\mh{z}$-th charging point. Effectively, each route in this paper can be viewed as a travel and charging plan for an EV driver.}. We define the \emph{charger-route incidence matrix} $\mathbf \Acr \in \mathbb R^{n_\mathrm{C}\times n_\mathrm{R}}$, whose $(\mh{z},r)$-th entry is $\mathbbm{1}\{\mbox{if charging point $\mh{z}$ is associated with route $r$}\}$, to summarize how charging points are associated with routes. 
 Under our setting, each column of $\Acr$ contains exactly one nonzero entry that is 1. Since every charger is connected to a bus, we define the \emph{charger-bus incidence matrix} $\Acb \in \mathbb R^{n_\mathrm{C}\times n_\mathrm{P}}$, whose $(\mh{z},i)$-th entry is $\mathbbm{1}\{\mbox{if charger $\mh{z}$ is connected to bus $i$}\}$.
 
 \subsubsection{Traveler Choices}
 \condense{Each traveler $j \in \mathcal J$ picks a route, which also fixes the charging point. Let $x_r \in [0,1]$ be the fraction of travelers choosing route $r$, and let $\mathbf x \in \mathbb R^{n_\mathrm{R}}$ collect these fractions. A travel pattern $\mathbf x$ is \emph{admissible} if $\mathbf 1^\top \mathbf x =1$ and $\mathbf x \ge \mathbf 0$.}
 
 Given the traveler choices summarized by $\mathbf x$, it is easy to verify that the \emph{link flows}, i.e., the fractions of the travelers driving on the links, can be written as $\Alr \mathbf x \in \mathbb R^{m_\mathrm{T}}$. The fractions of travelers charging at different charging points can be obtained from the elements in the vector $\Acr \mathbf x \in \mathbb R^{n_\mathrm{C}}$. 
 
 Travelers make their route/charging choices to minimize their cost associated with traffic congestion delays and charging. As commonly done in the literature \cite{alizadeh2016optimal}, we model the travel cost of link $\ell$ as an affine function of the flow on the link, i.e., $\alpha_\ell (\Alr \mathbf x)_\ell + \beta_\ell$, where $\alpha_\ell \ge 0$ \emph{inversely} scales with the capacity of the link, and $\beta_\ell \ge 0 $ models the travel cost without traffic. By summing the link-based travel costs of all links in a route, we obtain the travel cost for a traveler picking route $r$ as
\begin{equation}
	c^\mathrm{tr}_r (\mathbf x)= (\Alr )^\top_r \left[\mathrm{diag}(\bm \alpha) \Alr \mathbf x + \bm \beta\right],
\end{equation}
 where $(\Alr )^\top_r$ is the $r$-th row of the matrix $(\Alr )^\top$. 
 
 The charging cost of a traveler depends on the electricity price at the charging point en route. Denote the electricity price at bus $i$ by $\lambda_i$. Given the vector of prices $\bm \lambda\in \mathbb R^{n_\mathrm{P}}$, we can obtain the prices at the charging locations as $\Acb \bm \lambda \in \mathbb R^{n_\mathrm{C}}$, and the prices for the charging locations associated with the routes as $(\Acr)^\top\Acb \bm \lambda \in \mathbb R^{n_\mathrm{R}}$. 
 We assume that all the travelers have an identical charging energy need\footnote{\jqe{Non-identical charging needs, even vehicles with no charging needs (e.g., gasoline cars), can be incorporated by considering a finite number of vehicle groups, each with its own $\rho$ value. Such an extension is omitted due to the page limit.}} $\rho>0$. Then the charging cost of a traveler picking route $r$ is 
 \begin{equation}\label{eq:pi}
 	\pi_r\triangleq\pi_r (\bm \lambda) = \rho (\Acr)_r^\top\Acb \bm \lambda,
 \end{equation}
 where $(\Acr)_r^\top$ denotes the $r$-th row of the matrix $(\Acr)^\top$, and the total cost of a traveler picking route $r$ is then 
 \begin{equation}\label{eq:t:cost}
 	c_r(\mathbf x,\ \bm \lambda) = (\Alr )^\top_r \left[\mathrm{diag}({\bm \alpha}) \Alr \mathbf x + \bm \beta\right] + \pi_r(\bm \lambda).
 \end{equation}
 Given~\eqref{eq:pi} and~\eqref{eq:t:cost}, the vector-valued functions $\bm \pi(\bm \lambda)$ and $\mathbf c(\mathbf x, \bm \lambda)$ are defined accordingly.

  \subsection{Economic Dispatch for the Power Network}
\condense{Given the spatial power-load distribution, \emph{economic dispatch} (ED) determines generator dispatch and \emph{locational marginal prices} (LMPs). We study how travel decisions affect dispatch and how the resulting prices affect travel.}

To this end, let $\mathbf d(\mathbf x)\triangleq\rho (\Acb)^\top \Acr \mathbf x \in \mathbb R^{n_\mathrm{P}}$ be the charging loads at different buses induced by the travel choices. Given the spatial power load profile, we consider the following economic dispatch optimization
\begin{subequations}\label{eq:economic_dispatch}
    \begin{align}
        \Phi_\mathrm{P}(\mathbf x) \triangleq \min_{\mathbf g \in \mathbb{R}^{\np}, \ \mathbf p \in \mathbb R^{\np}} \quad & \frac{1}{2} \mathbf g^\top \mathbf Q \mathbf g + \bm \mu^\top \mathbf g;\\
        \mathrm{s.t.} \qquad \quad& \bm \lambda : \mathbf g - \mathbf d(\mathbf x)   = \mathbf p \label{eq:lmp_constraint};\\
        &\gamma: \mathbf 1^\top \mathbf p = 0;\\
        &\bm \eta: \mathbf H \mathbf p \le \bar{\mathbf f},
    \end{align}
\end{subequations}
where $\mathbf g$ is the power generation for generators located at different buses, $\mathbf p$ is the net power injection, $\mathbf Q \in \mathbb R^{\np \times \np}$ is \rev{a diagonal positive semidefinite matrix whose strictly positive diagonal entries correspond to generator buses} \mh{and for simplicity we use $Q_i$ to denote the $i$-th diagonal entry of $\mathbf{Q}$}, $\bm \mu \in \mathbb R^{\np}$ models the linear generation cost coefficients, and
$\bm \lambda\in \mathbb R^{\np}$, $\gamma \in \mathbb R$, and $\bm \eta \in \mathbb R^{m_\mathrm{P}}$ are dual variables associated with the corresponding constraints.

Problem \eqref{eq:economic_dispatch} is a parametric quadratic program: for different travel patterns $\mathbf x$, the problem has different optimal values and dual solutions. In particular, we denote the optimal value 
by $\Phi_\mathrm{P}(\mathbf x)$ and the dual solution associated with the power balance constraint\footnote{Except in Section~\ref{sec:mitigation}, we assume a simple pass-through price setting so the electricity price for a charging point is identical to the LMP at the corresponding bus.  In practice, the LMPs are prices in the wholesale electricity market, which may not be the same as the electricity price for charging stations connected to the bus via a power distribution network managed by a utility company. However, given the utility business model, we argue that the spatial profile of the long-term averages of wholesale electricity prices should resemble a similar pattern as that of retail electricity prices. }~\eqref{eq:lmp_constraint} by $\bm \lambda^\star(\mathbf x)$. When optimality is clear from the context, we sometimes simply use $\bm \lambda(\mathbf x)$. 

\begin{remark}\label{rem:compact_ED}
\rev{Problem~\eqref{eq:economic_dispatch} uses bus-indexed generation variables to keep later notation compact. If $\mathcal G$ denotes the set of generator buses, the compact notation should be read with $Q_i>0$ for $i\in\mathcal G$, $Q_i=0$ allowed only for $i\notin\mathcal G$, and the extra constraints $g_i=0$ for all $i\notin\mathcal G$. Thus, diagonal entries of $\mathbf Q$ at non-generator buses are not free generation cost coefficients. Equivalently, one may write ED only in the generator variables $\tilde{\mathbf g}\in\mathbb R^{|\mathcal G|}$ with a reduced quadratic cost matrix $\tilde{\mathbf Q}\succ0$; the bus-indexed notation embeds this reduced formulation into the bus space. Fixed non-EV demands can be included by replacing $\mathbf d(\mathbf x)$ in~\eqref{eq:lmp_constraint} with $\mathbf d_0+\mathbf d(\mathbf x)$, as in our numerical study in Section~\ref{sec:numerical}.}
\end{remark}

\subsection{Generalized User Equilibrium}
\condense{We have described how travel decisions affect charging prices through economic dispatch. The reverse dependence, from prices to travel patterns, is characterized through \emph{Generalized User Equilibrium} (GUE).}

\condense{In transportation, \emph{User Equilibrium} (UE), or Wardrop Equilibrium, models decentralized route choices in which no traveler can reduce cost by unilaterally changing routes. Since our travelers care about both travel and charging costs, we use the following UE notion for fixed electricity prices $\bm\lambda$.}
\begin{definition}[Transportation UE given LMPs]\label{def:t:eq}
	Given fixed LMPs $\bm \lambda$, an admissible travel pattern $\mathbf x^\star$ is said to be a transportation UE if for any $r \in [n_\mathrm{R}]$ such that $\mathbf x^\star_r >0$, we have $c_r(\mathbf x^\star,\ \bm \lambda) \le c_{r'}(\mathbf x^\star,\ \bm \lambda)$ for all $r' \in [n_\mathrm{R}]$.  
\end{definition}

It is not hard to see that Definition \ref{def:t:eq} coincides with the notion of Nash Equilibrium defined for the non-atomic game with players $\mathcal J$ and cost defined by \eqref{eq:t:cost}. Similar to the standard results for UE, we have the following alternative characterization: 
\begin{lemma}[Uniform Cost]\label{lemma:equal_cost}
   If $\mathbf x^\star$ constitutes a \jqe{transportation }UE given $\bm \lambda$, there exists a constant $C \geq 0$ such that $c_r(\mathbf x^\star,\ \bm \lambda) \equiv C$, for all $r$ such that $x^\star_r > 0$.
\end{lemma}

\condense{Since $\bm \lambda$ depends on $\mathbf x$ through economic dispatch and $\mathbf x$ depends on $\bm \lambda$ through UE, the coupled equilibrium is defined as follows:}
\begin{definition}[GUE for the Coupled System]\label{def:GUE}
		The \emph{(travel pattern, price)} pair $(\mathbf x^\star, \bm \lambda^\star)$ constitutes a GUE for the coupled power and transportation system if 
	 \begin{enumerate}
	 	\item[(a)] $\mathbf x^\star$ constitutes a \rev{transportation UE} given $\bm \lambda^\star$, and 
	 	\item [(b)] $\bm \lambda^\star =\bm  \lambda^\star(\mathbf x^\star)$ defined via the economic dispatch optimization~\eqref{eq:economic_dispatch}.
	 \end{enumerate}
\end{definition}

\rev{GUEs can be characterized by a single convex program.
\begin{lemma}[GUE Characterization]\label{lemma:GUE_convex_main}
A pair $(\mathbf x^\star,\bm\lambda^\star)$ is a GUE if and only if there exist $(\mathbf g^\star,\mathbf p^\star)$ such that $(\mathbf x^\star,\mathbf g^\star,\mathbf p^\star)$ solves
\begin{equation}\label{eq:auxiliary_opt_main}
\begin{aligned}
\min_{\mathbf{x},\mathbf{g},\mathbf{p}}\quad
& \frac{1}{2}\mathbf g^\top\mathbf Q\mathbf g+\bm\mu^\top\mathbf g\\
&+\frac{1}{2}(\Alr\mathbf x)^\top\mathrm{diag}(\bm\alpha)\Alr\mathbf x
 +\bm\beta^\top\Alr\mathbf x\\
\mathrm{s.t.}\quad
& \bm\lambda:\mathbf p=\mathbf g-\mathbf d(\mathbf x),\quad
\mathbf 1^\top\mathbf p=0, \quad \mathbf H\mathbf p\leq\bar{\mathbf f},\\
&
\mathbf 1^\top\mathbf x=1,\quad \mathbf x\geq \mathbf 0,
\end{aligned}
\end{equation}
\rev{and $\bm\lambda^\star$ is the optimal dual variable associated with the first constraint in~\eqref{eq:auxiliary_opt_main}.}
\end{lemma}
\begin{proof}[\rev{Proof sketch}]
For fixed $\bm\lambda$, transportation UE is equivalent to the KKT conditions of the Beckmann program \cite{beckmann1956studies} with route cost $\mathbf{c}(\mathbf{x},\bm \lambda)$. Combining those KKT conditions with the KKT conditions of~\eqref{eq:economic_dispatch} gives the KKT conditions of~\eqref{eq:auxiliary_opt_main}.
\end{proof}
\rev{Lemma \ref{lemma:GUE_convex_main} implies GUE existence whenever~\eqref{eq:auxiliary_opt_main} is feasible. GUE is unique if the objective is strictly convex on feasible directions; a sufficient condition is $\mathbf Q \succ \mathbf 0$, $\bm\alpha>\mathbf 0$, and full column rank of $\Alr$. \textit{Throughout the paper, we do not assume uniqueness of GUE.}}
}

\subsection{Braess' Paradoxes}\label{sec:social_cost_metrics}
\condense{In the classical transportation BP, increasing some road capacities can increase total social cost. With coupling, we ask whether such paradoxes can occur within and across systems, by examining how road capacities, embedded in $\bm \alpha$, and line capacities $\mathbf{\bar f}$ affect social cost metrics.}

Given $(\bm \alpha, \mathbf{\bar f})$, let $(\mathbf x^\star, \bm \lambda^\star)$ be a GUE. We consider the following social cost metrics that all depend on $(\bm \alpha, \mathbf{\bar f})$:
\begin{enumerate}
	\item[(a)] \emph{Transportation System Social Cost}: $\Phi_\mathrm{T}\triangleq (\mathbf{x}^\star)^\top \mathbf{c}^\mathrm{tr}(\mathbf{x}^\star)$. 
	\item[(b)] \emph{Power System Social Cost}: $\Phi_\mathrm{P}\triangleq \Phi_\mathrm{P}(\mathbf x^\star)$.
	\item[(c)] \emph{Coupled System Social Cost}: $\Phi_\mathrm{C}\triangleq \Phi_\mathrm{T}+ \Phi_\mathrm{P}$.
\end{enumerate}

We can then define BP for our coupled system. 
\begin{definition}[Generalized Braess' Paradox]\label{def:gbp}
	For any $s \in \{\mathrm{T},\, \mathrm{P},\, \mathrm{C}\}$, we say the coupled system exhibits a (generalized) BP (a) if there exists an $\ell_\mathrm{T} \in [m_\mathrm{T}]$ such that $\partial \Phi_s/\partial \alpha_{\ell_\mathrm{T}} <0$, \emph{or} (b) if there exists an $\ell_\mathrm{P} \in [m_\mathrm{P}]$ such that $\partial \Phi_s/\partial \bar f_{\ell_\mathrm{P}} >0$, provided that the derivatives exist.
\end{definition}

\begin{remark}[\rev{Derivative-based BP and finite expansions}]\label{rem:derivative_finite_bp}
\rev{Definition~\ref{def:gbp} provides only local characterizations of Braess' paradoxes for infinitesimal expansion. The effects of finite (large) expansion can be determined by integrating the corresponding derivatives along an expansion path. Let $(\bm \alpha^0, \bar{\mathbf{f}}^0)$ be the nominal capacity vector, and $\bm\alpha^{(\ell_\mathrm{T})}(\eta)$ denote the vector obtained from $\bm\alpha^0$ by replacing only its $\ell_\mathrm{T}$-th component by $\eta$. For a road expansion from $\alpha^0_{\ell_\mathrm{T}}$ to $\alpha^1_{\ell_\mathrm{T}}<\alpha^0_{\ell_\mathrm{T}}$, with all other parameters fixed, we have
\begin{align*}
\Phi_s&\bigl(\bm\alpha^{(\ell_\mathrm{T})}(\alpha^1_{\ell_\mathrm{T}}),\bar{\mathbf f}^0\bigr)-\Phi_s\bigl(\bm\alpha^{(\ell_\mathrm{T})}(\alpha^0_{\ell_\mathrm{T}}),\bar{\mathbf f}^0\bigr)
\\
&=\int_{\alpha^0_{\ell_\mathrm{T}}}^{\alpha^1_{\ell_\mathrm{T}}}
\frac{\partial\Phi_s}{\partial\alpha_{\ell_\mathrm{T}}}\bigl(\bm\alpha^{(\ell_\mathrm{T})}(\eta),\bar{\mathbf f}^0\bigr)\,\mathrm d\eta.
\end{align*}
Similarly, let $\bar{\mathbf f}^{(\ell_\mathrm{P})}(\zeta)$ denote the vector obtained from $\bar{\mathbf f}^0$ by replacing only its $\ell_\mathrm{P}$-th component by $\zeta$. A line expansion from $\bar f^0_{\ell_\mathrm{P}}$ to $\bar f^1_{\ell_\mathrm{P}}>\bar f^0_{\ell_\mathrm{P}}$ gives
\begin{align*}
\Phi_s&\bigl(\bm\alpha^0,\bar{\mathbf f}^{(\ell_\mathrm{P})}(\bar f^1_{\ell_\mathrm{P}})\bigr)-\Phi_s\bigl(\bm\alpha^0,\bar{\mathbf f}^{(\ell_\mathrm{P})}(\bar f^0_{\ell_\mathrm{P}})\bigr)
\\
&=\int_{\bar f^0_{\ell_\mathrm{P}}}^{\bar f^1_{\ell_\mathrm{P}}}
\frac{\partial\Phi_s}{\partial\bar f_{\ell_\mathrm{P}}}\bigl(\bm\alpha^0,\bar{\mathbf f}^{(\ell_\mathrm{P})}(\zeta)\bigr)\,\mathrm d\zeta.
\end{align*}
Here, each partial derivative is taken with respect to the indicated capacity coordinate while all other components of $(\bm\alpha^0,\bar{\mathbf f}^0)$ are held fixed. In practice, the derivatives can be evaluated at one GUE induced by each intermediate parameter value. Thus, a single derivative sign gives the local trend, while the sign of the corresponding integral predicts the net effect of a prescribed finite expansion.}
\end{remark}

\condense{This definition contains six derivative types, from three costs and two capacity classes. Putting aside $\Phi_\mathrm{C}$, which sums the other two costs, gives four BP types: (a) \emph{Type T-T}: road expansion increases $\Phi_\mathrm{T}$, (b) \emph{Type P-P}: line expansion increases $\Phi_\mathrm{P}$, (c) \emph{Type T-P}: road expansion increases $\Phi_\mathrm{P}$, and (d) \emph{Type P-T}: line expansion increases $\Phi_\mathrm{T}$. We also define \emph{Type T-C} and \emph{Type P-C} BPs for $\Phi_\mathrm{C}$.}

\begin{remark}[\rev{Modeling scope and extensions}]\label{rem:model_scope_extensions}
	\rev{The model is intentionally static and stylized so that GUE admits a convex-program characterization and analytically tractable capacity-sensitivity analysis. Affine link costs are standard in EV route/charging equilibrium models \cite{he2014network,alizadeh2016optimal}, and can be replaced by separable nondecreasing costs whenever the Beckmann potential is available \cite{wardrop1952road,beckmann1956studies}. Identical charging demand can be relaxed to finitely many classes by writing $\mathbf d(\mathbf x)$ as a sum of loads contributed by traveler classes. Each route should be read as a joint route/charging plan, since physical routes with multiple feasible chargers are represented by route copies. LMP pass-through pricing is a standard abstraction \cite{schweppe1988spot,he2013integrated,alizadeh2016optimal}. It represents the wholesale-energy component of charging
	prices; retail tariffs may include additional distribution, station, and
	policy charges. Section~\ref{sec:mitigation} considers more general route-specific prices, whose BP derivatives must be recomputed under the chosen price. The linearized ED model is a practical baseline that captures transmission congestion and spatial LMP feedback. AC and detailed retail-price models would require a separate sensitivity analysis. Finally, charging station waiting costs can be included if needed; nonseparable or dynamic queueing models \cite{liu2022mdc} require separate analysis.}
\end{remark}


\section{Occurrence of Braess' Paradoxes}\label{sec:examples}
\condense{We study BP occurrence through simple coupled systems. Section~\ref{sec:2by2-example} presents a 2-Route 2-Bus system with type T-T and T-P BPs. Section~\ref{sec:2by3-example} presents a 2-Route 3-Bus system that displays nearly all BP types. These examples reveal the key mechanisms and motivate the analytical characterizations in Section~\ref{sec:NS_conditions_sec}.} \new{Proofs of this section are deferred to \hyperref[apd:auxiliary_results]{Appendix~C}.}

\subsection{Transportation Expansion Induced Braess' Paradoxes}\label{sec:2by2-example}
\condense{We demonstrate type T-T and T-P BPs, i.e., transportation system expansion induced BP (TBP), using a simple coupled system with non-trivial route choices (Fig.~\ref{fig:2routes}). The transportation system has one origin and two substitutable destinations. The destinations connect to charging stations served by two different buses.}

\begin{figure}[!htbp]
\scriptsize
\usetikzlibrary{arrows.meta, positioning, shapes.geometric}
\tikzset{>=stealth} 
\centering
\begin{tikzpicture}[scale=.3]
\draw[thick, fill =black](-18,2) circle (0.5);
\draw[thick, fill =black](-8.5,4.1) circle (.5);
\draw[thick, fill =black](-8.5,-0.1) circle (.5);

\draw (-14,3) node[above] {$x_1$};
\draw (-14,-0.5) node[above] {$x_2$};
\draw[->](-18,2)--(-9.2,4);
\draw[->](-18,2)--(-9.2,0);

\begin{scope}[shift = {(0, .5)}]
\draw [line width=2pt,color=blue] (-1,7) -- ++(0,-3.0);
\draw [line width=2pt,color=blue] (-1,.5) -- ++(0,-3.0);
\draw [-, color =blue] (-1.1,5.5) -- (1,5.5);
\draw [-, color =blue] (-1.1,-.5) -- (1.03,-.5);
\draw [->, color =blue] (1,5.5) -- (1,2.275);
\draw [-, color =blue] (1,2.274) -- (1,-.5);
\path (1,2.275) node[blue,right] {$f \leq \bar f$};
\draw [blue] (-3.5,6) circle (.8);
\path [blue] (-3.5,6) node {$ g_1$};
\draw [-, blue] (-2.7,6) -- (-1,6);
\draw [-, blue] (-1,4.5) -- ++ (-2,0) -- ++ (0,-1.5);
\draw [fill, blue] (-3,3) -- ++(-0.3,0.3) -- ++(.6,0) -- cycle;
\path [blue] (-4,3) node[below] {$\rho x_1$};

\draw [blue] (-3.5,0) circle (.8);
\path [blue] (-3.5,0) node {$ g_2$};
\draw [-, blue] (-2.7,0) -- (-1,0);
\draw [-, blue] (-1,0-1.5) -- ++ (-2,0) -- ++ (0,-1.5);
\draw [fill, blue] (-3,0-3) -- ++(-0.3,0.3) -- ++(.6,0) -- cycle;
\path [blue] (-3,0-3) node[below] {$\rho x_2$};
\end{scope}

\path (-13.3,-4.2) node {Transportation Network};
\path(-0.8,-4.2) node[blue]{Power Network};
\end{tikzpicture}
\caption{2-Route 2-Bus example. The arrow on the power line indicates the positive flow direction. }
\label{fig:2routes}
\end{figure}

We consider linear travel cost in this setting, with $\bm \beta = \mathbf 0$ and without loss of generality $\alpha_1 > \alpha_2$. We assume \new{$Q_1 = Q_2$ and} $\bm \mu = \mathbf{0}$, and we only constrain the power flow from bus 1 to 2 and not the reverse direction since $\alpha_1 > \alpha_2$.
The \new{transportation UE} condition for an admissible $\mathbf x$, given the LMPs $\bm \lambda$, can be expressed as
\begin{equation}\label{eq:2bus:eq}
	\alpha_1 x_1 + \rho \lambda_1 = \alpha_2 x_2 + \rho \lambda_2. 
\end{equation} 

\subsubsection{Uncongested Power Network}\label{sec:2x2:unc}
\condense{The GUE depends on whether the power network is congested. When it is uncongested, the two routes have identical charging costs, so charging does not affect route choices. Thus, as in classical BP results for traffic networks (e.g., Theorem 1 in \cite{milchtaich2006network}), type T-T BP is not expected. Also, line-capacity changes affect neither power cost nor LMPs, so neither type P-T nor P-P BP occurs.}

\condense{Unconstrained power flow eliminates spatial load imbalance, so power cost is unaffected. We summarize this case as follows.}
\begin{proposition}[2-Route 2-Bus System: Uncongested Case]\label{prop:2routes_uncongested}
	Generalized \mh{BP} does not occur for the 2-Route 2-Bus coupled system when the power network is uncongested. 
\end{proposition} 


\subsubsection{Congested Power Network}\label{sec:2x2:c}
When the power network is congested, we have (a) the charging cost at bus 1 is lower than that at bus 2, (b) the travel cost of route 1 is larger than that at route 2, and (c) $x_1^\star < x_2^\star$.\footnote{\rev{Statements (a), (b), and (c) are compatible because travelers minimize the sum of travel and charging costs, not either term separately. Through analytically computing the GUE, one can show that $\lambda_1^\star<\lambda_2^\star$ and $x_1^\star < x_2^\star$, while $\alpha_1 x_1^\star > \alpha_2 x_2^\star$, i.e., the travel cost of route 1 is higher. The UE equality~\eqref{eq:2bus:eq} balances two effects: the lower charging price on route 1 offsets its higher travel cost, while the lower travel cost on route 2 offsets its higher charging price.}} In fact, we can explicitly compute $\mathbf x^\star$, $\bm \lambda^\star$, and cost metrics, which yields the following results. 
\begin{proposition}[2-Route 2-Bus System: Congested Case]\label{prop:2routes_congested}
	The 2-Route 2-Bus coupled system always exhibits type T-P \mh{BP} when the power network is congested. Furthermore, it also exhibits type T-T \mh{BP} for certain $\alpha_1, \alpha_2, \bar f$, and $\rho$.  
\end{proposition}

%

\rev{For type T-P BP, the congested power line cannot eliminate load imbalance. Increasing a road capacity makes the corresponding route and charging point more attractive; since total flow is fixed to $1$, one road expansion must increase the imbalance. Increasing road 2 capacity (reducing $\alpha_2$) makes the already attractive route/charger more attractive, \emph{increasing spatial load imbalance and power cost at the GUE}. Thus road expansion creates a \emph{negative externality} to the power system. In contrast, expanding road 1 \emph{does not} lead to type T-P BP, since traffic flow attracted to road 1 balances the power loads on the two buses.}

To understand the potential type T-T \mh{BP}, we observe 
\begin{equation}\label{eq:dphi_t_dalpha1}
	\frac{\partial \Phi_\mathrm{T}}{\partial \alpha_1} = (x_1^\star)^2 + 2 \rho(\lambda_2^\star - \lambda_1^\star) \frac{\partial x_1^\star}{\partial \alpha_1}. 
\end{equation}
The first term captures the direct effect on drivers using road 1, while the second captures the traffic-relocation effect $\partial x_1^\star/\partial \alpha_1$ on total transportation cost. With $\lambda_2^\star > \lambda_1^\star$, traffic shifts can marginally worsen equilibrium travel because drivers are compensated by charging-price differences. When this relocation effect dominates, type T-T BP occurs. Fig.~\ref{fig:partial_derivatives_alpha1} shows that for some $\rho$, $\partial \Phi_\mathrm{T}/\partial \alpha_1<0$, so coupling introduces type T-T BP into a transportation system without classical BP.

\begin{figure}[!htbp]
    \centering
\includegraphics[width=.16\textwidth]{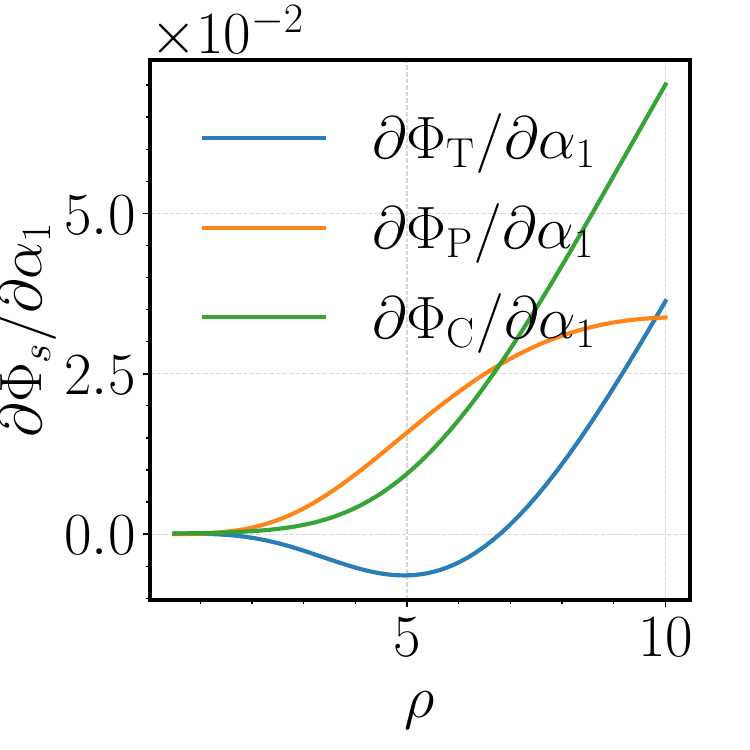}
\caption{Derivatives of $\Phi_s$, $s \in \{\mathrm{T},\, \mathrm{P},\, \mathrm{C}\}$ with respect to $\alpha_1$ when $\alpha_1 = 100, \alpha_2 = 1, \bar f = 0.2$, and $\rho \in [0.5,10]$.}
\label{fig:partial_derivatives_alpha1}
\end{figure}

We close the section by proving that power system expansion induced \mh{BP}s will never occur for the coupled network.

\begin{proposition}[No Type P-T and P-P in the 2-Route 2-Bus System]\label{prop:perfect_power_expansion_2by2}
	Type P-T and P-P BPs never occur for the 2-Route 2-Bus coupled network.
\end{proposition}
 
\condense{The lack of type P-T and P-P BPs stems from the simple power network.}\footnote{In fact, for this power network, increasing the line capacity will always \jqe{(weakly)} reduce the spatial LMP differential across the charging points, which cannot deteriorate transportation or power system performance.} \condense{Next, we consider a power system with a loop where both type P-T and P-P BPs can occur.}

\subsection{Power Expansion Induced Braess' Paradoxes}\label{sec:2by3-example}
\begin{figure}[!htbp]
\scriptsize
\usetikzlibrary{arrows.meta, positioning, shapes.geometric}
\tikzset{>=stealth}
\centering
\begin{tikzpicture}[scale=0.3]

\draw[thick, fill=black] (-18,2) circle (0.5);
\draw[thick, fill=black] (-8.5,4.1) circle (0.5);
\draw[thick, fill=black] (-8.5,-0.1) circle (0.5);

\draw (-14,3) node[above] {$x_1$};
\draw (-14,-0.5) node[above] {$x_2$};
\draw[->] (-18,2) -- (-9.2,4);
\draw[->] (-18,2) -- (-9.2,0);

\begin{scope}[shift={(0,.5)}]

\draw[line width=2pt,color=blue] (-1,7) -- ++(0,-3.0);
\draw[line width=2pt,color=blue] (-1,0.5) -- ++(0,-3.0);
\draw[line width=2pt,color=blue] (6,3.75) -- ++(0,-3.0);

\draw[-,color=blue] (-1.1,5.5) -- ++(1.1,0);
\draw[-,color=blue] (-1.1,0) -- ++(1.1,0);

\draw[->,color=blue] (0,5.5) -- (0,2);
\draw[-,color=blue] (0,2.274) -- (0,0);
\path (0,1.5) node[blue,left] {$f_1 \leq \bar f_1$};

\draw[-,blue] (-1.1,6) -- ++(2,0) -- ++(0,-2.5) -- ++(5.1,0);
\draw[->,blue] (6,3.5) -- ++(-3,0);
\path (3.5,3.5) node[blue,above] {$f_2 \leq \bar f_2$};

\draw[-,blue] (-1.1,-1.5) -- ++(2,0) -- ++(0,3.5) -- ++(5.1,0);
\draw[->,blue] (6,2.0) -- ++(-3,0);
\path (3.5,2.0) node[blue,below] {$f_3 \leq \bar f_3$};

\draw[blue] (-3.5,6) circle (.8);
\path[blue] (-3.5,6) node {$g_1$};
\draw[-,blue] (-2.7,6) -- (-1,6);
\draw[-,blue] (-1,4.5) -- ++(-2,0) -- ++(0,-1.5);
\draw[fill,blue] (-3,3) -- ++(-0.3,0.3) -- ++(.6,0) -- cycle;
\path[blue] (-3,3) node[below] {$\rho x_1$};

\draw[blue] (-3.5,0) circle (.8);
\path[blue] (-3.5,0) node {$g_2$};
\draw[-,blue] (-2.7,0) -- (-1,0);
\draw[-,blue] (-1,-1.5) -- ++(-2,0) -- ++(0,-1.5);
\draw[fill,blue] (-3,-3) -- ++(-0.3,0.3) -- ++(.6,0) -- cycle;
\path[blue] (-3,-3) node[below] {$\rho x_2$};

\draw[blue] (8.5,2.25) circle (.8);
\path[blue] (8.5,2.25) node {$g_3$};
\draw[-,blue] (7.7,2.25) -- ++(-1.7,0);

\end{scope}

\path (-13.3,-4.2) node {Transportation Network};
\path (3,-4.2) node[blue] {Power Network};

\end{tikzpicture}
\caption{2-Route 3-Bus example. The arrows on power lines indicate the positive flow direction.}
\label{fig:3-bus_2-route}
\end{figure}
We demonstrate power system expansion induced Braess' paradoxes \new{(PBP)}, by analyzing a slightly more complex coupled system shown in Fig. \ref{fig:3-bus_2-route}. The power system here is structurally different from that in Section~\ref{sec:2by2-example}. The additional bus and the loop lead to new phenomena when power capacities are expanded. In particular, the LMP differential across bus 1 and 2 can \emph{increase} when certain line capacities are expanded.

The destinations are \jqe{co-located with} charging stations connecting to two different buses (buses 1 and 2) in a 3-Bus power network. We assume $\bm \beta = \mathbf{0}$ and consider only linear travel costs, and the UE condition is identical to \eqref{eq:2bus:eq}. We assume $\bm \mu = \mathbf{0}$, $f_\ell$ denotes the power flow between three buses with positive direction indicated in Fig. \ref{fig:3-bus_2-route}, and $\bar f_\ell$ is flow capacity for the $\ell$-th line (note that here we have deviated from our notation in \eqref{eq:economic_dispatch} slightly to simplify the exposition). We can in fact analytically solve for GUE given any power system congestion pattern; \new{see  \hyperref[apd:auxiliary_results]{Appendix~C} for the analytical expressions}. Without loss of generality, \mh{to show PBP, we assume we change $\bar f_3$, and to show TBP, we assume we change $\alpha_1$.} \mh{We summarize main results of this section in Proposition \ref{thm:3by2_summary}. 

\begin{proposition}[\mh{BP}s in 2-Route 3-Bus Coupled System]\label{thm:3by2_summary}
	The 2-Route 3-Bus coupled system can exhibit any one of type T-T, T-P, P-P, P-T, and P-C \mh{BP}s for some model parameters.
\end{proposition}


}

These results suggest that all except type T-C BP\footnote{It is unclear how to obtain type T-C \mh{BP} in this coupled system. However, one can prove that if the transportation system can itself exhibit transportation \mh{BP} without power system (i.e. type T-T \mh{BP} when $\rho = 0$), then the coupled system can exhibit type T-C \mh{BP}.} can occur. In general, type T-C BP can also occur for coupled systems with more complex transportation systems. In the following, detailed explanations of various types of \mh{BP} are provided.

\subsubsection{Type P-T and P-C BPs} We show type P-T and P-C BPs in Fig.~\ref{fig:phi_t_bar_f}. Type P-T BP occurs because \emph{line expansion can change LMPs in a way that incentivizes traffic-flow relocation that worsens congestion}. We have

\begin{equation}\label{eq:phitf3}
	\frac{\partial \Phi_\mathrm{T}}{\partial \bar f_3} = 2 (\alpha_1 x_1^\star - \alpha_2 x_2^\star) \frac{\partial x_1^\star}{\partial \bar f_3} =  2 \rho (\lambda_2^\star - \lambda_1^\star) \frac{\partial x_1^\star}{\partial \bar f_3}, 
\end{equation}
i.e., the change in the transportation cost is driven by traffic-flow relocation induced by line expansion. The relocation of traffic flow is due to the change in LMPs, as $x_1^\star = (\rho(\lambda_2^\star - \lambda_1^\star) + \alpha_2)/(\alpha_1 + \alpha_2).$
When $\alpha_1 x_1^\star - \alpha_2 x_2^\star = \rho(\lambda_2^\star - \lambda_1^\star)$ and $\partial x_1^\star/\partial \bar f_3$ have the same sign, line expansion nudges traffic to the already busier route and type P-T BP occurs.\footnote{An alternative explanation based on the change of LMP differentials is provided in \hyperref[apd:auxiliary_results]{Appendix~C}.}

In addition, $\Phi_\mathrm{C}$ increases in $\bar f_3$ over some parameter region, so type P-C BP also occurs. Here, the traffic-worsening effect of line expansion dominates the generation-cost reduction from shifting power loads.

\begin{figure}[!htbp]
\centering

\begin{subfigure}[t]{0.159\textwidth}
    \centering
\includegraphics[width=\textwidth]{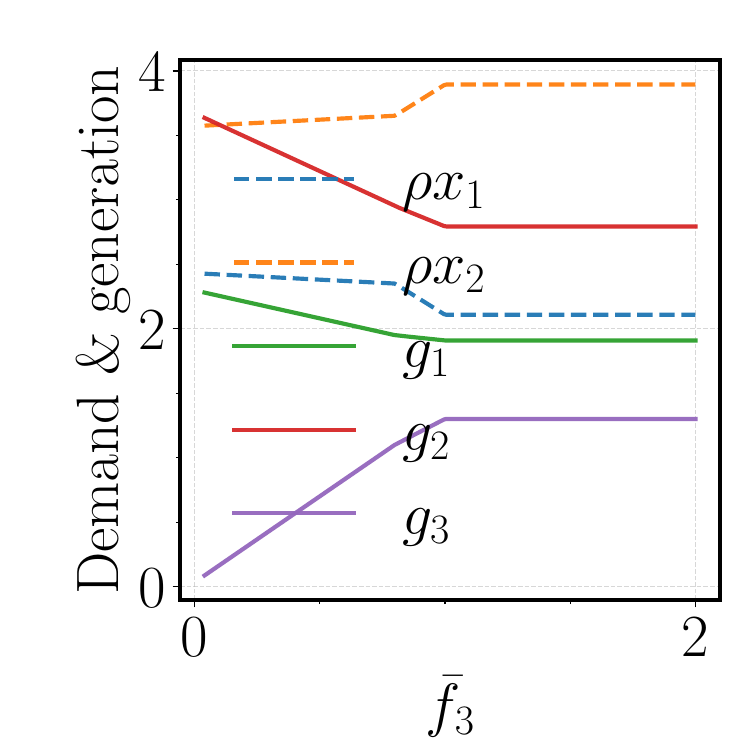}
    \caption{Demand \& generation}
\end{subfigure}
\hfill
\begin{subfigure}[t]{0.158\textwidth}
    \centering
\includegraphics[width=\textwidth]{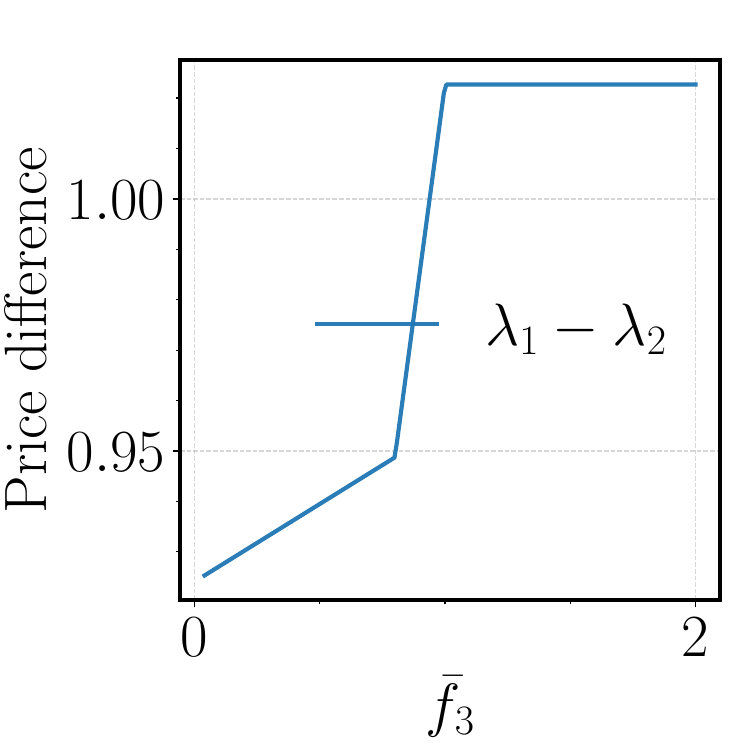}
    \caption{Price difference}
\end{subfigure}
\hfill
\begin{subfigure}[t]{0.158\textwidth}
    \centering
    \includegraphics[width=\textwidth]{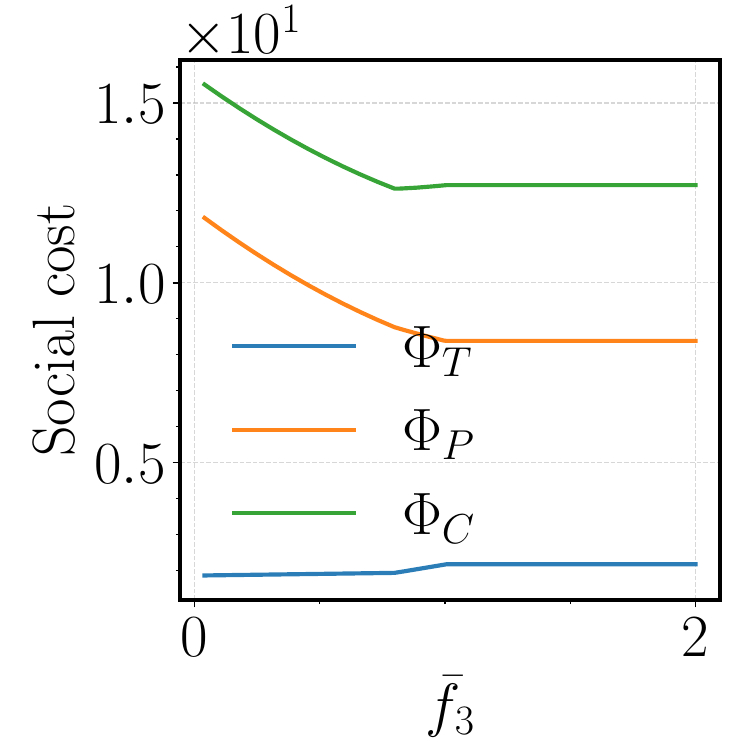}
    \caption{Social cost}
\end{subfigure}


\caption{Model parameters are set as $\bm \alpha := [1,10]^\top$, $\rho := 6$, $\mathbf{Q} := \mathrm{diag}([2,1,1])$, and $\mathbf{\bar f} := [0.1,0.3,\bar f_3]^\top$, where $\bar f_3 \in [0.04,2.0]$.}
\label{fig:phi_t_bar_f}
\end{figure}

\subsubsection{Type P-P BP} We use a different setting to demonstrate type P-P BP; parameters and GUE variation with $\bar f_3$ appear in Fig.~\ref{fig:phi_p_bar_f}. This example is special because bus 1 generates for free. Type P-P BP occurs because \emph{line-expansion-induced LMP changes can trigger load relocation unfavorable for the grid.}

With coupling, increasing $\bar{f}_3$ changes LMPs and relocates load so that $g_3^\star$ increases while $g_2^\star$ stays nearly unchanged, increasing $\Phi_\mathrm{P}$. Fig.~\ref{fig:phi_p_bar_f}-(b) shows that the congestion pattern switches at a critical $\bar f_3$ value; before the switch, type P-P BP occurs. Under this pattern, $\partial \Phi_\mathrm{P}/\partial \bar f_3 = 1.7251 \bar f_3-0.0525 > 0$ for all $\bar f_3$ inducing it.


\begin{figure}[!htbp]
\centering

\begin{subfigure}[t]{0.159\textwidth}
    \centering
\includegraphics[width=\textwidth]{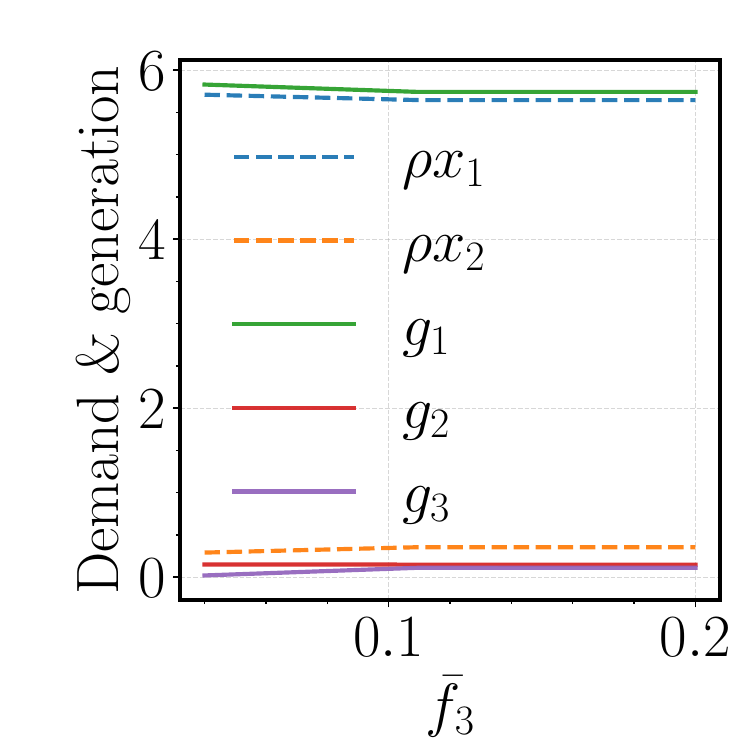}
    \caption{Demand \& generation}
\end{subfigure}
\hfill
\begin{subfigure}[t]{0.158\textwidth}
    \centering
\includegraphics[width=\textwidth]{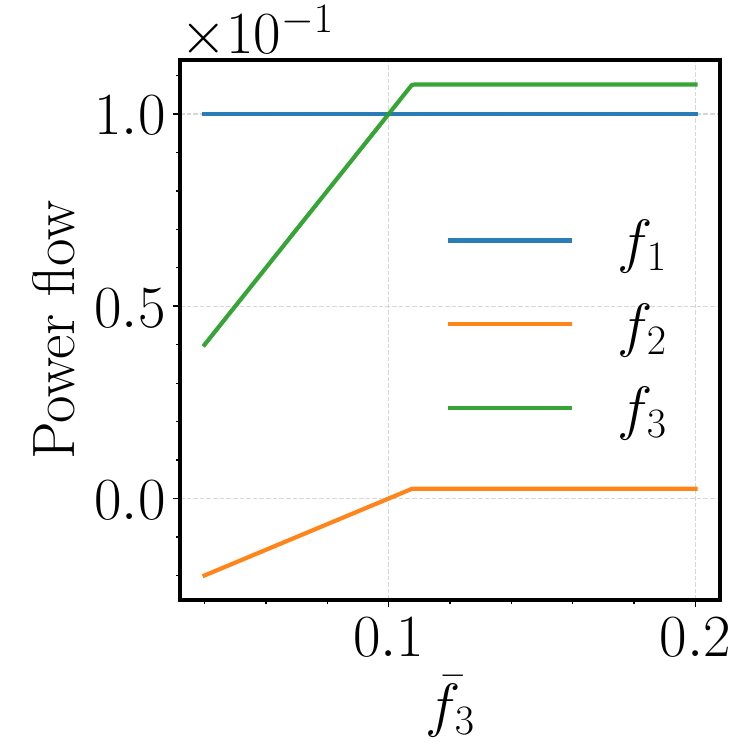}
    \caption{Power flow}
\end{subfigure}
\hfill
\begin{subfigure}[t]{0.158\textwidth}
    \centering
    \includegraphics[width=\textwidth]{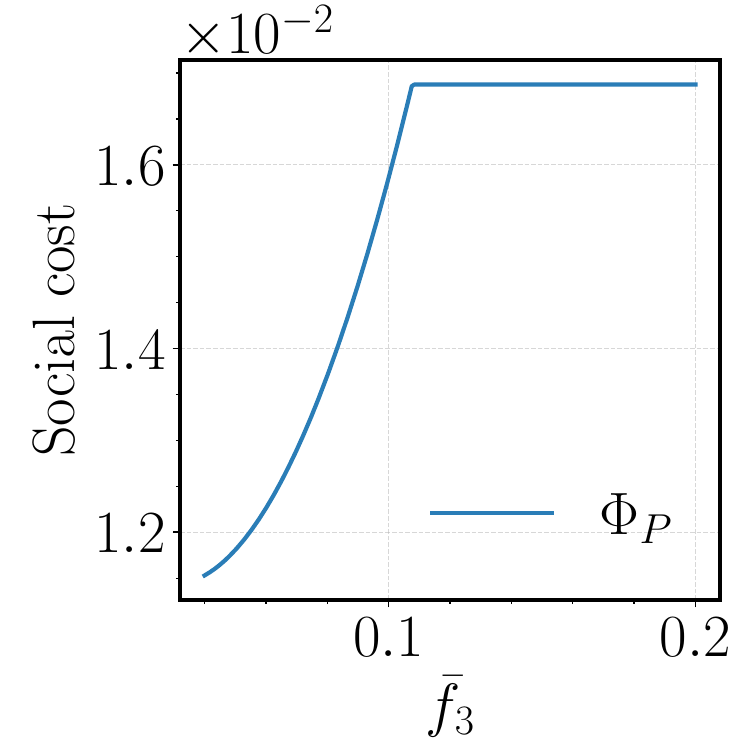}
    \caption{Social cost}
\end{subfigure}

\caption{Model parameters are set as $\bm \alpha := [1,1]^\top, \rho := 6, \mathbf{Q} := \mathrm{diag}([0,1,1])$, and $\mathbf{\bar f} := [0.1,0.3,\bar f_3]^\top$, where $\bar f_3 \in [0.04,0.2]$.}
\label{fig:phi_p_bar_f}
\end{figure}



\subsubsection{Type T-T and T-P BPs}
\mh{Not surprisingly, type T-T and type T-P \mh{BP}s occur in this coupled system. The mechanisms behind stay the same as those in Section~\ref{sec:2by2-example}. We visualize them in Fig.~12 and Fig.~13 in \hyperref[apd:auxiliary_results]{Appendix~C}.
}

%

\section{Necessary and Sufficient Conditions}\label{sec:NS_conditions_sec}
The previous sections show that \textit{all types of BPs can occur} in simple coupled systems. We now study \textit{general coupled systems}, deriving necessary and sufficient BP conditions and tools for screening real-world systems. We proceed from the uncongested case in Section~\ref{subsec:uncongested} to the \emph{radial} network with arbitrary congestion in Section~\ref{subsec:arbitrary_pattern}, and then briefly discuss in Section~\ref{sec:detection_general_networks} how to screen BP for general systems where the power network is not radial.

Throughout this section, we assume $\bm \alpha > \mathbf{0}$ and $\mathbf{Q} > \mathbf{0}$ to avoid degeneracies. \mh{We denote by $\mathcal{R}^\mathrm{act} \subseteq [\nr]$ the set of \textit{active routes} at GUE $(\mathbf{x}^\star, \bm\lambda^\star)$, i.e. routes $r$ with $x_r^\star > 0$, and let $\mh{R}\triangleq|\mathcal{R}^\mathrm{act}|$ be the number of active routes.} Since all entries of $\mathbf{x}^\star$ corresponding to inactive routes are $0$'s, we abuse \mh{$\mathbf{x}^\star$} to denote the vector containing only nonzero entries \mh{of $\mathbf{x}^\star$.}

Lemma \ref{lemma:GUE_convex_main} proved that the non-atomic game with players $\mathcal{J}$ and cost \eqref{eq:t:cost} is a potential game, and GUEs of a given system are equivalent to the optimal solutions of a convex program parameterized by $(\bm \alpha, \bar{\mathbf{f}})$. 
As social cost metrics $\Phi_s, s \in \{\mathrm{T},\mathrm{P},\mathrm{C}\}$ depend on GUE, perturbing $(\bm \alpha, \mathbf{\bar f})$ affects $\Phi_s$. 
\new{Given the definition of BPs, we only consider infinitesimal perturbations of $(\bm \alpha, \bar{\mathbf{f}})$ such that the constraint-binding pattern of the convex program remains unchanged.\footnote{ \rev{Although GUEs may be non-unique, one can show that the social cost
$\Phi_s$, $s\in\{\mathrm T,\mathrm P,\mathrm C\}$, is identical across GUEs at the same $(\bm \alpha, \bar{\mathbf{f}})$. Therefore, $\Phi_s$ is a well-defined function of $(\bm\alpha,\bar{\mathbf f})$ and its derivatives with respect to capacity parameters are independent of GUE selection. Thus, it suffices to pick \emph{one} GUE $(\mathbf{x}^\star, \bm \lambda^\star)$, and compute derivatives by applying the chain rule on the chain of dependence $\Phi_s \leftarrow (\mathbf{x}^\star, \bm \lambda^\star) \leftarrow (\bm \alpha, \bar{\mathbf{f}})$.}}}


\subsection{Uncongested Power Network Case}\label{subsec:uncongested}
We first characterize \mh{BPs} for systems with \textit{uncongested} power network at GUE. For such systems, \mh{LMPs} are identical. The following theorem provides characterizations of \mh{BPs}.
\begin{theorem}[\mh{Necessary and Sufficient Conditions of BP, Uncongested Power Network Case}]\label{thm:uncongested}
	For a system with uncongested power network at GUE, PBP  and type T-P BP never occur. Moreover, type T-T BP occurs when perturbing link $\ell_\mathrm{T}$ if and only if classical BP occurs when perturbing $\ell_\mathrm{T}$, in the transportation network consisting of the $\mh{R}$ active routes with \mh{a unit flow}.
\end{theorem}
PBP \mh{never occurs} since every line is uncongested. When $\mh{R = 1}$, the unique route always carries all \mh{traffic flow}, precluding TBP. When $R\geq 2$, type T-P BP never occurs since even though charging load may relocate, the uncongested power system can recover the original generation profile. Since \mh{LMPs} are identical, travelers' decisions only depend on travel costs, and type T-T BP reduces to the classical BP. 

\subsection{\mh{Radial Power Network Case}}\label{subsec:arbitrary_pattern}
\condense{With power-network congestion, coupling creates new BP types. We first consider radial (tree) power networks, where congestion patterns have a clear relationship with spatial LMPs. This enables the network reduction used below to derive necessary and sufficient conditions.}

\subsubsection{Network Reduction} \jqe{A coupled system with radial power network can be reduced into a simpler coupled system as follows (see Fig.~\ref{fig:aggregation}): For the power network, 
 the line congestion patterns partition it into \emph{subnetworks},
within which buses are connected by uncongested lines and share the same LMP. Subnetworks induce corresponding \emph{route bundles} in the transportation network,} i.e., active routes whose associated chargers are connected to buses in the same subnetwork. Consequently, routes in the same route bundle see the same charging price and can be treated collectively. 

Each subnetwork admits an aggregate quadratic generation cost representation, and each route bundle admits an aggregate travel cost representation.\footnote{The detailed derivations are deferred to \hyperref[sec:supp_materials]{Appendix~D}.} Hence the original system reduces to an aggregated system whose subnetworks become buses and route bundles become routes.

\condense{The reduction leaves two tasks: deriving BP conditions for the aggregated system, which resembles a \emph{fully congested radial network}, and translating those insights back to the original system. We address them next.}

\begin{figure}[h]
\centering
\resizebox{!}{0.25\textheight}{ 
\begin{tikzpicture}[
    node distance=1cm,
    every node/.style={font=\small},
    road/.style={draw, thick, -{Latex[length=1.5mm]}},
    powerline/.style={draw, thick},
    congested/.style={draw, thick, dashed, red!70!black},
    couple/.style={draw, dashed, gray, -{Latex[length=1.5mm]}},
    maparrow/.style={draw, line width=4pt, gray!20, -{Latex[length=3mm]}},
    tnode/.style={circle, draw=blue!80!black, thick, fill=blue!5, minimum size=5.5mm, inner sep=0pt},
    cnode/.style={rectangle, rounded corners=2pt, draw=green!60!black, thick, fill=green!5, minimum size=5.5mm},
    pnode/.style={circle, draw=red!80!black, thick, fill=red!5, minimum size=5.5mm, inner sep=0pt},
    gen/.style={rectangle, draw=orange!80!black, thick, fill=orange!10, minimum size=5mm},
    bundlebox/.style={draw=blue!50, rounded corners=4pt, dashed, fill=blue!5, fill opacity=0.3},
    subnetbox/.style={draw=red!50, rounded corners=4pt, dashed, fill=red!5, fill opacity=0.3},
    >=Latex
]

\begin{scope}[shift={(0,0)}]
    \node[font=\small\bfseries, anchor=west] at (-4, 4.2) {I. Original Coupled System};

    \node[anchor=west] at (-3.2, -1.3) {Transportation Network};
    \node[anchor=west] at (4.0, -1.3) {Power Network};

    \node[tnode] (O)  at (-4, 1.5) {$o$};
    \node[tnode] (1)  at (-2.5, 3.0) {$1$};
    \node[tnode] (4)  at (-2.5, 0.0) {$4$};
    \node[tnode] (2)  at (-1, 3.0) {$2$};
    \node[tnode] (5)  at (-1, 0.0) {$5$};
    \node[tnode] (3)  at (-1.75, 1.5) {$3$};
    \node[tnode] (D)  at (0.5, 1.5) {$d$};

    \draw[road] (O) -- (1); \draw[road] (1) -- (2); \draw[road] (2) -- (D);
    \draw[road] (O) -- (4); \draw[road] (4) -- (5); \draw[road] (5) -- (D);
    \draw[road] (O) -- (3) -- (D);

    \node[bundlebox, minimum width=3.2cm, minimum height=0.8cm, label={[blue]above: Route bundle $1$ ($\mathcal{P}^1$)}] at (-1.75, 3.0) {};
    \node[bundlebox, minimum width=3.2cm, minimum height=2.3cm, label={[blue]below: Route bundle 2 ($\mathcal{P}^2$)}] at (-1.75, 0.75) {};

    \node[pnode] (B1) at (3.5, 2.5) {$b_1$};
    \node[pnode] (B2) at (4.5, 2.5) {$b_2$};
    \node[pnode] (B3) at (5.5, 2.5) {$b_3$};

    \node[pnode] (B4) at (3.5, 0.5) {$b_4$};
    \node[pnode] (B5) at (4.5, 0.5) {$b_5$};
    \node[pnode] (B6) at (5.5, 0.5) {$b_6$};
    \node[pnode] (B7) at (7.5, 1.5) {$b_7$};

    \draw[powerline] (B1) -- (B2); \draw[powerline] (B2) -- (B3);
    \draw[powerline] (B4) -- (B5); \draw[powerline] (B5) -- (B6);
    \draw[congested] (B3) -- node[sloped,above]{\scriptsize congested} (B7); \draw[congested] (B6) -- node[sloped,above]{\scriptsize congested} (B7);

    \node[subnetbox, minimum width=3cm, minimum height=1.2cm, label={[red!70!black]above:Subnetwork $1$ ($\hat{b}_1$)}] at (4.5, 2.5) {};
    \node[subnetbox, minimum width=3cm, minimum height=1.2cm, label={[red!70!black]below:Subnetwork $2$ ($\hat{b}_2$)}] at (4.5, 0.5) {};
    \node[subnetbox, minimum width=1cm, minimum height=1.2cm, label={[red!70!black]below:Subnetwork $3$ ($\hat{b}_3$)}] at (7.5, 1.5) {};
    
    \draw[couple, <->]  (2) to node[above] {\scriptsize{charger}} (B1);
    \draw[couple, <->, bend left=25] (3) to node[above] {\scriptsize{charger}} (B5);
    \draw[couple, <->] (5) to node[above] {\scriptsize{charger}} (B4);
\end{scope}

\draw[maparrow] (1.5, -0.8) -- node[right, font=\footnotesize\bfseries, black] {Aggregation} (1.5, -1.8);

\begin{scope}[shift={(0,-5.2)}]
    \node[font=\small\bfseries, anchor=west] at (-4, 3.2) {II. Aggregated System};

    \node[tnode] (Oh) at (-4, 1.5) {$o$};
    \node[cnode, minimum width=1.5cm] (P1) at (-1.75, 2.5) {$\mathcal{P}^1$};
    \node[cnode, minimum width=1.5cm] (P2) at (-1.75, 0.5) {$\mathcal{P}^2$};
    \node[tnode] (Dh) at (0.5, 1.5) {$d$};

    \draw[road] (Oh) -- (P1); \draw[road] (P1) -- (Dh);
    \draw[road] (Oh) -- (P2); \draw[road] (P2) -- (Dh);

    \node[pnode] (H1)  at (4.5, 2.5) {$\hat{b}_1$};
    \node[pnode] (H2)  at (4.5, 0.5) {$\hat{b}_2$};
    \node[pnode] (H3)  at (6, 1.5) {$\hat{b}_3$};

    \draw[congested] (H3) -- node[above, font=\tiny, black] {$\hat{f}_1$} (H1);
    \draw[congested] (H3) -- node[below, font=\tiny, black] {$\hat{f}_2$} (H2);

    \draw[couple, <->] (P1.east) to (H1.west);
    \draw[couple, <->] (P2.east) to (H2.west);
\end{scope}
\end{tikzpicture}
}
\caption{Illustration of subnetworks, route bundles, and the aggregation procedure. \new{The transportation network in the original coupled system contains three routes $r_1: o \to 1 \to 2 \to d$, $r_2: o \to 3 \to d$, and $r_3: o \to 4 \to 5 \to d$, with corresponding chargers connecting to buses $b_1, b_5$, and $b_4$.} Buses $b_1,b_2,b_3$ ($b_4,b_5,b_6$) with lines connecting them form subnetwork $1$ ($2$). Bus $b_7$ alone forms subnetwork $3$. \new{Route $r_1$ forms route bundle $\mathcal{P}^1$, and routes $r_2$ and $r_3$ form route bundle $\mathcal{P}^2$.} }
\label{fig:aggregation}
\end{figure}




\subsubsection{Fully Congested Radial Power Network Case}\label{subsec:completely_congested}
\jqe{We start with the fully congested radial power network case because it captures the essential BP mechanisms in the aggregated system without the additional notation introduced by network reduction. It is also the simplest aggregated setting, in which each subnetwork consists of one bus.
More specifically,} we characterize BPs for systems satisfying the following at GUE: 
\begin{enumerate}
	\item[\textbf{A}1] {T}he \mh{radial} power network is \textit{fully congested} (i.e., all \mh{power} lines are congested);
	\item[\textbf{A}2] {A}ny two \textit{active routes} share neither \mh{links nor buses}.
\end{enumerate}  

{While \textbf{A}1 is \jqe{a temporary assumption to be relaxed in Section~\ref{subsec:general_networks}).}} \textbf{A}2 simplifies the analysis of TBP: perturbing $\mh{\ell_\mathrm{T}}$ directly affects at most one route, with other routes affected only indirectly through traffic flow \mh{relocation}.\footnote{\rev{\textbf{A}2 is not intended as a universal modeling restriction on transportation networks. It is a simplifying assumption for the statement that if active routes do not share links or buses, a road-capacity perturbation can be attributed to a single active route, and the direct and relocation effects have closed forms. In the general arbitrary congestion pattern case (Section \ref{subsec:general_networks}), \textbf{A}2 is relaxed to allow link and bus sharing. See the discussion in Section \ref{subsec:general_networks} for details.}} By \textbf{A}2, \mh{each active route $r \in \mathcal{R}^\mathrm{act}$ is associated with a unique bus $i_r \in [\np]$. \new{We note that even if active routes share buses, we can still obtain BP characterizations but with more involved notations.} We denote by $\mathcal{I}^\mathrm{act} \subseteq [\np]$ the set of buses that are associated with an active route. Moreover, for any bus $i \in \mathcal{I}^\mathrm{act}$, we denote $r_i$ as the (unique) active route $r \in \mathcal{R}^\mathrm{act}$ associated with it. Since no two active routes share links,} when we say \textit{perturbing a route $r$}, we mean perturbing \jqe{the capacity/cost parameter of }a link contained (only) by route $r$.


\new{The following theorem provides complete characterizations of BPs for coupled systems satisfying \textbf{A}1 and \textbf{A}2 at GUE.}
\begin{theorem}[Necessary and Sufficient Conditions of BP, Fully Congested Radial Power Network Case]\label{thm:fully_congested}
For a system satisfying \textbf{A}1 and \textbf{A}2 at GUE $(\mathbf{x}^\star, \bm \lambda^\star)$, the following hold\mh{:}
\begin{enumerate}
	\item[(a)] If $\mh{R} = 1$, neither TBP nor PBP occurs.

	\item[(b)] If $\mh{R} \geq 2$,
	\begin{enumerate}
\item[(b1)] type T-T BP occurs \mh{when perturbing a route $r \in \mathcal{R}^\mathrm{act}$ if and only if}
\begin{align}\label{eq:tt_fully_congested}
x^\star_{\mh{r}} < \psi_{\mh{r}} \triangleq \omega_{\mh{r}} \sum_{\mh{r'}=1}^R \tilde{\omega}_{\mh{r'}} (2\hat{\alpha}_{\mh{r}} x_{\mh{r}}^\star + \hat{\beta}_{\mh{r}} -2\hat{\alpha}_{\mh{r'}} x_{\mh{r'}}^\star + \hat{\beta}_{\mh{r'}}),
\end{align}
where $\hat{\alpha}_{\mh{r}} \triangleq \sum_{\ell:\Alr_{\ell,\mh{r}} = 1} \alpha_\ell$, $\hat{\beta}_{\mh{r}} \triangleq \sum_{\ell:\Alr_{\ell,\mh{r}} = 1} \beta_\ell$, $\omega_{\mh{r}} \triangleq 1/(\hat{\alpha}_{\mh{r}}+\rho^2 Q_{\mh{i_r}})$, and $\tilde{\omega}_{\mh{r}} \triangleq \omega_{\mh{r}} / \sum_{\mh{r'}=1}^R \omega_{\mh{r'}}$;

\item[(b2)] type T-P BP occurs when perturbing \mh{a} route $\mh{r} \in \mathcal{R}^\mathrm{act}$ if and only if
\begin{align}\label{eq:tp_fully_congested}
\sum_{\mh{i_{r'} \in \mathcal{I}^\mathrm{act}}} \tilde{\omega}_{\mh{r'}} (\lambda_{\mh{i_r}}^\star - \lambda_{\mh{i_{r'}}}^\star) > 0;
\end{align}

\item[(b3)] type P-T BP occurs when perturbing $\ell_\mathrm{P} = (i,i')$ with power flow $i \to i'$ if and only if
\begin{align}\label{eq:pt_fully_congested}
Q_{i} \psi_{r_i}\mh{\mathbbm{1}\{i \in \mathcal{I}^\mathrm{act}\}} - Q_{i'} \psi_{r_{i'}} \mh{\mathbbm{1}\{i' \in \mathcal{I}^\mathrm{act}\}} < 0;
\end{align}  

\item[(b4)] type P-P BP occurs when perturbing $\ell_\mathrm{P} = (i,i')$ with power flow $i \to i'$ if and only if
\begin{align}\label{eq:pp_fully_congested}
\varsigma_{i}\mh{\mathbbm{1}\{i \in \mathcal{I}^\mathrm{act}\}} - \varsigma_{i'}\mh{\mathbbm{1}\{i' \in \mathcal{I}^\mathrm{act}\}} > 0,
\end{align}
where $\varsigma_i \triangleq \lambda_i^\star - \rho^2 Q_i \omega_{\mh{r_i}} \sum_{\mh{j\in \mathcal{I}^\mathrm{act}}} \tilde{\omega}_{j} (\lambda_i^\star - \lambda_{j}^\star)$.
\end{enumerate}
\end{enumerate}
\end{theorem}
\begin{proof}[\rev{Proof sketch}]
    \rev{Under \textbf{A}1--\textbf{A}2, each active route $r$ is associated with a unique bus $i_r$. The GUE conditions form a linear system involving $(\bm \alpha, \bar{\mathbf{f}})$ and $(\mathbf{x}^\star, \bm \lambda^\star)$. Differentiating this system gives the sensitivities of $(\mathbf{x}^\star, \bm \lambda^\star)$ with respect to $(\bm \alpha, \bar{\mathbf{f}})$. Applying the chain rule to $\Phi_s$ yields \eqref{eq:tt_fully_congested}--\eqref{eq:pp_fully_congested}.}
\end{proof}

When $\mh{R} = 1$, the non-occurrence of TBP adopts the same interpretation as that of Theorem~\ref{thm:uncongested}. On the power side, increasing line capacity cannot induce \mh{traffic flow} relocation {as there is only one active route}, preventing PBP.

\new{Theorem \ref{thm:fully_congested}-(b1) characterizes type T-T BP.} \mh{Inequality} \eqref{eq:tt_fully_congested} has a form similar to \eqref{eq:dphi_t_dalpha1} in Section~\ref{sec:2x2:c}\new{, so the same interpretation applies}. The LHS \new{(up to algebraic manipulation)} \mh{is related to} \textit{how capacity change of \mh{route $r$} impacts the travel cost of travelers choosing route $r$}, which can be understood as a direct BP-preventing \mh{effect}, \mh{since increasing link capacity decreases $\Phi_\mathrm{T}$ by decreasing travel costs of routes using this link, preventing type T-T BP from occurring}{;} the RHS \mh{$\psi_r$ is related to} \textit{how unit \mh{traffic} flow relocation \mh{from route $r$} impacts the total travel cost}, which can be viewed as an indirect BP-promoting effect, \mh{since traffic flow relocation has the effect of increasing $\Phi_\mathrm{T}$ by moving traffic flow into route $r$, which then promotes type T-T BP.} Intuitively, {i}nequality \eqref{eq:tt_fully_congested} compares the two effects. If the BP-promoting effect dominates, type T-T BP occurs.

\new{Theorem \ref{thm:fully_congested}-(b2) characterizes type T-P BP.} One sufficient condition for \eqref{eq:tp_fully_congested} to hold is \mh{LMPs associated with active routes} are not identical. One can then expand the active route with the highest charging price to make it more attractive, then the \mh{traffic flow} relocation induces higher load at the bus with the highest price, which causes type T-P BP.
		
\new{Theorem \ref{thm:fully_congested}-(b3) characterizes type P-T BP.} \mh{Inequality} \eqref{eq:pt_fully_congested} characterizes the effect of \mh{traffic} flow relocation induced by line capacity change on $\Phi_\mathrm{T}$. The same term $\psi_r$ in \eqref{eq:tt_fully_congested} appears, but it now characterizes \textit{how unit \mh{traffic} flow relocation from route $r$ induced by LMP change at bus $i_r$ impacts the total travel cost}. One can argue when line $\ell_\mathrm{P} = (i,i')$ with power flow $i \to i'$ is perturbed, $\psi_{r_i} <(>) 0$ corresponds to a BP-promoting(preventing) effect, and $\psi_{r_{i'}} <(>)0$ corresponds to a BP-preventing(promoting) effect. Therefore, the LHS of \eqref{eq:pt_fully_congested} is the aggregate BP-promoting(preventing) effect. If strictly negative, it corresponds to an aggregate BP-promoting effect, and thus type P-T BP occurs.

	
\new{Theorem \ref{thm:fully_congested}-(b4) characterizes type P-P BP.} The $\varsigma_{i}$ term in \eqref{eq:pp_fully_congested} characterizes \textit{how unit load relocation from bus $i$ induced by LMP change impacts the total generation cost}. Similarly, the LHS of \eqref{eq:pp_fully_congested} can be viewed as the aggregate BP-promoting(preventing) effect. If strictly positive, it is BP-promoting, and thus type P-P BP occurs. 

Theorem \ref{thm:fully_congested} also implies several relations among BP types.
\begin{corollary}[\mh{Relations of BPs}]\label{cor:BP_relations}
	Under \mh{\textbf{A}1 and \textbf{A}2}, the following statements hold\mh{:}
	\begin{enumerate}
		\item[(a)] If type P-T BP occurs when perturbing $\ell_\mathrm{P}\triangleq(i_r,i),i \notin \mathcal{I}^\mathrm{act}$, \mh{with power flow $i_r \to i$}, then type T-T BP \mh{does} not occur when perturbing route $r$;
		\item[(b)] If type P-T BP does not occur when perturbing $\ell_\mathrm{P}\triangleq(i,i_r), i \notin \mathcal{I}^\mathrm{act}$, \mh{with power flow $i \to i_r$}, then type T-T BP \mh{does} not occur when perturbing route $r$;
		\item[(c)] If type P-P BP does not occur when perturbing $\ell_\mathrm{P}\triangleq(i_r,i),i \notin \mathcal{I}^\mathrm{act}$, \mh{with power flow $i_r \to i$}, then type T-P BP occurs when route $r$ is perturbed;
		\item[(d)] If type P-P BP occurs when perturbing $\ell_\mathrm{P}\triangleq(i,i_r), i \notin \mathcal{I}^\mathrm{act}$, \mh{with power flow $i \to i_r$}, then type T-P BP occurs when route $r$ is perturbed;
		\item[(e)] If a transmission line $\ell_\mathrm{P} \triangleq (i,i'), i,i' \notin \mathcal{I}^\mathrm{act}$ is perturbed, then \mh{PBP does not occur}.
	\end{enumerate}
\end{corollary}
The proof follows directly from \eqref{eq:tt_fully_congested}--\eqref{eq:pp_fully_congested}; additional discussion is deferred to \hyperref[sec:supp_materials]{Appendix~D}. Fig.~\ref{fig:BP_relations} visualizes these implications.

\begin{figure}[!htbp]
\scriptsize
\centering
\begin{tikzpicture}[
    node distance=2cm,
    main/.style={circle, draw, thick, minimum size=0.7cm},
    intermediate/.style={circle, draw, thick, minimum size=0.7cm},
    bpbox/.style={
        draw,
        rectangle,
        rounded corners,
        thick,
        align=center,
        fill opacity=0.3,
        text opacity=1
    }
]
    \node (tt) [bpbox, fill=blue!20] at (0,2) {T-T BP \\ occurs};
    \node (tp) [bpbox, fill=blue!20] at (2,2) {T-P BP\\ occurs};
        
    \node (pt) [bpbox, fill=blue!20] at (0,0) {P-T BP\\ occurs};
    \node (pp) [bpbox, fill=blue!20] at (2,0) {P-P BP\\ occurs};
        
    \node (tt_non) [bpbox, fill=magenta!20] at (4,2) {T-T BP \\ not occur};
    \node (tp_non) [bpbox, fill=magenta!20] at (6,2) {T-P BP\\ not occur};
        
    \node (pt_non) [bpbox, fill=magenta!20] at (4,0) {P-T BP\\ not occur};
    \node (pp_non) [bpbox, fill=magenta!20] at (6,0) {P-P BP\\ not occur};

    \draw[->,>=latex,thick,ForestGreen!80!Black,transform canvas={yshift=0pt}]
        (pt) -- node[pos=0.2,sloped,above]{\scriptsize C1a} (tt_non);
    \draw[->,>=latex,thick,ForestGreen!80!Black,transform canvas={yshift=0pt}]
        (pp_non) -- node[pos=0.2,sloped,above]{\scriptsize C1c} (tp);
    \draw[->,>=latex,thick,ForestGreen!80!Black,transform canvas={yshift=0pt}]
        (tt) -- node[pos=0.2,sloped,above]{\scriptsize C1a-C} (pt_non);
    \draw[->,>=latex,thick,ForestGreen!80!Black,transform canvas={yshift=0pt}]
        (tp_non) -- node[pos=0.2,sloped,above]{\scriptsize C1c-C} (pp);
    
    \draw[->,>=latex,thick,Magenta!80!Black,transform canvas={yshift=0pt}]
        (pt_non) -- node[pos=0.2,right]{\scriptsize C1b} (tt_non);
    \draw[->,>=latex,thick,Magenta!80!Black,transform canvas={yshift=0pt}]
        (pp) -- node[pos=0.2,left]{\scriptsize C1d} (tp);
    \draw[->,>=latex,thick,Magenta!80!Black,transform canvas={yshift=0pt}]
        (tt) -- node[left]{\scriptsize C1b-C} (pt);
    \draw[->,>=latex,thick,Magenta!80!Black,transform canvas={yshift=0pt}]
        (tp_non) -- node[right]{\scriptsize C1d-C} (pp_non);
    
    \node at (3.2,-0.7) {%
    \protect\tikz{\protect\draw[->, ForestGreen!80!Black, thick, transform canvas={yshift=3pt}] (0,0) -- (0.4,0);}%
    Either route $r$ or line $(i_r,i)$ with $i \notin \mathcal{I}^\mathrm{act}$, {and power flow $i_r \to i$,} is perturbed
    };
    \node at (3.2,-1) {%
    \protect\tikz{\protect\draw[->, Magenta!80!Black, thick, transform canvas={yshift=3pt}] (0,0) -- (0.4,0);}%
    Either route $r$ or line $(i,i_r)$ with $i \notin \mathcal{I}^\mathrm{act}$, {and power flow $i \to i_r$,} is perturbed
    };
\end{tikzpicture}
\caption{Relations of occurrences and non-occurrences of BPs when the network at GUE is fully congested and active routes are non-overlapping. \mh{$X$\protect\tikz{\protect\draw[->, red, Black, thick, transform canvas={yshift=3pt}] (0,0) -- (0.4,0);} $Y$} stands for \mh{$X$ implies $Y$}, the text C1x on each arrow represents Corollary \ref{cor:BP_relations}-(x), and postfix -C represents the contrapositive.}
\label{fig:BP_relations}
\end{figure}

\subsubsection{\new{Beyond the Fully Congested Case}}\label{subsec:general_networks}
\new{As discussed before, even when the original power network is neither uncongested nor fully congested, the aggregated network is fully congested \rev{(\textbf{A}1 no longer needed)}. Thus Section~\ref{subsec:completely_congested} applies after replacing routes/buses by route bundles/subnetworks. However,} aggregation induces transportation-side BP in the aggregated system, termed \emph{type T-T (T-P) aggregated TBP} \new{(type T-T (T-P) ATBP)}, since perturbing a link affects route bundles containing it. \rev{The route-bundle relaxation of \textbf{A}2, denoted \textbf{A}2' and detailed in \hyperref[sec:supp_materials]{Appendix~D}, only requires that no two route bundles share links; routes within a bundle may still share links. Under \textbf{A}2', perturbing a link directly affects at most one route bundle, while other bundles are affected only through traffic flow relocation. This preserves the direct-versus-relocation interpretation of Theorem~\ref{thm:fully_congested}, with routes and buses replaced by route bundles and subnetworks.} 

\rev{PBP mechanisms are the same as Theorem~\ref{thm:fully_congested}-(b3) and (b4) after replacing quantities by their aggregated versions. TBP characterizations are more involved than Theorem~\ref{thm:fully_congested}-(b1) and (b2), because road expansion affects route bundles through one additional dependence layer. They appear as a theorem of alternative; see Theorem 10 in \hyperref[sec:supp_materials]{Appendix~D}.}



\subsection{General Networks}\label{sec:detection_general_networks}
\jqe{For non-radial power networks, analytical necessary and sufficient conditions are challenging, but numerical BP screening can be done efficiently by derivative computation:}
1) GUEs are optimal solutions of a convex program assuming that $(\bm \alpha, \bar{\mathbf{f}}) \in \mathbb{R}_{++}^{\mt} \times \mathbb{R}_{++}^{\mp}$ (Lemma \ref{lemma:GUE_convex_main}); 2) the optimal solution is differentiated with respect to parameters, approximately or analytically; and 3) the signs of derivatives are checked to detect BP. We have also shown these derivatives with respect to $\alpha_{\lt}, \lt \in [\mt]$, and $\bar{f}_{\lp}, \lp \in [\mp]$, exist almost everywhere over $(\bm \alpha,\bar{\mathbf{f}}) \in \mathbb{R}_{++}^{\mt} \times \mathbb{R}_{++}^{m_\mathrm{p}}$, so the approach applies to essentially all practical coupled systems.

\section{Mitigation}\label{sec:mitigation}

If a coupled system exhibits BPs, it is useful to mitigate them. We focus on \emph{charging pricing}: as charging infrastructure and its business ecosystem grow, adopting new charging pricing rules may face less friction than modifying grid dispatch protocols. The LMP-based pricing used in most of the paper captures long-term spatial electricity cost variation and appears in prior work \citep{he2013integrated,alizadeh2016optimal}, but may not be the exact practical tariff. Thus alternative charging pricing complements LMP-based pricing.

In the case of \mh{LMP-based} pricing, given LMPs $\bm \lambda \in \mathbb R^{\np}$, the route-specific charging cost is $\bm\pi(\bm \lambda) \in \mathbb R^{\nr}$ computed as in \eqref{eq:pi}. We consider an alternative charging-pricing policy that results in a different route-specific charging-cost vector $\bm \Pi \in \mathbb R^{\nr}$. 

The route-specific charging cost $\bm \Pi$ can be \emph{adaptive} or \emph{static}. In the adaptive case, $\bm \Pi$ may depend on the traffic state $\mathbf x$ and price $\bm \lambda$ induced by the policy; in the static case, $\bm \Pi$ is an exogenous constant vector. We first study system-optimal adaptive pricing, then optimize static prices to mitigate BP.

\subsubsection{System-Optimal Pricing} Such pricing policies nudge the GUE so it optimizes one of the social cost metrics. These policies are obtained by leveraging our structural characterization of the GUE (Lemma \ref{lemma:GUE_convex_main}).
Enforcing system-optimal pricing aligns individual travelers' incentives with the system costs, and eliminates certain (but not all) types of \mh{BP}.

\begin{theorem}[System-Optimal Charging Pricing]\label{thm:sysoptPi}
    Given fixed model parameters, 
    there exist pricing policies $\bm \Pi_\mathrm{T}^\star$, $\bm \Pi_\mathrm{P}^\star$, and $\bm \Pi_\mathrm{C}^\star$, 
    that induce GUE achieving the minimal $\Phi_\mathrm{T}$, $\Phi_\mathrm{P}$, and $\Phi_\mathrm{C}$ respectively. Furthermore, the following statements hold:
\begin{enumerate}
	\item[(a)] $\bm \Pi_\mathrm{T}^\star \triangleq (\Alr)^\top \mathrm{diag}(\bm \alpha) \Alr \mathbf{x} \in \mathbb{R}^{\nr}$ eliminates type T-T, P-T, P-P, and P-C \mh{BPs};
	\item[(b)] $\bm \Pi_\mathrm{P}^\star \triangleq\bm \pi (\bm \lambda) - (\Alr)^\top \mathrm{diag}(\bm \alpha) \Alr \mathbf{x} - (\Alr)^\top \bm \beta \in \mathbb{R}^{\nr}$ eliminates type T-T, T-P, T-C, and P-P \mh{BPs};
	\item[(c)] $\bm \Pi_\mathrm{C}^\star \triangleq \bm \Pi_\mathrm{T}^\star(\mathbf x) + \bm \pi (\bm \lambda)$ eliminates type T-C and P-C \mh{BPs},
\end{enumerate}
where $(\mathbf x, \bm \lambda)$ involved is the GUE induced by the respective pricing policies.
\end{theorem}
\begin{proof}[\rev{Proof sketch}]
\rev{The proof follows the same KKT-matching idea as the proof of Lemma~\ref{lemma:GUE_convex_main}. For each pricing rule, write the UE conditions under the new route cost and stack them with the ED KKT conditions when the pricing rule depends on $\bm\lambda$. The stacked conditions coincide with the KKT conditions of a convex program: under $\bm\Pi_\mathrm{T}^\star$ the objective is $\Phi_\mathrm{T}$, under $\bm\Pi_\mathrm{P}^\star$ the objective is $\Phi_\mathrm{P}$, and under $\bm\Pi_\mathrm{C}^\star$ the objective is $\Phi_\mathrm{C}$. The BP-elimination statements then follow from standard value-function monotonicity with respect to capacity parameters.}
\end{proof}
\rev{Our GUE characterization implies that enforcing $\bm \Pi^\star_s$ makes GUE equivalent to centralized optimization of $\Phi_s$, $s \in \{\mathrm{T},\mathrm{P},\mathrm{C}\}$. Specifically, (a) $\bm \Pi^\star_\mathrm{T}$ optimizes $\Phi_\mathrm{T}$ by internalizing travelers' congestion externalities, as in \citep{he2013integrated}. It eliminates PBP because the centralized problem is independent of $\mathbf{\bar f}$, so line expansion does not affect $\mathbf x$ at GUE. Sensitivity analysis also implies increasing any $\alpha_{\ell_\mathrm{T}}$ cannot decrease optimal transportation cost. (b) $\bm \Pi^\star_\mathrm{P}$ optimizes $\Phi_\mathrm{P}$ by reimbursing travel costs, so route choice is effectively based on LMPs and charging loads move to buses with minimal LMPs. It eliminates TBP because the centralized problem is independent of $\bm \alpha$, while increasing $\mathbf{\bar f}$ only enlarges the feasible region. (c) $\bm \Pi_\mathrm{C}^\star$ minimizes total social cost and eliminates type T-C and P-C BPs.}

\subsubsection{Static Pricing} \rev{We seek exogenous charging prices $\bm \Pi \in \mathcal P \subseteq \mathbb R^{\nr}$, where $\mathcal P$ is a convex set encoding elementary policy constraints (i.e., box constraints). When $\bm \Pi$ is fixed, the travel-dispatch influence is no longer mutual. Given $\bm \Pi$, the transportation UE $\mathbf x(\bm \Pi)$ can be characterized by a convex program (see Lemma 12 in \hyperref[sec:proofs]{Appendix~E}),}
\rev{which in turn through the charging load drives the economic dispatch dual solution $\bm \lambda(\bm \Pi)$. We denote the two convex programs determining $(\mathbf{x}(\bm \Pi), \bm \lambda(\bm \Pi))$ compactly by:}
\begin{align}\label{eq:two_convex_programs}
	\mathbf{x}(\bm \Pi) = \mathrm{UE}(\bm \Pi), \qquad \bm \lambda(\bm \Pi) = \mathrm{ED}(\mathbf{x}(\bm \Pi)).
\end{align}
As $\mathbf{x}(\bm \Pi)$ does not depend on $\mathbf{\bar f}$, changing line capacity has no effect on $\Phi_\mathrm{T}$. Meanwhile, charging load $\mathbf{d}$ also does not change with $\mathbf{\bar f}$, so expanding a line would never increase $\Phi_\mathrm{P}$.
\begin{corollary}[Static Pricing Eliminates PBP]\label{cor:perfect_power_expansion}
    Under static charging pricing $\bm \Pi$, increasing transmission line capacity $\bar f_{\ell_\mathrm{P}}, \forall \ell_\mathrm{P} \in [m_\mathrm{P}]$  never increases $\Phi_s, s \in \{\mathrm{T},\mathrm{P},\mathrm{C}\}$.
\end{corollary}

\rev{Corollary~\ref{cor:perfect_power_expansion} shows that static pricing eliminates PBP by breaking the coupling. To test whether it can eliminate TBP, we derive $\partial \Phi_\mathrm{T}/\partial \alpha_{\lt}$ and $\partial \Phi_\mathrm{P}/\partial \alpha_{\lt}$ under $\bm \Pi$ (\hyperref[sec:proofs]{Appendix~E}). These derivatives depend on $(\mathbf{x}(\bm \Pi),\bm \lambda(\bm \Pi))$ and hence on $\bm \Pi$. We therefore view the programs defined in \eqref{eq:two_convex_programs} as parametric in $\bm \Pi$, and partition $\mathcal P$ into regions where solution binding patterns are invariant.}

\begin{definition}[\mh{Critical Region}]\label{def:constraint_binding}
Given a constraint-binding pattern $\mathcal B \triangleq (\mathcal{R}, \mathcal{L}_\mathrm{P}) \subseteq [\nr]\times [\mp]$,
the set of static pricing policies $\bm \Pi \in \mathcal P$ that induce $(\mathbf{x}(\bm \Pi), \bm \lambda(\bm \Pi))$ such that $\{r:x_r(\bm \Pi) = 0\} =\mathcal{R}$ and $\{\ell_\mathrm{P}:(\mathbf{Hp})_\mathrm{\ell_\mathrm{P}} = \bar f_{\ell_\mathrm{P}}\} = \mathcal{L}_\mathrm{P}$ is called a critical region and denoted by $\mathcal P^\mathrm{cr}_{\mathcal B}$, where $\mathbf p$ is the optimal power injection of \eqref{eq:economic_dispatch} given $\mathbf{x}(\bm \Pi)$.
\end{definition}

In fact, $\mathbf{x}(\bm \Pi)$ and $\bm \lambda(\bm \Pi)$ are affine functions of the parameters $\bm \Pi$ in each critical region, and overall piecewise affine. 
\rev{
\begin{lemma}[GUE is Piecewise Affine in Charging Price]\label{thm:GUE_affine}
    Let $\bm \Pi \in \mathcal P^\mathrm{cr}_{\mathcal B}$ for some $\mathcal B$. There exist $\mathbf{K} \in \mathbb{R}^{\nr \times \nr}$, $\mathbf{C} \in \mathbb{R}^{\np \times \nr}$, $\mathbf{v}\in \mathbb{R}^{\nr}$, and $\mathbf{w}\in \mathbb{R}^{\np}$ such that $\mathbf{x}(\bm \Pi) = \mathbf{K}\bm \Pi + \mathbf{v}$ and $\bm \lambda(\bm \Pi) = \mathbf{C}\bm \Pi+\mathbf{w}$. Furthermore, $\mathcal P^\mathrm{cr}_{\mathcal B}$ is a convex set.
\end{lemma}
}
\begin{proof}[\rev{Proof sketch}]
\rev{Given a critical region $\mathcal P^\mathrm{cr}_{\mathcal B}$, the zero-flow route set $\mathcal R$ and binding line set $\mathcal L_\mathrm{P}$ are fixed. The UE program in \eqref{eq:two_convex_programs} can then be reduced to routes not in $\mathcal{R}$, and its KKT conditions form a linear system in the active entries of $\mathbf x(\bm\Pi)$ and $\bm\Pi$; solving this system gives $\mathbf x(\bm\Pi)=\mathbf K\bm\Pi+\mathbf v$, with inactive route entries fixed at zero. Given this affine $\mathbf x(\bm\Pi)$, fixing $\mathcal L_\mathrm{P}$ reduces the ED KKT conditions to a linear system in $(\mathbf g,\mathbf p,\bm\lambda)$ whose right-hand side is affine in $\mathbf x(\bm\Pi)$. Hence $\bm\lambda(\bm\Pi)=\mathbf C\bm\Pi+\mathbf w$. The region itself is described by the same active-route and line-binding inequalities under these affine expressions, so it is convex.}
\end{proof}

The explicit expressions of $\mathbf{x}(\bm \Pi)$ and $\bm \lambda(\bm \Pi)$ allow us to express $\partial \Phi_\mathrm{T}/\partial \alpha_{\ell_\mathrm{T}}$ and $\partial \Phi_\mathrm{P}/\partial \alpha_{\ell_\mathrm{T}}$ in terms of $\bm \Pi$. 
We now state the \mh{BP} mitigation problem as a feasibility program:
\begin{subequations}\label{eq:BP_elimination}
    \begin{align}
        \mathrm{find}\quad &  \bm \Pi \in \mathcal P^\mathrm{cr}_{\mathcal B} \\
    \mathrm{s.t.} \quad &   
    \frac{\partial \Phi_\mathrm{T}}{\partial \alpha_{\ell_\mathrm{T}}} \geq 0, \quad \frac{\partial \Phi_\mathrm{P}}{\partial \alpha_{\ell_\mathrm{T}}} \geq 0, \quad \forall \ell_\mathrm{T} \in [m_\mathrm{T}],\label{eq:constraint_BP}\\
    &\bm \Pi^\top \mathbf{x}(\bm \Pi) \ge \theta,\label{eq:constraint_revenue}
    \end{align}
\end{subequations}
where 
\eqref{eq:constraint_BP} requires eliminating type T-T and T-P BPs through $\bm \Pi$ (which together also ensure type T-C BP will not occur), and \eqref{eq:constraint_revenue} ensures the collected charging revenue is no smaller than some threshold $\theta \in \mathbb{R}$. 
Problem \eqref{eq:BP_elimination} appears to be daunting, but it turns out to be structurally simple as all constraints are either affine or convex quadratic:
\begin{theorem}[BP Mitigation as a Convex Program]\label{thm:BP_convex}
    All constraints in \eqref{eq:BP_elimination} are convex, and thus \eqref{eq:BP_elimination} is convex.
\end{theorem}
\noindent\textit{\rev{Proof sketch.}} \rev{Within a fixed critical region, Lemma~\ref{thm:GUE_affine} makes $\mathbf{x}(\bm\Pi)$ and $\bm\lambda(\bm\Pi)$ affine in $\bm\Pi$, so \eqref{eq:constraint_BP} becomes affine. The revenue condition is a concave-superlevel constraint, and intersecting it with the convex critical region and affine half-spaces gives a convex feasible set.}

\rev{Optimization~\eqref{eq:BP_elimination} identifies a static price $\bm \Pi$ within a critical region for binding pattern $\mathcal B$. We iterate over critical regions using parametric-programming tools \citep{kvasnica2004multi} and stop once a feasible $\bm \Pi$ is found; a nonexistence certificate over $\mathcal P$ requires exploring all regions intersecting $\mathcal P$.}

\rev{
\begin{remark}[Computational Complexity of Static Pricing]
    In practice, \(\mathcal P\) may contain many critical regions, and exhaustive enumeration can be exponential in the problem size. We tested three coupled systems: IEEE \(9\)-bus/Bay Area (\texttt{small}), IEEE \(14\)-bus/Bay Area (\texttt{medium}), and IEEE \(57\)-bus/Anaheim (\texttt{large}) \cite{transportationNetworksResearch}. We use an active-set adjacency implementation based on the critical-region construction in MPT \cite{kvasnica2004multi}. After each region is discovered, we solve \eqref{eq:BP_elimination} and proceed to a new adjacent region only if \eqref{eq:BP_elimination} is infeasible. With five independent random seeds, the procedure examined, on average, 5.4, 3.8, and 1.2 regions for the \texttt{small}, \texttt{medium}, and \texttt{large} systems, respectively, and found a feasible static price in all five runs for each system. The corresponding mean wall-clock times were 1.63, 3.08, and 22.36 seconds, including region discovery, GUE solves, and feasibility checks. These experiments do not remove the worst-case exponential dependence, but they demonstrate the possibility of early termination once a static price is found. Details are provided in \hyperref[sec:additional_numerical_study]{Appendix~F}.
\end{remark}
}

\section{Numerical Study}\label{sec:numerical}
\condense{We demonstrate BPs on a coupled system adapted from \cite{alizadeh2016optimal}, combining the IEEE 9-bus test system with a California Bay Area transportation network (Fig.~\ref{fig:transportation_network}).}

\emph{System Settings.}
\rev{We use the setting in \cite{alizadeh2016optimal} with slight modifications. The generation-cost coefficients are $\mathbf{Q}=\mathrm{diag}([0.11,0.085,0.1225,0,0,0,0,0,0]^\top)$ \$/MW$^2$ and $\bm \mu = [5,1.2,1,0,0,0,0,0,0]^\top$ \$/MW. Here buses $1$--$3$ are generators, buses $4$--$9$ have generation fixed to zero, and the zeros in $\mathbf Q$ at buses $4$--$9$ do not represent free generation variables. The numerical ED uses $\mathbf p=\mathbf g-\mathbf d_0-\mathbf d(\mathbf x)$ with base loads $\mathbf{d}_0 = [0,480,0,10,160,80,0,40,120]^\top$ MW. We use $\mathbf{\bar f} = [250,250,25,300,10,250,250,250,250]^\top$ MW, increasing bus-2 load to $480$MW and reducing capacities of lines $(5,6)$ and $(6,7)$ to $25$MW and $10$MW to promote congestion. The right panel of Fig.~\ref{fig:power_network} shows the power network.}


\rev{We consider one O-D pair, (\textit{Davis}, \textit{San Jose}), and still observe TBPs and PBPs. Four chargers are placed at \textit{Winters}, \textit{Fairfield}, \textit{Mtn.View}, and \textit{Fremont}. We assume $N=15,000$ EVs, each with fixed charging demand $\rho = 20\mathrm{kW}$. As in \cite{alizadeh2016optimal}, link cost is $c_{\ell}=\alpha_{\ell}x_{\ell}+\beta_{\ell}$, where $\alpha_{\ell}x_{\ell}$ is traffic-dependent and $\beta_{\ell}$ is fixed. The vector $\bm \beta$ is obtained from the same fixed travel times using a value of time $0.32$\$/min ($\approx19$\$/hour). The same physical route with different chargers is treated as different routes, giving $10$ routes in total. We call this the \emph{base setting} and perturb it to demonstrate BPs.}

\begin{figure}[h]
  \centering
  \begin{subfigure}[t]{0.2\textwidth}
    \centering
   \includegraphics[width=\textwidth]{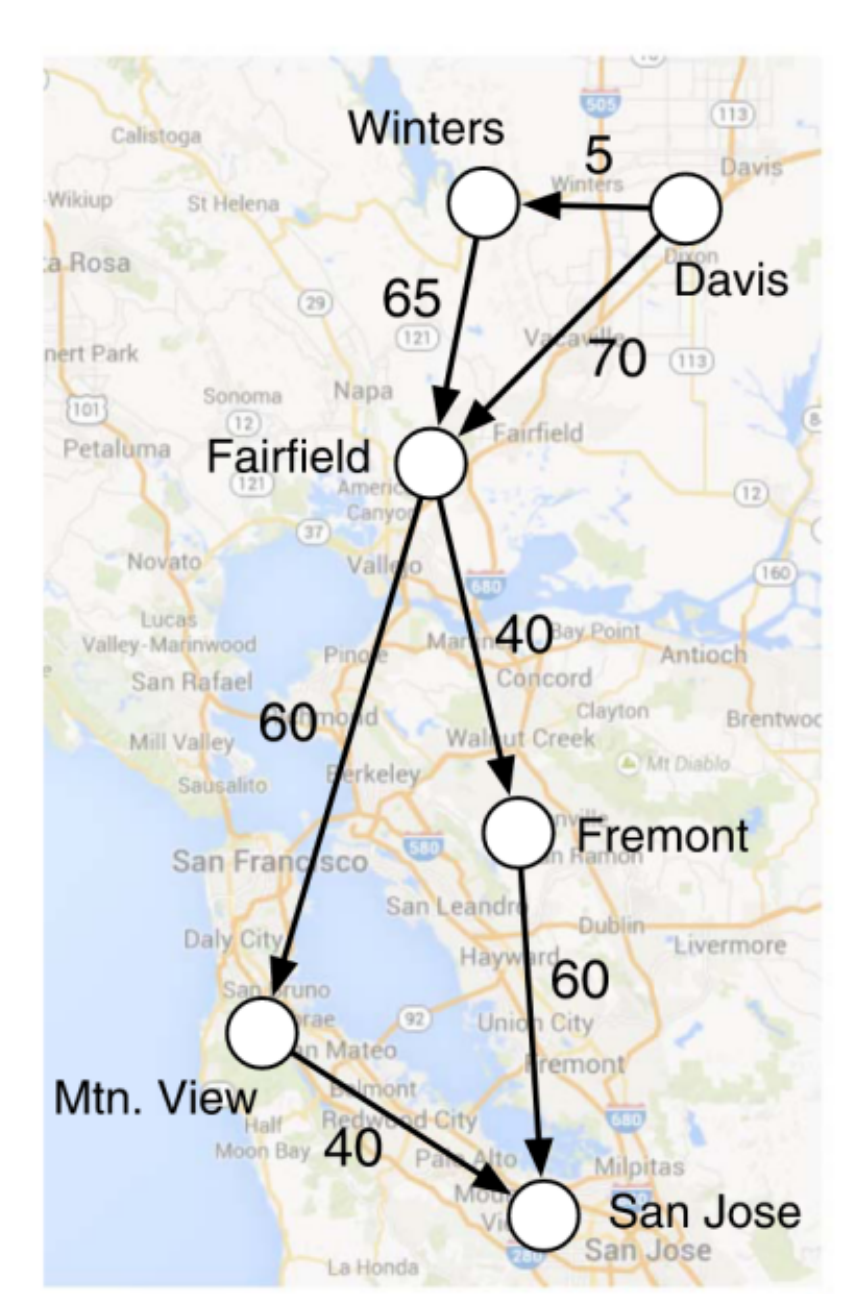}
    \caption{Transportation Network}
    \label{fig:transportation_network}
  \end{subfigure}
  \quad
  \begin{subfigure}[t]{0.2\textwidth}
    \centering
    \includegraphics[width=\textwidth]{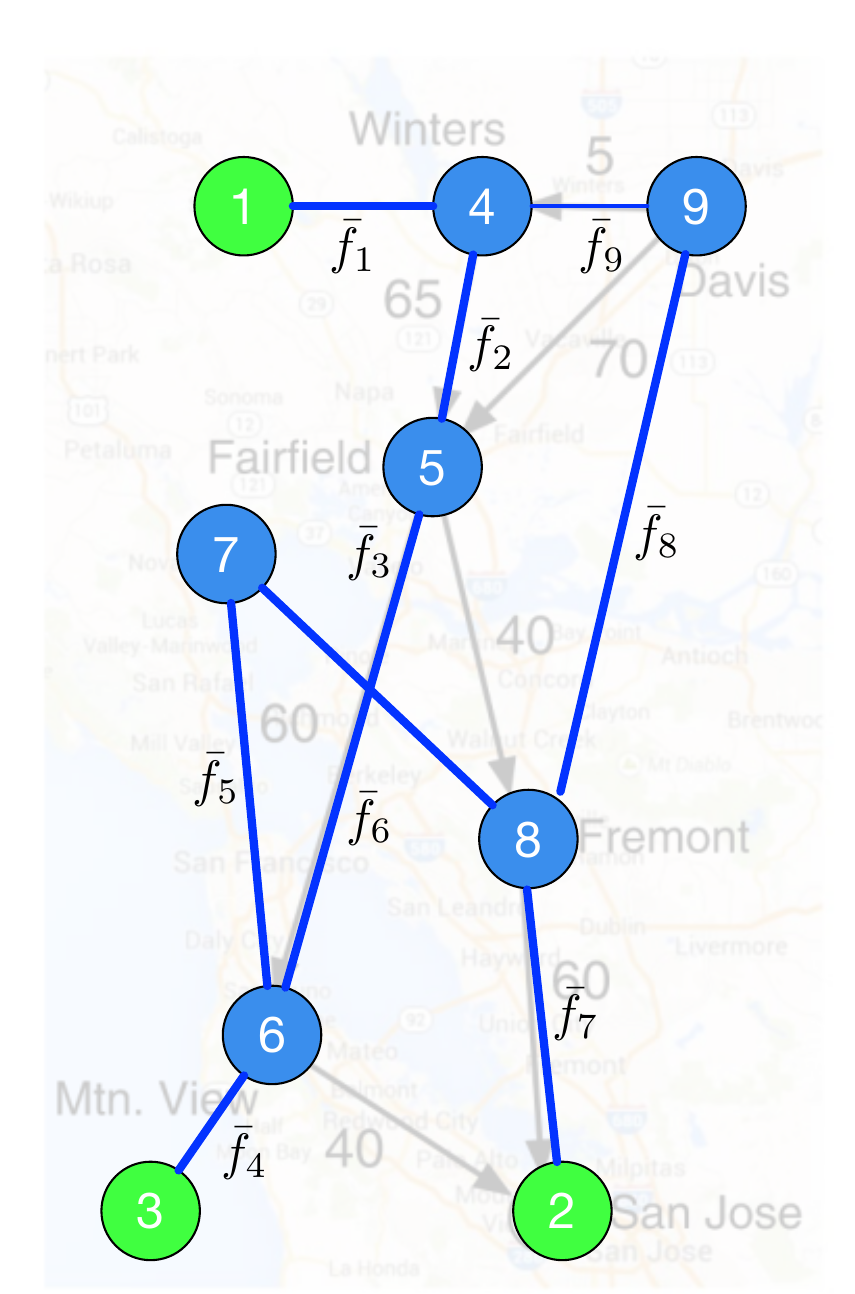}
    \caption{Power Network}
    \label{fig:power_network}
  \end{subfigure}
  \caption{The coupled system. The {left panel} is the transportation network. Numbers on links represent the fixed travel times  {in minutes}. The {right panel} is the power network, in which buses in green are generators and others in blue are load buses.}
\end{figure}

\subsubsection{Type P-P}
As Sections~\ref{subsec:uncongested} and \ref{subsec:completely_congested} show, PBP requires power-network congestion. We modify the base setting by decreasing the capacity of line $(2,8)$, the seventh entry in $\mathbf{\bar f}$, to $160$MW and keep other settings unchanged. Expanding $(2,8)$ from $160$MW to $200$MW relocates load, increasing generation at bus $2$ and decreasing it at bus $3$ (Fig.~\ref{fig:new_pp}-(a)). Since bus $2$ has much larger marginal generation cost than bus $1$ and the price difference increases (Fig.~\ref{fig:new_pp}-(b)), $\Phi_\mathrm{P}$ rises by $\approx2.4\%$ (Fig.~\ref{fig:new_pp}-(c)).

\begin{figure}[!htbp]
\centering
\begin{subfigure}[t]{0.155\textwidth}
    \centering
    \includegraphics[width=\textwidth]{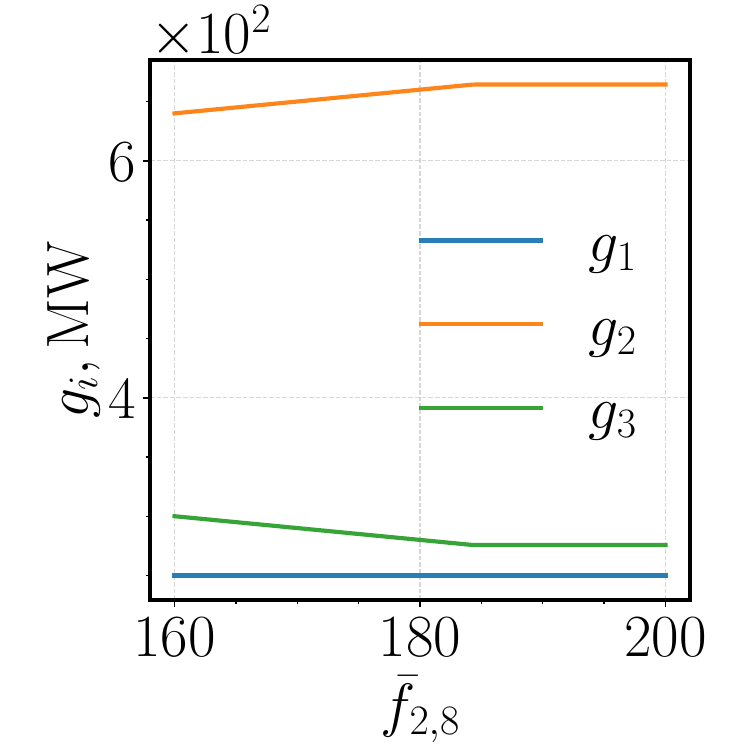}
    \caption{Generation}
\end{subfigure}
\hfill
\begin{subfigure}[t]{0.155\textwidth}
    \centering
    \includegraphics[width=\textwidth]{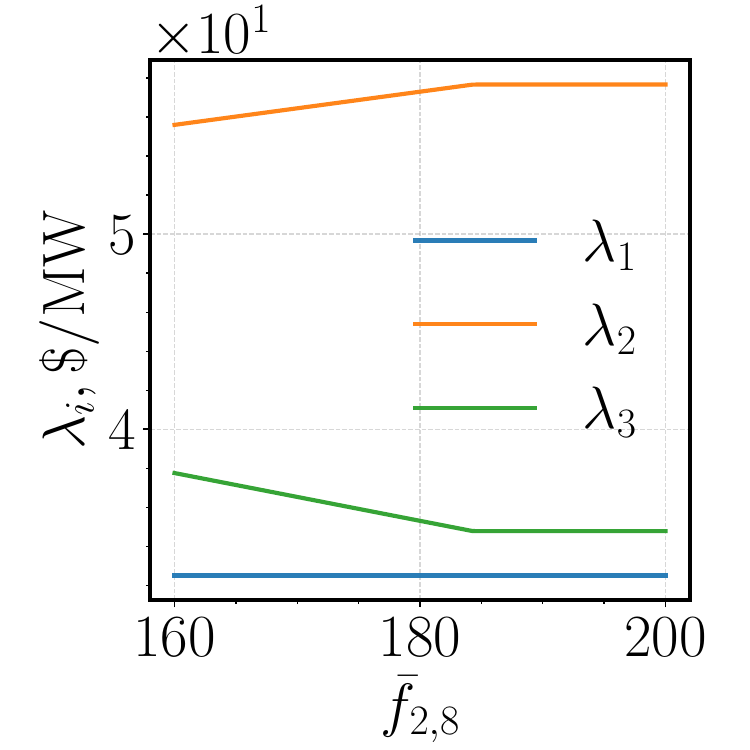}
    \caption{Price}
\end{subfigure}
\hfill
\begin{subfigure}[t]{0.16\textwidth}
    \centering
    \includegraphics[width=\textwidth]{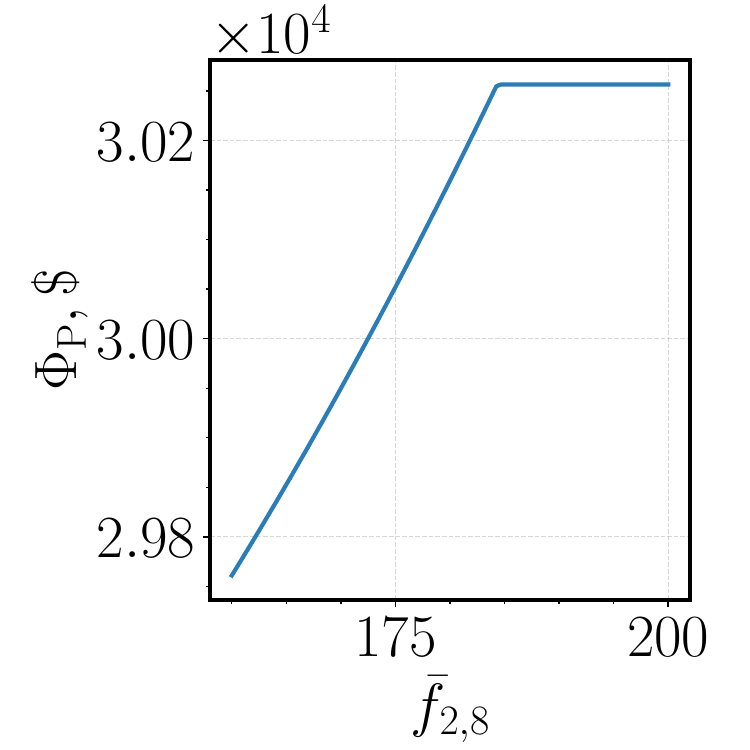}
    \caption{$\Phi_\mathrm{P}$}
\end{subfigure}
\caption{\rev{Type P-P BP occurs when expanding line $(2,8)$ ($\Phi_\mathrm{P}$ rises by about $2.4\%$). The horizontal axes show the line-capacity values, and the vertical axes show generator outputs, marginal prices, and $\Phi_\mathrm{P}$.}}
\label{fig:new_pp}
\end{figure}

\begin{figure}[!htbp]
\centering

\begin{subfigure}[t]{0.155\textwidth}
    \centering
    \includegraphics[width=\textwidth]{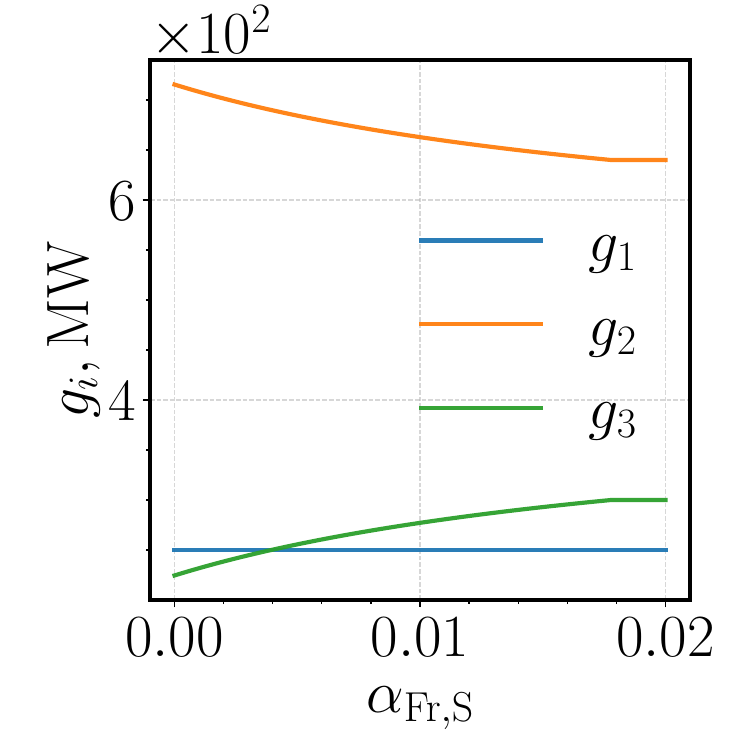}
    \caption{Generation}
\end{subfigure}
\hfill
\begin{subfigure}[t]{0.155\textwidth}
    \centering
    \includegraphics[width=\textwidth]{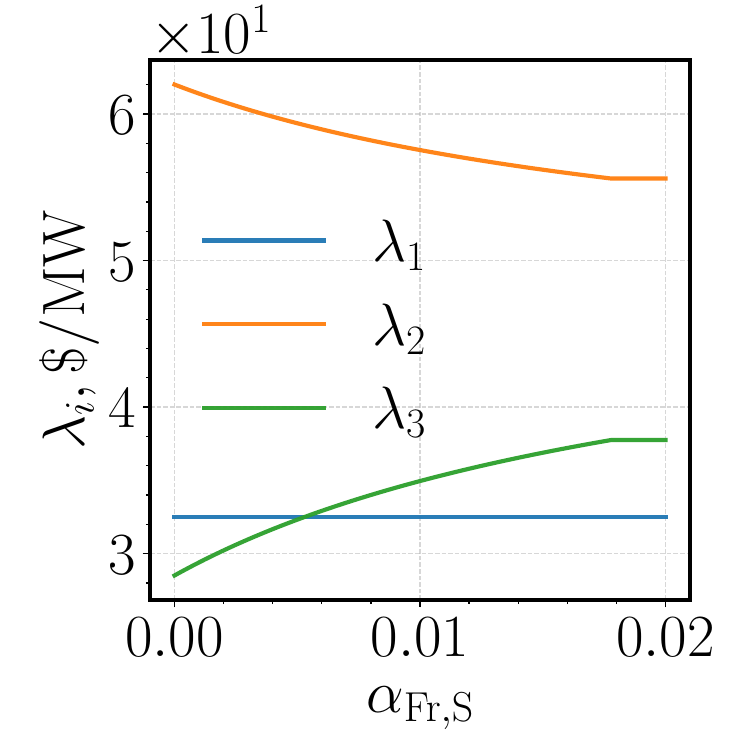}
    \caption{Price}
\end{subfigure}
\hfill
\begin{subfigure}[t]{0.155\textwidth}
    \centering
    \includegraphics[width=\textwidth]{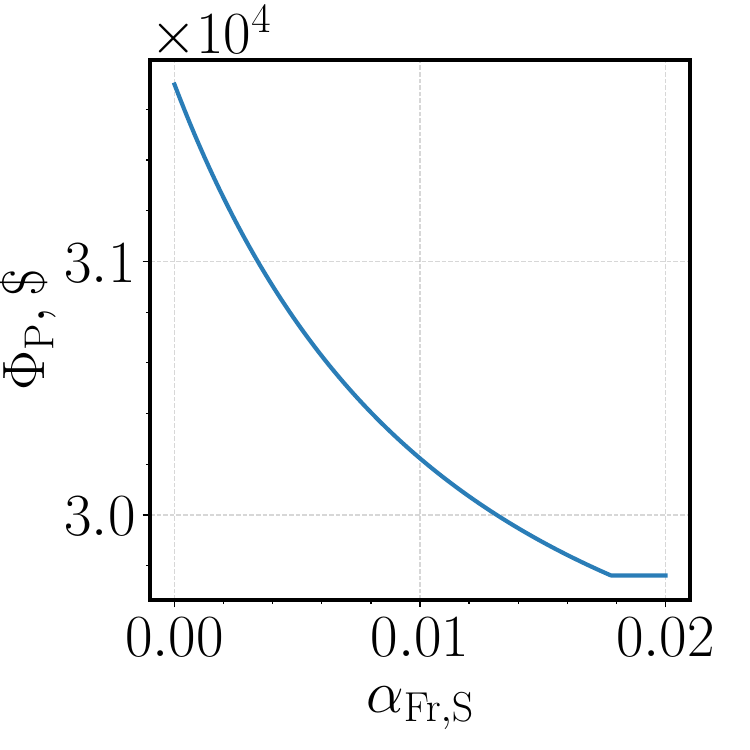}
    \caption{$\Phi_\mathrm{P}$}
\end{subfigure}
\caption{\rev{Type T-P BP occurs when expanding the road (\textit{Fremont}, \textit{San Jose}) ($\Phi_\mathrm{P}$ rises by about $6.4\%$). The horizontal axes show the road-capacity values, and the vertical axes show generator outputs, marginal prices, and $\Phi_\mathrm{P}$.}}
\label{fig:new_tp}
\end{figure}

%

\subsubsection{Type T-P}
Recall that type P-P BP is driven by \textit{load relocation toward expensive generators}. The same effect can be achieved by modifying the setting of one road in the base setting: we expand the road connecting \textit{Fremont} and \textit{San Jose}. 
Decreasing $\alpha_{\mathrm{Fr,S}}$ \rev{places an even heavier generation burden on} expensive bus $2$ (Fig.~\ref{fig:new_tp}-(a),(b)), increasing $\Phi_\mathrm{P}$ by $\approx6.4\%$ (Fig.~\ref{fig:new_tp}-(c)). Remarkably, the ``P-P occurs $\to$ T-P occurs'' implication in Fig.~\ref{fig:BP_relations} still holds although the power network is \emph{not} radial.

\begin{figure}[!htbp]
\centering

\begin{subfigure}[b]{0.155\textwidth}
    \centering
    \includegraphics[width=\textwidth]{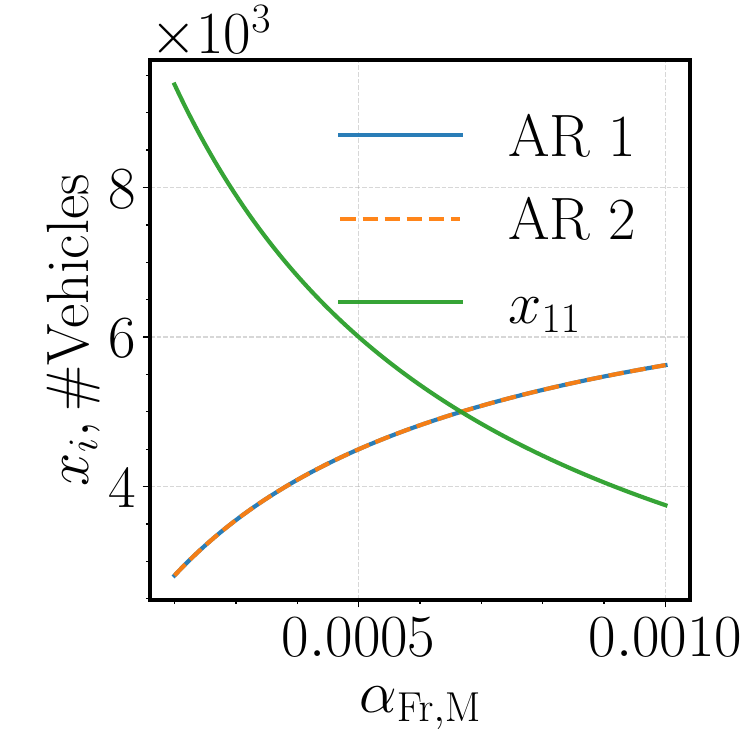}
    \caption{Traffic flow}
\end{subfigure}
\hfill
\begin{subfigure}[b]{0.155\textwidth}
    \centering
    \includegraphics[width=\textwidth]{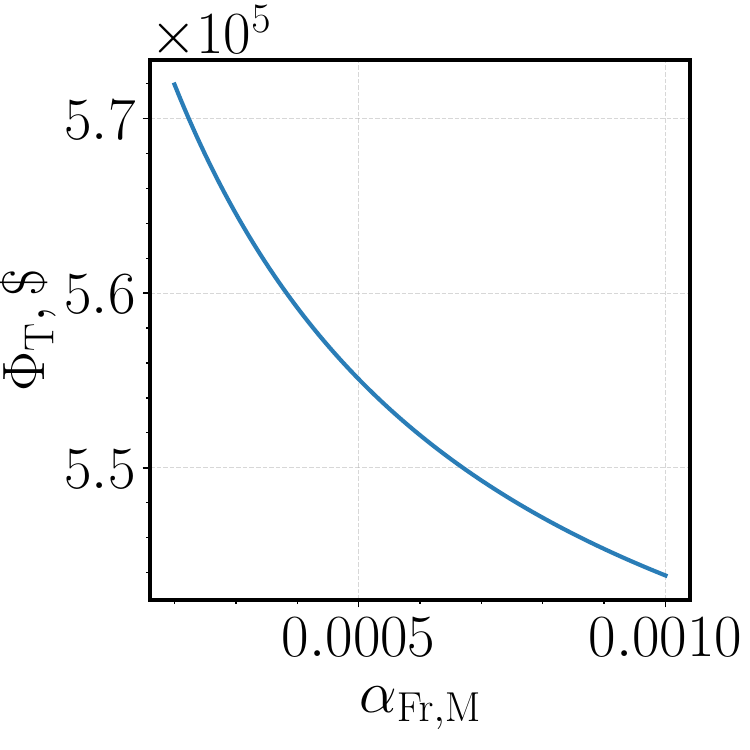}
    \caption{$\Phi_\mathrm{T}$}
\end{subfigure}
\hfill
\begin{subfigure}[b]{0.155\textwidth}
    \centering
    \includegraphics[width=\textwidth]{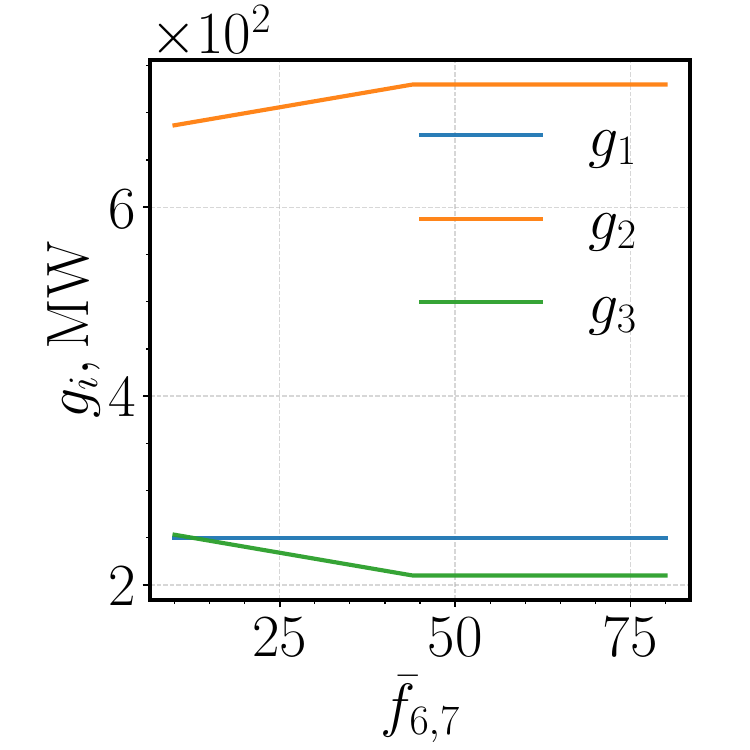}
    \caption{Generation}
\end{subfigure}

\begin{subfigure}[b]{0.155\textwidth}
    \centering
    \includegraphics[width=\textwidth]{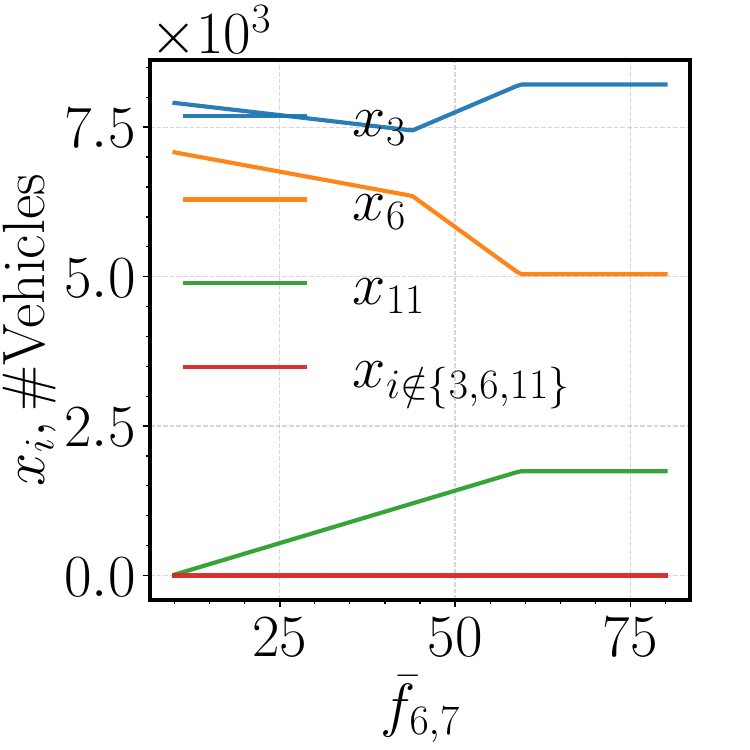}
    \caption{Traffic flow}
\end{subfigure}
\hfill
\begin{subfigure}[b]{0.155\textwidth}
    \centering
    \includegraphics[width=\textwidth]{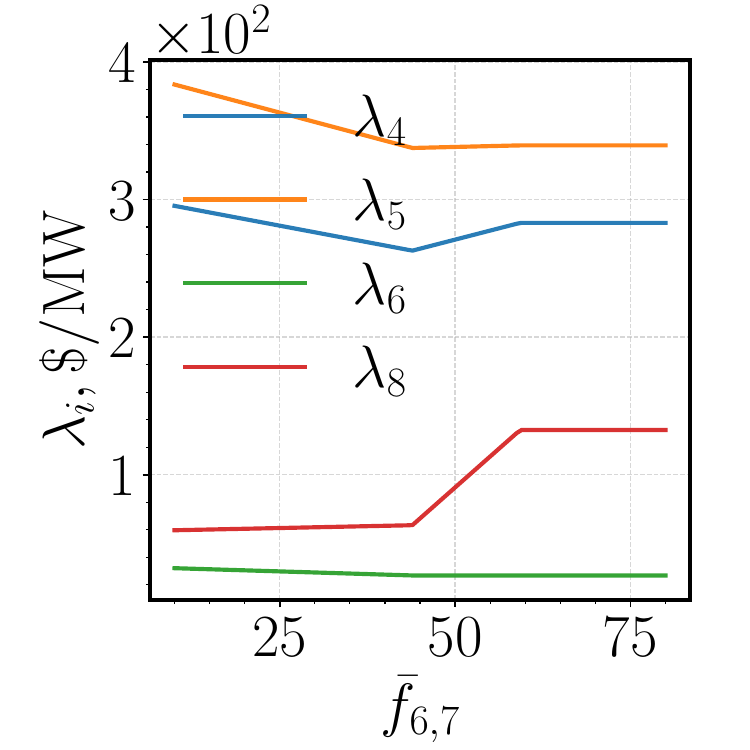}
    \caption{Price}
\end{subfigure}
\hfill
\begin{subfigure}[b]{0.155\textwidth}
    \centering
    \includegraphics[width=\textwidth]{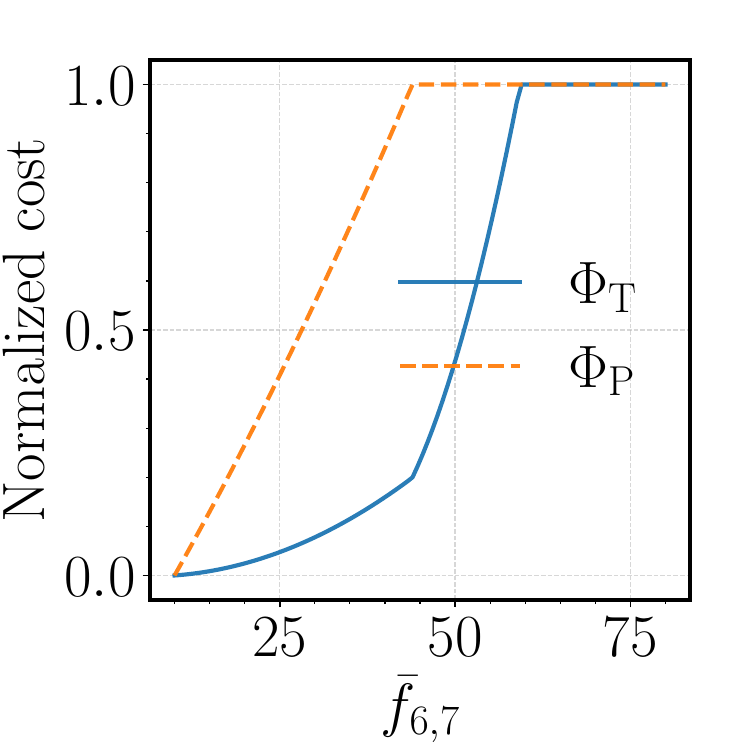}
    \caption{$\Phi_\mathrm{T}$ and $\Phi_\mathrm{P}$}
\end{subfigure}

\caption{\rev{Type T-T, P-T, and P-P BPs occur in the modified network with the added (\textit{Fremont}, \textit{Mtn.View}) road. Panels (a) and (b) expand that road, and the axes report route flows and $\Phi_\mathrm{T}$ over the road-capacity sweep. Panels (c)--(f) expand line $(6,7)$. The axes report generator outputs, route flows, prices, and social costs.}}
\label{fig:my_2by3}
\end{figure}


\subsubsection{Type T-T}
Type T-T BP does not occur in the base or perturbed settings, so we introduce different settings to produce it. Inspired by the Wheatstone network \citep{braess2005paradox}, we add road\footnote{{It can be shown that the previous type P-P and T-P BPs still occur in the new coupled system under the same setting, as long as we set the cost of traveling through (\textit{Fremont}, \textit{Mtn.View}) to be sufficiently large.}} (\textit{Fremont}, \textit{Mtn.View}). The new transportation network contains the Wheatstone network formed by \textit{Fairfield, Mtn.View, Fremont}, and \textit{San Jose}.
To mimic classical BP, we set capacities of $(5,6)$ and $(6,7)$ to $250$MW so the power network is uncongested at GUE, and set $\bm \alpha=10^{-4}\times[0,0,0,0,6.67,6.67,0,10]^\top$ and $\bm \beta=[0,0,0,20,10,10,20,0]$, with last entries for the new road. The transportation network then effectively reduces to a Wheatstone network with three routes: two symmetric aggregate routes (ARs) not passing through (\textit{Fremont, Mtn.View}) and one passing through it. The new route containing the road, route $11$ (\textit{Davis, Winters, Fairfield, Fremont, Mtn.View, San Jose}), is associated with the charger at \textit{Winters}.
{In Fig.~\ref{fig:my_2by3}-(a) and (b),} we expand {the capacity of} (\textit{Fremont, Mtn.View}). 
Fig.~\ref{fig:my_2by3}-(a) shows that road expansion relocates flow to route $11$. When $\alpha_{\mathrm{(Fr,M)}}$ is small, almost all traffic is attracted to route $11$, matching the classical BP mechanism. Fig.~\ref{fig:my_2by3}-(b) confirms that $\Phi_\mathrm{T}$ increases by $\approx6.5\%$; since all travelers see the same price, this recovers classical BP.

\subsubsection{Type P-T} Type P-T and P-P BPs occur simultaneously in this modified network (i.e., the modified network considered for producing type T-T BP).\footnote{Note the simultaneous occurrence of type P-T and P-P BPs is not covered in Fig.~\ref{fig:BP_relations}, which suggests that the diagram like Fig.~\ref{fig:BP_relations} for general systems might be much more complicated.} We keep the same $\bm \alpha$ and $\bm \beta$ as above except $\alpha_{\mathrm{Fr,M}} = 2\times 10^{-4}$. Fig.~\ref{fig:my_2by3}-(d) shows that expanding line $(6,7)$ induces the same flow relocation as type T-T BP, due to decreasing charging-price differentials (Fig.~\ref{fig:my_2by3}-(e)). Traffic attracted to route $11$ increases $\Phi_\mathrm{T}$ by $\approx1.0\%$ and puts heavier generation burden on expensive bus $2$ (Fig.~\ref{fig:my_2by3}-(c)), increasing $\Phi_\mathrm{P}$ by $\approx5.2\%$ (Fig.~\ref{fig:my_2by3}-(f)).

Above, network settings are fixed and BPs are induced by perturbing $\bm \alpha$ and $\bar{\mathbf{f}}$. Two further questions are \emph{BP sensitivity}, i.e., how parameters other than $\bm \alpha$ and $\bar{\mathbf{f}}$ affect BP extent, and \emph{BP mitigation}, i.e., whether Section~\ref{sec:mitigation}'s pricing policies work. We study how $\rho$ and $\mathbf{Q}$ affect type P-T BP extent, and test adaptive pricing in the four settings producing four BP types. The results show that both $\rho$ and $\mathbf{Q}$ significantly affect type P-T BP extent, and adaptive pricing eliminates certain BP types as predicted. Additional results are in \hyperref[sec:additional_numerical_study]{Appendix~F}.

\section{Conclusion}\label{sec:Conclusion}

This paper studied Braess' paradoxes in coupled power and transportation systems. Through simple examples, we show that all types of BP can arise and identify the underlying mechanisms. We establish necessary and sufficient conditions for BP occurrence in coupled systems with progressively richer power-network structures, and for general systems with non-radial power networks, we propose a computational framework for BP screening. We also study mitigation through charging pricing and provide theoretical guarantees for BP elimination. Overall, this work represents a first step toward understanding how power-transportation coupling affects equilibrium and infrastructure planning. The present analysis is based on stylized models with limitations that invite future study, including mixed EV and internal-combustion fleets and explicit charging-station capacity constraints.

\bibliographystyle{IEEEtran} 
\footnotesize
\bibliography{bib} 

\normalsize
\onecolumn
\begin{center}
\begin{tcolorbox}[
 width=\linewidth,
 colback=white,
 colframe=black,
 boxrule=0.8pt,
 arc=0mm,
left=2mm,right=2mm,top=2mm,bottom=2mm
]
\centering
{\Large\bfseries Supplementary Materials for}\\[3pt]
{\Large\bfseries ``Braess' Paradoxes in Coupled Power and Transportation Systems''}
\end{tcolorbox}
\end{center}
\section{Appendices}
\subsection{GUE as the Optimal Solution to A Convex Program}\label{sec:GUE_as_optimal_solution}
In order to screen BPs, we need to compute the GUE for general networks, which amounts to identify a \emph{(travel pattern, price)} pair satisfying Definition \ref{def:GUE}. As observed in \S\ref{sec:2by2-example}, the problem of identifying a transportation system equilibrium given the charging prices can be cast as a variant of the standard problem of finding the transportation system UE, with the linear part of the travel cost modified to include the charging cost. This allows us to adopt the following classical result using convex optimization to compute the UE for a transportation system \citep{beckmann1956studies}:
\begin{lemma}[Transportation  UE given $\bm \lambda$]\label{lemma:tse}
    An admissible $\mathbf{x^\star}$ is a transportation system UE for a given $\bm \lambda$ if and only if it solves
\begin{subequations}\label{eq:solving_transportation_equilibrium}
    \begin{align}
        \min_{\mathbf{x}} ~& \quad  \frac{1}{2} \mathbf{(\Alr x)^\top \mathrm{diag}(\bm \alpha) \Alr x + [ (\Alr)^\top \bm \beta + \bm \pi(\bm \lambda)]^\top \mathbf x},\\
        \mathrm{s.t.} ~& \quad \nu: \mathbf{1^\top x} = 1; \quad
         \bm \xi: \mathbf{x} \geq \mathbf{0}.
    \end{align}
    \end{subequations}
\end{lemma}
\begin{proof}
As the problem is convex and Slater's condition holds, the KKT conditions are necessary and sufficient for optimality. Let $\mathbf{x}^\star$ be an optimal solution given $\bm \lambda^\star$. The stationarity condition is:
    \begin{align}
        \bm \xi + \nu \mathbf{1} = (\Alr)^\top \left[\mathrm{diag}(\bm \alpha) \Alr \mathbf{x}^\star + \bm \beta\right] + \bm \pi(\bm \lambda^\star) = c_r(\mathbf{x}^\star, \bm \lambda^\star).
    \end{align}
    Therefore, $c_r(\mathbf{x}^\star, \bm \lambda^\star) = \mh{\nu}$ if $x_r^\star > 0$, and $c_r(\mathbf{x}^\star, \bm \lambda^\star) = \xi_r + \mh{\nu} \geq \mh{\nu}$ if $x_r^\star = 0$, which is consistent with Definition \ref{def:t:eq}.
\end{proof}

    

Equipped with the optimization-based characterization of the transportation UE, we can further construct an optimization with both transportation and power system variables and objective terms whose solution coincides with the GUE for the coupled system. This idea of converting GUE computation to an optimization problem is fundamentally identical to that used in \textit{potential games} \citep{monderer1996potential}, where strategic interactions can be transformed into an optimization framework that captures the equilibrium state. Such techniques have also been used in similar contexts in \cite{he2013integrated} when considering dynamic wireless charging and in \cite{he2013optimal} for charging station planning.
\begin{theorem}[GUE Computation]\label{thm:solving_wardrop}
    The \emph{(travel pattern, price)} pair $(\mathbf x^\star, \bm \lambda^\star)$ is a generalized user equilibrium (GUE) for the coupled power and transportation system if and only if it solves
    \begin{subequations}\label{eq:auxiliary_opt}
        \begin{align}
    J:=    \min_{\mathbf{x},\,\mathbf{g},\,\mathbf{p}} &\quad  \frac{1}{2} \mathbf g^\top \mathbf Q \mathbf g + \bm \mu^\top \mathbf g + \frac{1}{2} \mathbf x^\top (\Alr)^\top \mathrm{diag}(\bm \alpha) \Alr \mathbf x + \bm \beta^\top \Alr\mathbf x  \\            \mathrm{s.t.} &\quad \bm \lambda: \mathbf{p =  g-d(x)};
            \quad \gamma: \mathbf{1^\top p} = 0;
            \quad  \bm \eta: \mathbf{Hp \leq \bar f};
            \quad \nu: \mathbf{1^\top x} = 1;
            \quad \bm \xi: \mathbf{x \geq 0}.
        \end{align}
    \end{subequations}
\end{theorem}
\begin{proof}
As the Slater's condition holds under our parameter assumptions for \eqref{eq:economic_dispatch}, \eqref{eq:solving_transportation_equilibrium}, and \eqref{eq:auxiliary_opt}, 
    KKT conditions are necessary and sufficient for optimality.
    We then show an $(\mathbf x^\star, \bm \lambda^\star)$ pair satisfies \eqref{eq:economic_dispatch} and \eqref{eq:solving_transportation_equilibrium} (i.e., it is a GUE) if and only if it also satisfies \eqref{eq:auxiliary_opt}, by matching the KKT conditions of \eqref{eq:auxiliary_opt} with that of  \eqref{eq:economic_dispatch} and \eqref{eq:solving_transportation_equilibrium}.
\end{proof}

Theorem~\ref{thm:solving_wardrop} offers both an alternative characterization of the GUE for the coupled system and an efficient way to compute GUE as the solution of a convex quadratic program. As a byproduct, we have the following result on the existence and uniqueness of GUE for general networks.

\begin{corollary}[Existence and Uniqueness of GUE]\label{cor:existence}
    There exists at least one GUE for each $(\bm \alpha,  \bar{\mathbf{f}}) \in \mathbb{R}_+^{\mt} \times  \mathbb{R}_{++}^{m^\mathrm{p}}$. Moreover, \mh{if $\bm \alpha \in \mathbb{R}_{++}^{\mt}$}, then GUE is unique.
\end{corollary}
\begin{proof}
The objective function of \eqref{eq:auxiliary_opt} is coercive and the feasible region is closed. The existence of an optimal solution is guaranteed by Weierstrass theorem. The optimal solution is in fact finite (i.e. there exists $M \in \mathbb{R}$ such that $\|(\mathbf{g}^\star,\mathbf{p}^\star)\| < M$) because we assume $\mathbf{Q}$ to be diagonal with strictly positive entries so any entry of $\mathbf{g}^\star$ cannot be infinity. Uniqueness is due to the strict convexity of \eqref{eq:auxiliary_opt} if $\Alr$ is assumed to be of full column rank.
\end{proof}


\begin{remark}[Caveat for the Existence of GUE]
A caveat for Corollary~\ref{cor:existence} is that we have assumed sufficient generation capacity at every bus of the power network to simplify our exposition, which in practice, may be implied by the (usually conservative) charging infrastructure interconnection/permitting processes. As a result, the spatial shift of charging loads driven by the capacity changes will not lead to infeasibility  for the economic dispatch problem~\eqref{eq:economic_dispatch} or the joint optimization for computing the GUE of the coupled system~\eqref{eq:auxiliary_opt} in our model. Whether this observation can be extrapolated to real systems depends on the permitting process in place for the particular system under consideration.
\end{remark}

\newpage
\subsection{Generalization to Multiple O-D Pairs}\label{sec:generalization_to_multiple_O-D_pairs}
Let $\mathcal{W}$ be the set of all O-D pairs, $X_{i,j} \geq 0$ be the traffic demand for O-D pair $(i,j) \in \mathcal{W}$, $\mathcal{R}_{i,j}$ be the collection of routes connecting O-D pair $(i,j)$, and $x^r_{i,j} \geq 0$ be the traffic flow on route $r \in \mathcal{R}_{i,j}$. We collect all $x_{i,j}^r$ for a fixed $(i,j)$ and $r \in \mathcal{R}_{i,j}$ into a column vector $\mathbf{x}_{i,j} \in \mathbb{R}^{|\mathcal{R}_{i,j}|}_+$.

Define incidence matrices $\Alr_{(i,j)}\in \mathbb{R}^{m_\mathrm{T}\times |\mathcal{R}_{i,j}|}$ and $\Acr_{(i,j)} \in \mathbb{R}^{\nc \times |\mathcal{R}_{i,j}|}$ as:
\begin{subequations}
	\begin{align}
		&\left(\Alr_{(i,j)}\right)_{\ell,r} := \mathbbm{1}\left\{ \text{link $\ell$ is contained in route $r \in \mathcal{R}_{i,j}$} \right\}, \quad (\ell,r) \in [\mt] \times [|\mathcal{R}_{i,j}|],\\
		&\left(\Acr_{(i,j)}\right)_{z,r} := \mathbbm{1}\left\{ \text{charging station $z$ is contained in route $r \in \mathcal{R}_{i,j}$} \right\}, \quad (z,p) \in [\nc] \times [|\mathcal{R}_{i,j}|].
	\end{align}
\end{subequations}
Then, the link flows through O-D pair $(i,j)$ can be computed as $\mathbf{x}^{\mathrm{L}}_{i,j} := \Alr_{(i,j)} \mathbf{x}_{i,j}$ and loads contributed by traffic are $\mathcal{R}_{i,j}$ contributes are $\mathbf{d}_{i,j}(\mathbf{x}_{i,j}) := \rho \left(\Acb\right)^\top \Acr_{(i,j)} \mathbf{x}_{i,j}$. Therefore, the aggregate link flows and loads are $\mathbf{x}^{\mathrm{L}} = \sum_{(i,j) \in \mathcal{W}} \mathbf{x}^{\mathrm{L}}_{i,j}$ and $\mathbf{d}(\{\mathbf{x}_{i,j}:(i,j) \in \mathcal{W}\}) = \sum_{(i,j)\in \mathcal{W}} \mathbf{d}_{i,j}$. The flow constraints are:
\begin{subequations}
	\begin{align}
		& \mathbf{1^\top x}_{i,j} = X_{i,j}, \quad \forall (i,j) \in \mathcal{W},\\
		&\mathbf{x}_{i,j} \geq 0, \quad \forall (i,j) \in \mathcal{W}.
	\end{align}
\end{subequations}
\begin{lemma}[Transportation UE for Multiple O-D pairs Case]
	Given charging prices $\bm \lambda$, the transportation UE  is equivalent to optimal solution to the below convex program:
\begin{subequations}\label{eq:tse_multiple_OD}
	\begin{align}
		\min_{\mathbf{x}^{\mathrm{L}},\mathbf{x}_{i,j},(i,j)\in \mathcal{W}} ~& \frac{1}{2}(\mathbf{x}^{\mathrm{L}})^\top \mathrm{diag}(\bm \alpha)\mathbf{x}^{\mathrm{L}} + \bm \beta^\top \mathbf{x}^\mathrm{L}+\bm \lambda^\top \mathbf{d}(\{\mathbf{x}_{i,j}:(i,j) \in \mathcal{W}\}),\\
		\mathrm{s.t.} ~& \bm \omega:  \mathbf{x}^\mathrm{L} = \sum_{(i,j)\in \mathcal{W}} \Alr_{(i,j)}\mathbf{x}_{i,j},\\
		&\nu_{i,j}:\mathbf{1^\top x}_{i,j} = X_{i,j}, \quad \forall (i,j) \in \mathcal{W};\qquad \bm \xi_{i,j}:\mathbf{x}_{i,j} \geq 0, \quad \forall (i,j) \in \mathcal{W}.
	\end{align}
\end{subequations}
\end{lemma}
\begin{proof}
The stationarity conditions are:
\begin{subequations}
	\begin{align}
		&\mathrm{diag}(\bm \alpha) \mathbf{x}^\mathrm{L} + \bm \beta + \bm \omega = \mathbf{0},\label{eq:stationarity_multiple_OD_f}\\
		&\rho (\Acr_{(i,j)})^\top \Acb \bm \lambda - (\Alr_{(i,j)})^\top \bm \omega - \nu_{i,j}\mathbf{1} -\bm \xi_{i,j}= \mathbf{0}, \quad \forall (i,j) \in \mathcal{W}. \label{eq:stationarity_multiple_OD_x}
	\end{align}
\end{subequations}
Let $(i,j) \in \mathcal{W}$, left multiply \eqref{eq:stationarity_multiple_OD_f} by $(\Alr_{(i,j)})^\top$, and apply \eqref{eq:stationarity_multiple_OD_x}. Then we have
\begin{equation}\label{eq:eq_condition_multiple_OD}
	\left(\Alr_{(i,j)}\right)^\top \mathrm{diag}(\bm \alpha) \mathbf{x}^\mathrm{L} + \left(\Alr_{(i,j)}\right)^\top \bm \beta + \rho \left(\Alr_{(i,j)}\right)^\top \Acb \bm \lambda = \nu_{i,j}\mathbf{1} + \bm \xi_{i,j}.
\end{equation}
The first two terms in \eqref{eq:eq_condition_multiple_OD} is a vector with entries the travel costs for routes connecting $(i,j)$, and the third term is a vector with entries the charging costs associated those routes. \eqref{eq:eq_condition_multiple_OD} means any optimal solution $\mathbf{x}_{i,j}^\star$ to \eqref{eq:tse_multiple_OD} satisfies the transportation UE condition given $\bm \lambda$.

Conversely, given a transportation UE $\mathbf{x}^\star_{i,j}, (i,j) \in \mathcal{W}$, we can define $(\mathbf{x}^\mathrm{L})^\star:=\sum_{(i,j)\in \mathcal{W}}\Alr_{(i,j)}\mathbf{x}^\star_{i,j}$, $\bm \omega^\star:=-(\mathrm{diag}(\bm \alpha)(\mathbf{x}^\mathrm{L})^\star + \bm \beta)$, and dual variables $\nu_{i,j}^\star$ and $\bm \xi_{i,j}^\star$ in a way similar to the proof of Lemma \ref{lemma:tse}. It is straightforward to show the primal-dual tuple $((\mathbf{x}^\mathrm{L})^\star,\{\mathbf{x}^\star_{i,j},\nu^\star_{i,j},\bm \xi_{i,j}\}_{(i,j)\in \mathcal{W}})$ satisfies KKT conditions of \eqref{eq:tse_multiple_OD} and thus is optimal.
\end{proof}

\begin{theorem}[GUE Computation for Multiple O-D pairs]
	If we denote $\mathbf{x}:=[\mathbf{x}_{i,j}]_{(i,j)\in \mathcal{W}}$ as a vector with entries $\mathbf{x}_{i,j}$. The \emph{(price, travel pattern)}  pair $(\mathbf{x}^\star, \bm \lambda^\star)$ is a generalized user equilibrium for the coupled power and transportation system if and only if it solves
    \begin{subequations}\label{eq:GUE_Multiple_OD}
        \begin{align}
             \min_{\mathbf{x}_{i,j},(i,j)\in \mathcal{W},\mathbf{x}^\mathrm{L},\mathbf{g},\mathbf{p}} &~ \frac{1}{2}(\mathbf{x}^\mathrm{L})^\top \mathrm{diag}(\bm \alpha)\mathbf{x}^\mathrm{L} + \bm \beta^\top \mathbf{x}^\mathrm{L} + \frac{1}{2}\mathbf{g^\top Q g + \bm \mu^\top g},\\
            \mathrm{s.t.} &~ \bm \lambda: \mathbf{d(x) + p - g} = \mathbf{0};\quad \gamma:\mathbf{1^\top p} = 0;\quad \bm \eta:\mathbf{Hp \leq \bar f},\\
            &~ \bm \omega: \mathbf{x}^\mathrm{L} = \sum_{(i,j)\in \mathcal{W}}\Alr_{(i,j)}\mathbf{x}_{i,j};\quad \nu_{i,j}:\mathbf{1^\top x}_{i,j} = X_{i,j}; \quad \bm \xi_{i,j}:\mathbf{x}_{i,j} \geq 0, \quad \forall (i,j) \in \mathcal{W}.
        \end{align}
    \end{subequations}
\end{theorem}
\begin{proof}
	Its proof idea is the same as the proof of Theorem \ref{thm:solving_wardrop}. Since Slater's condition holds, KKT conditions are necessary and sufficient for optimality. KKT conditions of \eqref{eq:GUE_Multiple_OD} is the collection of KKT conditions of \eqref{eq:tse_multiple_OD} and \eqref{eq:economic_dispatch}. Therefore, any optimal solution $(\mathbf{x}^\star, \bm \lambda^\star)$ to \eqref{eq:GUE_Multiple_OD} must satisfy transportation UE condition and $\bm \lambda^\star$ is determined through economic dispatch given $\mathbf{x}^\star$. Conversely, given a GUE $(\mathbf{x}^\star, \bm \lambda^\star)$, we define $(\mathbf{x}^\mathrm{L})^\star := \sum_{(i,j)\in \mathcal{W}} \Alr_{(i,j)}\mathbf{x}_{i,j}^\star$ and dual variables $\gamma^\star, \bm \eta^\star, \bm\omega^\star,\nu_{i,j}^\star$ and $\bm \xi_{i,j}^\star$ similar to Theorem \ref{thm:solving_wardrop}. The primal-dual tuple satisfies KKT conditions of \eqref{eq:solving_transportation_equilibrium}.
\end{proof}

Throughout the paper, we only consider the case with one O-D pair and assume the traffic demand is unity. All results are based on convex programs in Lemma \ref{lemma:tse} and Theorem \ref{thm:solving_wardrop}. Problems \eqref{eq:tse_multiple_OD} and \eqref{eq:solving_transportation_equilibrium} generalizes those convex programs to multiple O-D pairs, and then all results in this paper can be easily generalized.
\newpage
\subsection{\mh{Supplementary Materials for \S\ref{sec:model} and \ref{sec:examples}}}\label{apd:auxiliary_results}

\begin{proof}[Proof of Lemma \ref{lemma:equal_cost}]
	This is a direct consequence of Definition \ref{def:t:eq}. Let $\mathbf{x}^\star$ be a transportation UE given $\bm \lambda$. For any $r_1,r_2 \in [\nr]$ such that $\mathbf{x}^\star_{r_i} > 0, i = 1,2$, we have $c_{r_1}(\mathbf{x}^\star, \bm \lambda) \leq c_{r_2}(\mathbf{x}^\star, \bm \lambda)$ and $c_{r_2}(\mathbf{x}^\star, \bm \lambda) \leq c_{r_1}(\mathbf{x}^\star, \bm \lambda)$. Therefore, $c_{r_1}(\mathbf{x}^\star, \bm \lambda) = c_{r_2}(\mathbf{x}^\star, \bm \lambda)$. There exists some constant $C$ such that $c_{r}(\mathbf{x}^\star, \bm \lambda) = C$ for any $r$ such that $\mathbf{x}^\star_r > 0$.
\end{proof}

\begin{proof}[Proof of Proposition \ref{prop:2routes_uncongested}]
	This proposition is a direct consequence of the discussion in \S\ref{sec:2x2:unc}.
\end{proof}

\begin{proof}[Proof of Proposition \ref{prop:2routes_congested}]
	This proposition is a direct consequence of the discussion in \S\ref{sec:2x2:c}.
\end{proof}

\begin{proof}[Proof of Proposition \ref{prop:perfect_power_expansion_2by2}]
It follows from Theorem \ref{thm:fully_congested}-(b3) and (b4), since for this particular example, the power network is radial and there are two subnetworks and two active route bundles each containing exactly one active route. 

Without loss of generality, we assume power flows from bus $1$ to $2$, and therefore $\lambda_1 < \lambda_2$. It can be checked that \eqref{eq:pt_fully_congested} reduces to 
	\begin{align}
		Q_1 \psi_1 - Q_2 \psi_2 \propto (Q_1+Q_2)(\lambda_2-\lambda_1) > 0.
	\end{align}
	Therefore, type P-T BP does not occur.
	
	It can also be checked that \eqref{eq:pp_fully_congested} reduces to 
	\begin{align}
		\varsigma_1 - \varsigma_2 = \left(1 - \frac{\rho^2 (Q_1+Q_2)}{\alpha_1+\alpha_2+\rho^2(Q_1+Q_2)}\right)(\lambda_1 - \lambda_2) < 0.
	\end{align}
	Therefore, type P-P BP does not occur.
%
\end{proof}

\paragraph{The 2-Route 3-Bus Coupled System} We set the shift-factor matrix $\mathbf{H}$ to be
\begin{equation}
\mathbf{H} = \begin{bmatrix}
	\widehat{\mathbf H}\\
	-\widehat{\mathbf H}
\end{bmatrix} \quad \mbox{with} \quad 
    \widehat{\mathbf{H}} = \begin{bmatrix}
        0 & -0.8 & -0.6 \\
        0 & 0.2 & 0.4 \\
        0 & -0.2 & 0.6
    \end{bmatrix},
\end{equation}
which corresponds to some susceptance values of the power lines. The economic dispatch problem for the 3-Bus power network \red{can be equivalently written as:}
\begin{subequations}\label{eq:ed:3bus}
    \begin{align}
        \Phi_\mathrm{P}(\mathbf x) := \min_{\mathbf{g},\mathbf{f},\mathbf{p}} \quad & \frac{1}{2} \mathbf{g^\top Q g},\\
        \mathrm{s.t.} \quad & \lambda_1: p_1 = g_1 - \rho x_1 = f_1 - f_2;\\
        & \lambda_2: p_2 = g_2 - \rho x_2 = -f_1 - f_3;\\
        & \lambda_3: p_3 = g_3 = f_2 + f_3;\\
        & \mathbf f = \widehat{\mathbf H} \mathbf p;\\
        & |f_\ell| < \bar f_\ell,\quad  \ell=1,2,3;
    \end{align}
\end{subequations}
where we assume $\bm \mu = \mathbf{0}$ (i.e. the generation cost is purely quadratic), $f_\ell$ denotes the power flow between three buses with positive direction indicated in Figure \ref{fig:3-bus_2-route}, and $\bar f_\ell$ is flow capacity for the $\ell$-th line in both positive and negative flow directions (note that here we have deviated from our notation in \eqref{eq:economic_dispatch} slightly to simplify the exposition). 

\mh{Combining the UE equations given the charging prices~\eqref{eq:2bus:eq} and the economic dispatch problem~\eqref{eq:ed:3bus}, we can in fact analytically solve for GUE given any power system congestion pattern.
For instance,  when line congestion pattern is $f_1 = \bar f_1$ and $f_3 = \bar f_3$ (i.e., line 1 and line 3 are congested in the positive direction),  we have
\begin{subequations}
    \begin{align*}
    &x_1^\star = \frac{\alpha_2 + \rho^2 Q_2 - \frac{\rho Q_1}{3}(4\bar f_1 - \bar f_3)-\rho Q_2(\bar f_1+\bar f_3)}{\alpha_1+\alpha_2+\rho^2 (Q_1+Q_2)},\\
    &x_2^\star = \frac{\alpha_1 + \rho^2 Q_1 + \frac{\rho Q_1}{3}(4\bar f_1 - \bar f_3) + \rho Q_2(\bar f_1+\bar f_3)}{\alpha_1+\alpha_2+\rho^2 (Q_1+Q_2)},\\
    &g_1^\star = \rho x_1^\star + \frac{1}{3}(4\bar f_1-\bar f_3), \quad
    g_2^\star = \rho x_2^\star - (\bar f_1+\bar f_3),\\
    &g_3^\star = \frac{4\bar f_3 - \bar f_1}{3}, \quad f_2 = \frac{\bar f_1-\bar f_3}{3}, \quad
    \bm \lambda^\star = \mathbf{Q g^\star}.
    \end{align*}
\end{subequations}
}

\begin{proof}[Proof of Theorem \ref{thm:3by2_summary}]
	For this particular example, the occurrences of the mentioned BPs are all illustrated in \S\ref{sec:2by3-example}. 
\end{proof}

\subsubsection{An Alternative Explanation of the Occurrence of Type P-T BP} For the 2-Route 3-Bus coupled system, we can also conclude the existence of type P-T \mh{BP} by focusing on the changes in LMPs. 
Denote $\Delta \lambda_{21}^\star:= \lambda^\star_2 - \lambda^\star_1$ (also see the second panel of Figure \ref{fig:phi_t_bar_f}). Then we can rewrite \eqref{eq:phitf3} as
\begin{equation}
	\frac{\partial \Phi_\mathrm{T}}{\partial \bar f_3} = \frac{\rho^2}{\alpha_1 + \alpha_2} \frac{\partial (\Delta \lambda_{21}^\star)^2}{\partial \bar f_3}. 
\end{equation}
Thus we have type P-T \mh{BP} for this system if the line capacity expansion enlarges the spatial LMP differential seen at the two charging locations. From Figure \ref{fig:phi_t_bar_f}, increasing $\bar f_3$ amplifies the price difference $\Delta \lambda_{21}^\star$ (Figure \ref{fig:phi_t_bar_f}). Therefore, $\Phi_\mathrm{T}$ increases.

\subsubsection{Shift of Line Congestion Pattern} It is easily seen from Fig. \ref{fig:phi_t_bar_f} that the power system congestion pattern shifts  
     at $\bar f_3 = 0.8$. For $\bar f_3 <0.8$, $f_1=-\bar f_1$ and $f_3 = \bar f_3$ so line 1 is congested in the negative direction and line 3 is congested in the positive direction; for $0.8<\bar f_3<1$, 
     $f_2 = \bar f_2$ and $f_3 = \bar f_3$ so lines 2 and 3 are congested in the positive direction while line 1 is no longer congested. We can uniquely determine the GUE under two different line congestion patterns:
\begin{equation}
    x_1^\star = \begin{cases}
        \frac{\alpha_2 + \rho^2 + \rho Q_1(\frac{4}{3}\bar f_1 + \frac{1}{3}\bar f_3)+\rho Q_2(\bar f_1-\bar f_3)}{\alpha_1+\alpha_2+\rho^2(Q_1+Q_2)}, & \bar f_3 <0.8,\\
        \frac{\alpha_2 + \rho^2 + \rho Q_1(4\bar f_2 - \bar f_3) + \rho Q_2(3\bar f_2 - 2\bar f_3)}{\alpha_1+\alpha_2+\rho^2(Q_1+Q_2)}, & 0.8<\bar f_3<1.
    \end{cases}
\end{equation}
In both cases, $x_1^\star$ is a linear function in $\bar f_3$. Under the second congestion pattern , $x_1^\star$ has a negative slope that is smaller than the negative slope of $x_1^\star$ under the first congestion pattern, which explains the sharp decrement in $x_1^\star$ after $\bar f_3 = 0.8$. Similarly, one can explain the sharp increment in price difference $\lambda_1^\star - \lambda_2^\star$. 
These results highlight the role that power network congestion patterns play in how line capacity expansion impacts the traffic flows at the GUE.

\subsubsection{Type T-T and T-P BPs in the 2-Route 3-Bus Coupled System} We visualize the occurrences of type T-T and T-P BPs in the 2-Route 3-Bus coupled system in Figures \ref{fig:phi_t_alpha} and \ref{fig:phi_p_alpha}.
\begin{figure}[!htbp]
    \centering
    \includegraphics[width=.8\textwidth]{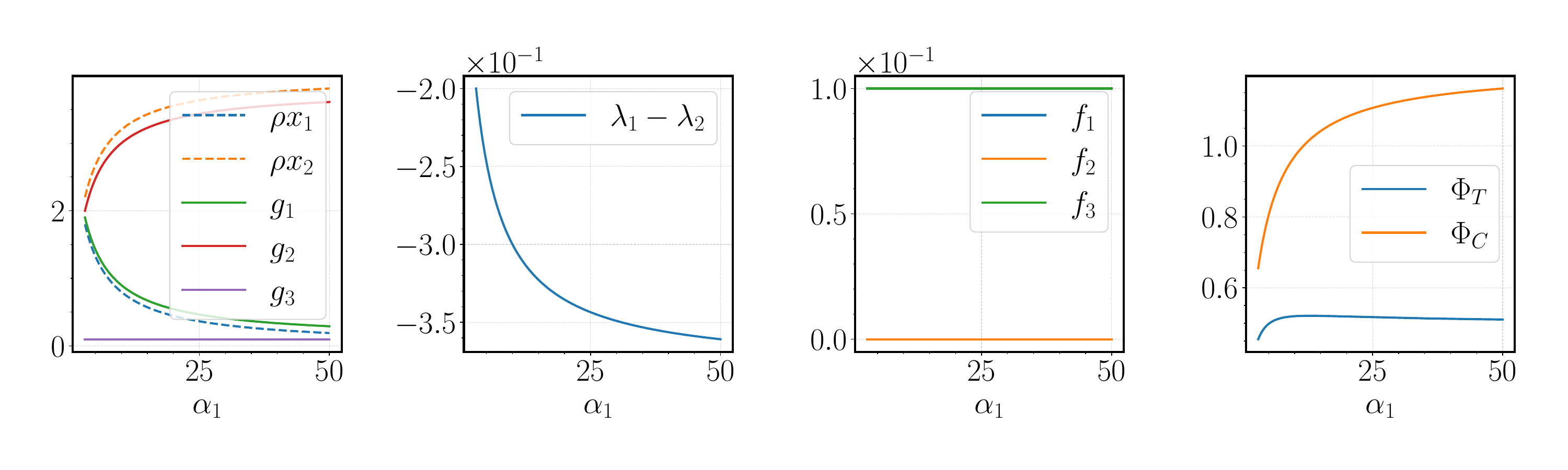}
    \caption{GUE and social cost metrics change with $\alpha_1$ where type T-T \mh{BP} occurs. Model parameters are set as $\bm \alpha := [\alpha_1,1]^\top, \rho := 4, \mathbf{Q} := \mathrm{diag}([0,0.1,0.1])$, and $\mathbf{\bar f} := [0.1,0.3,0.1]^\top$, where $\alpha_1\in[3,50]$.}
    \label{fig:phi_t_alpha}
\end{figure}

\begin{figure}[!htbp]
    \centering
\includegraphics[width=.8\textwidth]{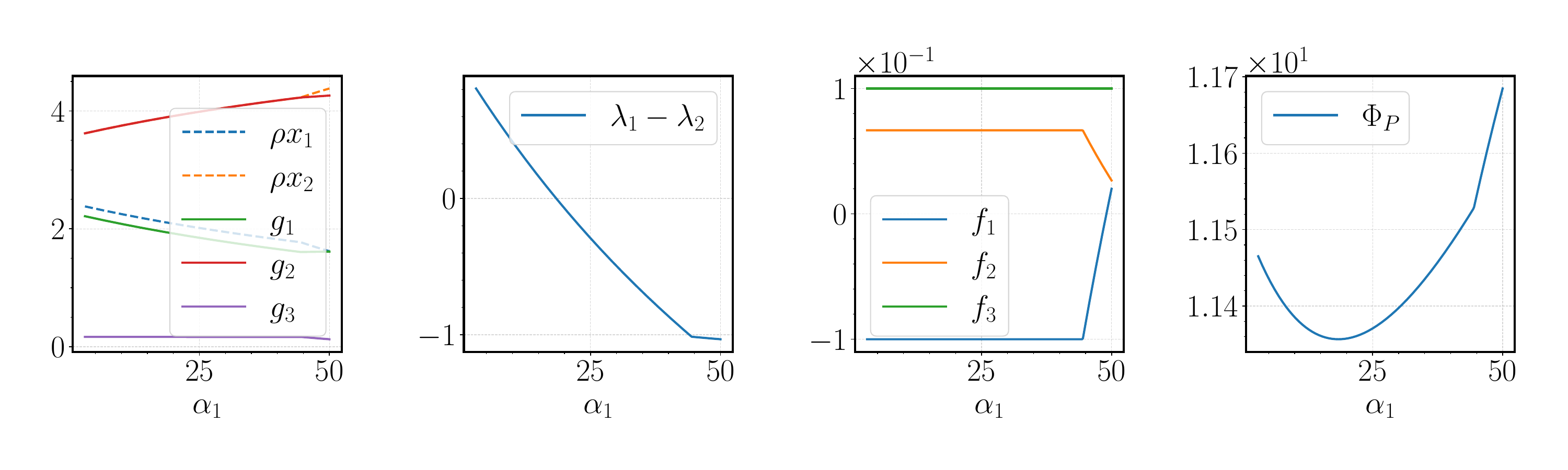}
    \caption{GUE and power system social cost change with $\alpha_1$ where type T-P \mh{BP} occurs. Model parameters are set as $\bm \alpha := [\alpha_1,10]^\top, \rho := 6, \mathbf{Q} := \mathrm{diag}([2,1,1])$, and $\mathbf{\bar f} := [0.1,0.3,0.1]^\top$, where $\alpha_1 \in [3,50]$.}
    \label{fig:phi_p_alpha}
\end{figure}

\newpage
\subsection{Supplementary Materials for \S \ref{sec:NS_conditions_sec}}\label{sec:supp_materials}
\subsubsection{Perturbation Analysis}\label{sec:perturbation_analysis}
Define an \textit{equivalence relation}\footnote{A binary relation $\sim$ on a set $\mathcal{X}$ is said to be an equivalence relation if and only if it is
\begin{enumerate}
	\item[(a)] reflexive: $x\sim x, \forall x \in \mathcal{X}$;
	\item[(b)] symmetric: $x\sim y \Leftrightarrow y \sim x, \forall x,y \in \mathcal{X}$;
	\item[(c)] transitive: $x\sim y, y\sim z \Rightarrow x\sim z$. \citep{grillet2007abstract}
\end{enumerate}
It could be easily checked that the ``$\sim$'' defined here is an equivalence relation on $\mathbb{R}_{++}^{\mt} \times \mathbb{R}_{+}^{\mp}$.} ``$\sim$'' such that two parameter pairs $(\bm \alpha_1, \bar{\mathbf{f}}) \sim (\bm \alpha_2,\bar{\mathbf{f}}_2)$ if their induced constraint binding patterns at the solution of \eqref{eq:auxiliary_opt} are the same. The relation $\sim$ induces equivalence classes on $\mathbb{R}_{+}^{\mt} \times \mathbb{R}_{++}^{\mp}$, which we call \textit{critical regions}. We assign indices to the $\mp+\nr$ inequality constraints by the set $[\mp + \nr]$. Any subset $\mathcal{B} \subseteq [\mp+\nr]$ uniquely specifies a constraint binding pattern that constraints indexed by $\mathcal{B}$ are binding, while other constraints are not binding.

\begin{definition}[Critical Region]
	Let $\mathcal{B} \subseteq [\mp + \nr]$ be the constraint binding pattern induced by the pair $(\bm \alpha, \bar{\mathbf{f}}) \in \mathbb{R}_{+}^{\mt} \times \mathbb{R}_{++}^{\np}$. Then, the equivalence class $\mathcal{C}_\mathcal{B} := \left\{(\bm \alpha',\bar{\mathbf{f}}'):(\bm \alpha',\bar{\mathbf{f}}') \sim (\bm \alpha, \bar{\mathbf{f}})\right\}$ is called the critical region of Problem \eqref{eq:auxiliary_opt} characterized by $\mathcal{B}$, and $\{\mathcal{C}_\mathcal{B}\}_{\mathcal{B} \subseteq [\mp+\nr]}$ is a finite collection of all critical regions. 
\end{definition}

Given a $(\bm \alpha, \bar{\mathbf{f}}) \in \mathcal{C}_\mathcal{B}$. \textit{We perturb $(\bm \alpha, \bar{\mathbf{f}})$ within $\mathrm{int}(\mathcal{C}_\mathcal{B})$}, which guarantees the constraint binding pattern of \eqref{eq:auxiliary_opt} remains unchanged. \mh{For it to work, we have the following lemma.}


\begin{lemma}\label{lemma:critical_region_properties}
	Critical region $\mathcal{C}_\mathcal{B} \subseteq \mathbb{R}_{+}^{\mt} \times \mathbb{R}_{++}^{\np}$ is convex  for all $\mathcal{B}$. Moreover, any critical region $\mathcal{C}_\mathcal{B}$ that has nonzero Lebesgue measure has nonempty interior.
\end{lemma}
\begin{proof}
	The first part of the statement comes from \cite{pistikopoulos2007multi}. Since the boundary of a convex set has measure zero \citep{lang1986note}, a critical region with nonzero measure must have nonempty interior.
\end{proof}

\subsubsection{Relation Between LMP and Power Line Congestion Pattern}\label{sec:relation_LMP_congestion}
\begin{lemma}[Relation Between LMP and Power Line Congestion Pattern]\label{lemma:relation_LMP_congestion}
	Let $(\mathcal{V}_\mathrm{P},\mathcal{E}_\mathrm{P})$ be a \mh{radial} power network and let $v_1,v_2 \in \mathcal{V}_\mathrm{P}$. If the unique path connecting $v_1$ and $v_2$ is uncongested, then $v_1$ and $v_2$ have the same locational marginal price.
\end{lemma}
\begin{proof}
	If $|\mathcal{V}_\mathrm{P}| = 1$ then it is trivially true. Consider the case when $|\mathcal{V}_\mathrm{P}| = 2$. If the power line connecting the only two buses is uncongested. \mh{It is known that our economic dispatch problem \eqref{eq:economic_dispatch} is equivalent to the below formulation using phase angles:} 
	\begin{subequations}\label{eq:DCOPF}
		\begin{align}
			\min_{g_i: i \in \mathcal{V}_\mathrm{P}} &\quad \mh{\frac{1}{2} \mathbf{g}^\top \mathbf{Q} \mathbf{g} + \bm \mu^\top \mathbf{g}},\\
			\mathrm{s.t} &\quad \lambda_i:g_i - d_i = \sum_{j \in \mathcal{N}_i} B_{ij}(\theta_i-\theta_j),\quad \forall i \in \mathcal{V}_\mathrm{P},\\
			&\quad \mu_{ij}^-, \mu_{ij}^+: -\bar f_{ij} \leq B_{ij}(\theta_i-\theta_j) \leq \bar f_{ij}, \quad \forall (i,j) \in \mathcal{E}_\mathrm{P},
		\end{align}
	\end{subequations}
	where $B_{ij}$ is the susceptance of line $(i,j)$, $\mathcal{N}_i$ is the buses that directly connected to bus $i$, and $\theta_i$ and $\theta_j$ are the voltage angles at buses $i$ and $j$, respectively.
	
	The Lagrangian of \eqref{eq:DCOPF} is:
	\begin{align}
\frac{1}{2} \mathbf{g}^\top \mathbf{Q} \mathbf{g} &+ \bm \mu^\top \mathbf{g} 
    + \sum_{i \in \mathcal{V}_\mathrm{P}} \lambda_i \left(\sum_{j \in \mathcal{N}_i} B_{ij}(\theta_i - \theta_j) - g_i + d_i\right)
    - \sum_{(i,j) \in \mathcal{E}_\mathrm{P}} \mu_{i,j}^{-}\!\left(\bar{f}_{ij} - B_{ij}(\theta_i-\theta_j)\right) \\
&\quad + \sum_{(i,j) \in \mathcal{E}_\mathrm{P}} \mu_{i,j}^{+}\!\left(B_{ij}(\theta_i-\theta_j) - \bar{f}_{ij}\right). \nonumber
\end{align}

	The stationarity condition with respect to the voltage angle $\theta_i$ implies that:
	\begin{equation}\label{eq:LMP_congestion_relation}
		\sum_{j \in \mathcal{N}_i} B_{ij}(\lambda_i-\lambda_j) = \sum_{j \in \mathcal{N}_i} B_{ij} (\mu_{ij}^+-\mu_{ij}^-),
	\end{equation}
	which establishes a relation between LMP and congestion conditions of power lines. When there are only two buses and one power line connecting them. The relation reduces to:
	\begin{equation}
		B_{12}(\lambda_1 - \lambda_2) = B_{12}(\mu_{12}^+-\mu_{12}^-).
	\end{equation}
	If the power line is uncongested, $\mu_{12}^+ = \mu_{12}^-$ by complementarity slackness. Therefore, $\lambda_1 = \lambda_2$ since $B_{12} \neq 0$. 
	
	Suppose the statement of the lemma holds for all power network with a tree topology and $K \geq 2$ buses. Consider a power system with a tree topology which has $K+1$ buses. Let $i,j$ be two buses. If the unique path connecting $i$ and $j$ has length $\leq K-1$, then $i$ and $j$ can be viewed as two buses in a power system with $K$ buses and the statement holds by our induction hypothesis. To see why this is true, pick an arbitrary leaf bus from the power system, say bus $r$. Let its parent bus be $r_\mathrm{pa}$. Then, the relation \eqref{eq:LMP_congestion_relation} for bus $r_\mathrm{pa}$ is:
	\begin{equation}
		B_{r_\mathrm{pa} r} (\lambda_{r_\mathrm{pa}} -\lambda_r) + B_{r'r_\mathrm{pa}}(\lambda_{r'}-\lambda_{r_\mathrm{pa}}) = B_{r_\mathrm{pa} r} (\mu^+_{r_\mathrm{pa} r} -\mu^{-}_{r_\mathrm{pa}r}) + B_{r'r_\mathrm{pa}}(\mu^+_{r'r_\mathrm{pa}}-\mu^-_{r'r_\mathrm{pa}}),
	\end{equation}
	where $r'$ is the parent bus of $r_\mathrm{pa}$ (if there is one). The relation \eqref{eq:LMP_congestion_relation} for bus $r$ is:
	\begin{equation}
		B_{r_\mathrm{pa}r}(\lambda_r-\lambda_{r_\mathrm{pa}}) = B_{r_\mathrm{pa}r}(\mu^-_{r_\mathrm{pa}r} - \mu^+_{r_\mathrm{\mathrm{pa}} r}).
	\end{equation}
	Therefore, the relation for bus $r_\mathrm{pa}$ is simply $B_{r'r_\mathrm{pa}}(\lambda_{r'}-\lambda_{r_\mathrm{pa}}) =  B_{r'r_\mathrm{pa}}(\mu^+_{r'r_\mathrm{pa}}-\mu^-_{r'r_\mathrm{pa}})$, acting as if bus $r$ does not appear in the power system. If $r_\mathrm{pa}$ has no parent (i.e. $r_\mathrm{pa}$ is the root bus), then similar argument still applies to show that relation between LMP and line congestion conditions at bus $r_\mathrm{pa}$ remains the same as if bus $r$ does not exist. Thus, we can ignore one leaf bus, which reduces the $(K+1)$-bus power system to a $K$-bus power system.
	
	What left is to show the statement continues to hold for the ignored leaf bus $r$ and an arbitrary bus $i$ in the power system. If the unique path connecting $i$ and $r$ is uncongested. The relation at bus $r$ tells that $r$ has the same LMP as its parent bus, we can apply the same argument to the parent bus again and again until we reach the conclusion that $i$ and $r$ have the same LMP.
	
	By mathematical induction, for any power system with a tree topology and finitely many power buses, any two buses whose connecting path is uncongested have the same LMP.
\end{proof}

\subsubsection{Additional Discussion of BP Relations}\label{sec:discussion_corollary1}
Corollary \ref{cor:BP_relations}-(a) and (b) are implied by \eqref{eq:tt_fully_congested} and \eqref{eq:pt_fully_congested}. If type P-T BP occurs when perturbing line $(i_r,i)$, $i \notin \mathcal{I}^\mathrm{act}$, \mh{with power flows from $i_r$ to $i$}, \eqref{eq:pt_fully_congested} implies $\psi_r < 0$, which further implies that \eqref{eq:tt_fully_congested} does not hold for route $r$. Similarly, one can argue (b) is true.
	
	Corollary \ref{cor:BP_relations}-(c) and (d) are implied by \eqref{eq:tp_fully_congested} and \eqref{eq:pp_fully_congested}. If type P-P BP does not occur when perturbing $(i_r,i)$ with $i\notin \mathcal{I}^\mathrm{act}$, \mh{and power flow $i_r \to i$},  \eqref{eq:pp_fully_congested} implies \mh{$\sum_{i' \in \mathcal{I}^\mathrm{act}} \tilde{\omega}_{i'} (\lambda_{i_r} - \lambda_{i'}) > 0$}. Theorem \ref{thm:fully_congested}-(b2) then concludes type T-P BP occurs. Similarly, one can argue (d) is true.
	
	The appearances of indicator functions in \eqref{eq:pt_fully_congested} and \eqref{eq:pp_fully_congested} are critical as if a line $(i,i'), i\notin \mathcal{I}^\mathrm{act}, i' \notin \mathcal{I}^\mathrm{act}$, \mh{with power flow $i \to i'$}, is perturbed, the change in prices $\lambda_i$ and $\lambda_{i'}$ is not perceived by travelers, and thus neither traffic flow nor load relocation is induced.

\subsubsection{Extension of Corollary \ref{cor:BP_relations} under $\hat{\bm \beta}=\kappa \mathbf{1}$ for some $\kappa \geq 0$} \label{sec:extension_corollary1}
Corollary \ref{cor:BP_relations} does not show any relation of \mh{Type T-T (P-T) and T-P (P-P) BPs}. If a further assumption \mh{$\hat{\bm \beta}=\kappa \mathbf{1}$ for some $\kappa \geq 0$}  is imposed, Type T-T (\mh{P-T}) and T-P (\mh{P-P}) BPs are related. We present the relations in the following corollary.	

	\begin{corollary}[\mh{Relations of BPs with Homogeneous $\hat{\bm \beta}$}]\label{cor:further_BP_relations}
		Under \mh{\textbf{A}1 and \textbf{A}2}, if we in addition assume $\hat{\bm \beta} = \kappa \mathbf{1}$ for some $\kappa \geq 0$, then, the following statements hold\mh{:}
		\begin{enumerate}
			\item[(a)] The occurrence of type T-T BP when perturbing a route $r \in \mathcal{R}^\mathrm{act}$ implies the non-occurrence of type T-P BP when perturbing route $r$;
			\item[(b)] Let bus $i \notin \mathcal{I}^\mathrm{act}$ be arbitrary. The following statements hold:
			\begin{enumerate}
				\item[(b1)] The non-occurrence of type P-T BP when perturbing a line $(i_r,i)$ (if exists) with power flow $i_r \to i$ implies the occurrence of type P-P BP when perturbing $(i_r,i)$;
				\item[(b2)] The occurrence of type P-T BP when perturbing a line $(i,i_r)$  (if exists) with power flow $i \to i_r$ implies the non-occurrence of type P-P BP when perturbing $(i,i_r)$;
				\item[(b3)] The occurrence of type P-T BP when perturbing a line $(i_r,i_{r'})$ (if exists) with power flow $i_r \to i_{r'}$ implies the non-occurrence of type P-P BP when perturbing $(i_r,i_{r'})$.
			\end{enumerate}
		\end{enumerate}
	\end{corollary} 
	
If $\hat{\bm \beta} = \mh{\kappa}\mathbf{1}$ for some $\kappa \geq 0$, $\psi_r = 2\omega_r \sum_{r'=1}^R \tilde{\omega}_{r'} (\lambda_{i_{r'}} - \lambda_{i_r})$. Type \mh{T-T} BP occurs \mh{when perturbing route $r$} if and only if $\psi_\mh{r} > x_\mh{r}^\star > 0$, \mh{which is equivalent to $\sum_{r'=1}^R \tilde{\omega}_{r'} (\lambda_{i_{r'}} - \lambda_{i_r}) > 0$ and further equivalent to the non-occurrence of type T-P BP when perturbing route $r$.} Therefore, Corollary \ref{cor:further_BP_relations}-(a) holds. 

The non-occurrence of type P-T BP when perturbing a line $(i_r,i)$ \mh{with power flow $i_r \to i$} implies $\psi_{r_i} < 0$, which is equivalent to $\varsigma_{i} > 0$. Therefore, type P-P BP occurs when perturbing line $(i_r,i)$ and Corollary \ref{cor:further_BP_relations}-(b1) holds. 

The occurrence of type P-T BP when perturbing a line $(i,i_r)$ with power flow $i \to i_r$ implies $\psi_{r_i} > 0$, which is equivalent to $\varsigma_{i} > 0$. Therefore, type P-P does not occur when perturbing line $(i,i_r)$ and Corollary \ref{cor:further_BP_relations}-(b2) holds. The same argument applies if line $(i_r,i_{r'})$ with power flow $i_r \to i_{r'}$ is perturbed.

\subsubsection{Radial Power Network with Arbitrary Congestion Pattern}\label{sec:radial_arbitrary_congestion}

\mh{
In \S\ref{subsec:uncongested} and \ref{subsec:completely_congested}, BP characterizations are provided for two extreme cases. In this section, we aim to relax previous assumptions and generalize BP characterizations. For a coupled system with a radial power network, if we group buses in the power network connected by uncongested lines, we form a collection of \textit{subnetworks} $\{(\mathcal{V}_\mathrm{P}^k,\mathcal{E}_\mathrm{P}^k)\}_{k=1}^K$, where $K$ denotes the number of bus groups, $\mathcal{V}_\mathrm{P}^k$ is the $k$-th group, and $\mathcal{E}_\mathrm{P}^k$ is the set of lines connecting buses in $\mathcal{V}_\mathrm{P}^k$. If we view each bus group as a single bus, then the resulting reduced power network is radial and fully congested, and \textbf{A}1 automatically holds. Moreover, \mh{each subnetwork is associated with a set (bundle) of routes with chargers connecting to buses in this subnetwork.}
}

\begin{definition}[Subnetwork]\label{def:subnetwork}
	A network $(\mathcal{V}_\mathrm{P}^k,\mathcal{E}_\mathrm{P}^k)$ is a \textit{subnetwork} if for any $i_1,i_2 \in \mathcal{V}_\mathrm{P}^k$, buses $i_1$ and $i_2$ are connected by a (unique) uncongested path of power lines in $\mathcal{E}_\mathrm{P}^k$.
\end{definition} 

\begin{definition}[Route Bundle]\label{def:route_bundle}
	The collection of routes $\mathcal{P}^k\triangleq\left\{r\in [\nr]:x_r^\star>0;\exists i \in \mathcal{V}_\mathrm{P}^k, (\Acr)^\top_r \Acb\mathbf{e}_i = 1\right\},$ is said to be the \textit{route bundle} associated with subnetwork $(\mathcal{V}_\mathrm{P}^k,\mathcal{E}_\mathrm{P}^k)$, where $\mathbf{e}_i$ is the $i$-th canonical basis vector. Denote by $\mh{\hat{x}_k^\star}\triangleq\sum_{r \in \mathcal{P}^k} x_r^\star$ the \textit{aggregate traffic flow} of routes in $\mathcal{P}^k$.
\end{definition}

From the definitions, any buses in the same subnetwork share the same price, and any routes in the same route bundle have the same travel cost (see Lemma \ref{lemma:properties_subnetworks_route_bundles} in \hyperref[sec:supp_reduced_transportation_network]{Appendix~\S\ref{sec:supp_reduced_transportation_network}}). We without loss of generality assume $\mathcal{P}^k \neq \emptyset, \forall k \in [K]$. Otherwise $\mh{\hat{x}_k^\star} = 0$ and it only introduces minor book-keeping burdens. We assume  columns of $\Alr$ and $\Acr$ are reduced to contain only those corresponding to routes in $\{\mathcal{P}^k\}_{k=1}^K$, and are permuted according to the order of $\{\mathcal{P}^k\}_{k=1}^K$. Similarly, we assume $\mathbf{Q}$ and $\bm \mu$ are permuted according to the order of $\{(\mathcal{V}_\mathrm{P}^k,\mathcal{E}_\mathrm{P}^k)\}_{k=1}^K$. We denote by $\mathbf{g}_k$ the generation of buses in $\mathcal{V}_\mathrm{P}^k$, and $\hat{g}_k\triangleq\mathbf{1}^\top \mathbf{g}_k$ the aggregate generation. Furthermore, let $\mathbf{Q}_k$ and $\bm \mu_k$ be the blocks of $\mathbf{Q}$ and $\bm \mu$ corresponding to $\mathcal{V}_\mathrm{P}^k$, respectively. \mh{Since any line connecting different subnetworks is congested, we introduce a selection matrix $\widehat{\mathbf{S}} \in \{-1,0,1\}^{K \times \mp}$ to determine the net power flow out of subnetworks by the relation $\hat{\mathbf{f}} = \widehat{\mathbf{S}} \bar{\mathbf{f}}$, where the $k$-th entry $\hat{f}_k$ of $\hat{\mathbf{f}}$ is the net power flow out of the $k$-th subnetwork.} \mh{The next lemma justifies why each route bundle/subnetwork can be viewed as if it was a single route/bus, respectively.} 

\begin{lemma}[\mh{Aggregated Costs}]\label{lemma:common_travel_cost_affine} 
\mh{Let $\mathbf{x}^\star$ and $\mathbf{g}^\star$ be the equilibrium traffic flow and generation profile of the system of interest at GUE. Then, the following statements hold for all $k \in [K]$:}
\begin{enumerate}
	\item The common travel cost $\hat{c}_k$ incurred by traveling through any route in $\mathcal{P}^k$ is \mh{related to} $\hat{\mathbf{x}}^\star$. Specifically, there exist $\hat{\bm \alpha}_k \in \mathbb{R}^M$ and $\hat{\beta}_k \in \mathbb{R}$ such that 
    \begin{align}\label{eq:expression_c_m}
        \hat{c}_k = \hat{\bm \alpha}_k^\top \hat{\mathbf{x}}^\star + \hat{\beta}_k;
    \end{align}
	\item The total generation cost of buses in $\mathcal{V}_\mathrm{P}^k$ is $\frac{1}{2}\widehat{Q}_k (\hat{g}_k^\star)^2 + \hat{\mu}_k \hat{g}_k^\star$, \mh{up to a constant}, where $\widehat{Q}_k \triangleq (\mathbf{1}^\top \mathbf{Q}_k^{-1} \mathbf{1})^{-1}$, $\hat{\mu}_k \triangleq \mathbf{1}^\top \mathbf{Q}_k^{-1}\bm \mu_k \widehat{Q}_k$, \mh{and $\hat{g}_k^\star\triangleq\mathbf{1}^\top \mathbf{g}_k^\star$.}
\end{enumerate}
\end{lemma}


With Lemma \ref{lemma:common_travel_cost_affine}, the coupled system can be \mh{aggregated to become} a smaller system where subnetworks become buses with cost coefficients $\{(\widehat{Q}_k, \hat{\mu}_k)\}_{k=1}^K$, and route bundles become routes with travel cost coefficients $\{(\hat{\bm \alpha}_k, \hat{\beta}_k)\}_{k=1}^K$, which we \mh{call} \textit{aggregated system}. \mh{Given a GUE $(\mathbf{x}^\star, \bm \lambda^\star)$, we call $(\hat{\mathbf{x}}^\star, \hat{\bm \lambda}^\star)$ the \textit{aggregate GUE} of the aggregated system.} 

Now, perturbing $\alpha_{\ell_\mathrm{T}}$ affects $\Phi_s,s\in\{\mathrm{T},\mathrm{P}\}$ through $\widehat{\bm \alpha}\triangleq[\hat{\bm \alpha}_1,...,\hat{\bm \alpha}_K]^\top \in \mathbb{R}^{K \times K}$ and $\hat{\bm \beta} \triangleq [\hat{\beta}_1,...,\hat{\beta}_K]^\top \in \mathbb{R}^K$. Then, by the chain rule, 
\begin{align}\label{eq:chain_rule}
    \frac{\partial \Phi_s}{\partial \alpha_{\ell_\mathrm{T}}} &=\sum_{k,k' \in [K]} \frac{\partial \Phi_s}{\partial \hat{\alpha}_{k,k'}}\frac{\partial \hat{\alpha}_{k,k'}}{\partial \alpha_{\ell_\mathrm{T}}} + \sum_{k \in [K]} \frac{\partial \Phi_s}{\partial \hat{\beta}_k} \frac{\partial \hat{\beta}_k}{\partial \alpha_{\ell_\mathrm{T}}}.
\end{align}
The terms $\partial \hat{\alpha}_{k,k'}/\partial \alpha_{\ell_\mathrm{T}}$ and $\partial \hat{\beta}_k/\partial \alpha_{\ell_\mathrm{T}}$ measure the sensitivities of parameters $\hat{\alpha}_{k,k'}$ and $\hat{\beta}_k$ of the \mh{aggregated} system with respect to $\alpha_{\ell_\mathrm{T}}$ of the original network. They are unavailable if only the \mh{aggregated system} is given. However, the other two terms $\partial \Phi_s/\partial \hat{\alpha}_{k,k'}$ and $\partial \Phi_s/\partial \hat{\beta}_k$ can be computed based solely on the \mh{aggregated system}, and are meaningful as they resemble $\partial \Phi_s/\partial \alpha_{\ell_\mathrm{T}}$, whose sign ties with TBP. We thus introduce the concept of \textit{\mh{aggregated} TBP}.
\begin{definition}[\mh{Aggregated} TBP]\label{def:reduced_transportation_BP}
    Given an \mh{aggregated} system with parameters $\{\hat{\bm \alpha}_k, \hat{\beta}_k\}_{k=1}^K$,  we say \mh{aggregated} TBP (\mh{A}TBP) occurs (with respect to (wrt) $\mathcal{P}^k$) if either $\partial \Phi_s/\partial \hat{\alpha}_{k,k'} < 0$ or $\partial \Phi_s/\partial \hat{\beta}_k < 0$ holds for some $k,k' \in [K]$. Similarly, \mh{we} use the notion of type T-T (\mh{T-P}) \mh{A}TBP if $s = \mathrm{T}$ ($s = \mathrm{P}$).
\end{definition}

\mh{Inspired by \textbf{A}\mh{2}, we start with a simple case allowing link sharing for routes in the same route bundle, but ban route sharing for routes in different route bundles.}
\begin{enumerate}
	\item[\mh{\textbf{A}2'}] No two route bundles share links.
\end{enumerate}

\mh{\textbf{A}2'} simplifies analysis since perturbing link $\ell_\mathrm{T}$ directly affects at most one route bundle, while other route bundles are only indirectly affected through traffic flow relocations. Moreover, \mh{under \textbf{A}2',} $\widehat{\bm \alpha}$ and $\hat{\bm \beta}$ possess good structures. In particular, $\widehat{\bm \alpha}$ is diagonal, and if $\ell_\mathrm{T}$ is contained (only) in $\mathcal{P}^k$, $\partial \hat{\alpha}_{k_1,k_2}/\partial \alpha_{\ell_\mathrm{T}} = 0$ except that $k_1 = k_2 = k$, and $\partial \hat{\beta}_{k'}/\partial \alpha_{\ell_\mathrm{T}} = 0$ except that $k' = k$. The results are formalized in Lemma \ref{lemma:properties_of_independent_networks} in \hyperref[sec:supp_reduced_transportation_network]{Appendix~\S\ref{sec:supp_reduced_transportation_network}}.

The following theorem establishes the necessary and sufficient conditions of \mh{A}TBP under \mh{\textbf{A}2'}.

\begin{theorem}[\mh{Necessary and Sufficient Conditions of ATBP}]\label{thm:theorem_independent_network}
    Suppose the system of \mh{satisfies \textbf{A}2'}, and its \mh{aggregated} system has parameters $\{\hat{\bm \alpha}_k,\hat{\beta}_k\}_{k=1}^K$ and aggregate GUE $(\hat{\mathbf{x}}, \hat{\bm \lambda})$. Then, the following statements hold:
    \begin{enumerate}
		\item[(a)] If $K=1$, \mh{A}TBP never occurs;
        \item[(b)] If $K \geq 2$, 
        \begin{enumerate}
            \item[(b1)] \mh{type T-T \mh{A}TBP occurs wrt $\mathcal{P}^k$ if and only if}
            \begin{align}
            	\hat{x}_k < \hat{\psi}_k\triangleq\omega_k \sum_{j=1}^K \tilde{\omega}_{j} (2\hat{\alpha}_{k,k}\hat{x}_k + \hat{\beta}_k -\hat{\alpha}_{j,j}\hat{x}_{j} + \hat{\beta}_{j}), \nonumber
            \end{align}
            \mh{where we abuse the notation $\omega_k\triangleq1/(\hat{\alpha}_{k,k} + \widehat{Q}_k)$, and $\tilde{\omega}_k$ is similarly defined as in Theorem \ref{thm:fully_congested}}-(b1);
            \item[(b2)] \mh{type T-P ATBP occurs wrt $\mathcal{P}^k$ if and only if}
            \begin{align}
            	\sum_{k'=1}^K \tilde{\omega}_k \left(\hat{\lambda}_k - \hat{\lambda}_{k'}\right) > 0;
            \end{align}
        \end{enumerate}
        \item[(c)] If \mh{type T-T (T-P)} \mh{A}TBP does not occur, and it holds for all $k$ and for any $\ell_\mathrm{T}$ used by $\mathcal{P}^k$, $\partial \hat{\alpha}_{k,k}/\partial \alpha_{\ell_\mathrm{T}} \geq 0$ and $\partial \hat{\beta}_k/\partial \alpha_{\ell_\mathrm{T}}\geq 0$, then \mh{type T-T (T-P)} BP \mh{does not} occur.
    \end{enumerate}
\end{theorem}

Theorem \ref{thm:theorem_independent_network}-(b1) and (b2) have almost identical forms to Theorem \ref{thm:fully_congested}-(b1) and (b2), except that $\hat{\alpha}_r$ and $\lambda_{i_r}$ are replaced by $\hat{\alpha}_{k,k}$ and $\hat{\lambda}_k$, respectively. Theorem \ref{thm:theorem_independent_network}-(c) is a direct consequence of \eqref{eq:chain_rule}. Moreover, we have the following theorem. 

\begin{theorem}[Necessary and Sufficient Conditions of \mh{BP, Radial Power Network Case}]\label{thm:NS_conditions_independent_network}
	\mh{Under \textbf{A}2'}, the following hold:
	\begin{enumerate}
		\item[(a)] If type T-T (T-P) BP occurs, there exists $k \in [K]$ such that exactly one of the following statements is true.
    \begin{enumerate}
        \item[(a1)] Type T-T (T-P) \mh{A}TBP occurs wrt $\mathcal{P}^k$;
        \item[(a2)] Classical transportation BP occurs in $\mathcal{P}^k$ when viewed as a transportation network with fixed net flow $\hat{x}_k$.
    \end{enumerate}
    The reverse direction holds, if the false one of (a1) and (a2) corresponds to a strict non-occurrence, i.e., the corresponding derivative is strictly positive;
    \item[(b)] Type P-T BP occurs if and only if there exists $\ell_\mathrm{P}=(i,i') \in [K] \times [K]$ with \mh{power flow $i \to i'$} such that
  	\begin{align}\label{eq:hat_PT}
  		\widehat{Q}_{i} \hat{\psi}_{i} - \widehat{Q}_{i} \hat{\psi}_{i'} < 0;
  	\end{align}
  	\item[(c)] Type P-P BP occurs if there exists $\ell_\mathrm{P}=(i,i') \in [K] \times [K]$ with \mh{power flow $i \to i'$} such that
  	\begin{align}\label{eq:hat_PP}
  		\hat{\varsigma}_{i} - \hat{\varsigma}_{i'} > 0,
  	\end{align}
  	where $\hat{\varsigma}_i\triangleq\hat{\lambda}_i^\star - \rho^2 \widehat{Q}_i \omega_i \sum_{k=1}^K \tilde{\omega}_k (\hat{\lambda}_i^\star - \hat{\lambda}_k^\star)$.
	\end{enumerate}
\end{theorem}

Theorem \ref{thm:NS_conditions_independent_network}-(a) has the form of the \textit{theorem of alternative}. Its reverse direction provides an interesting viewpoint to understand TBP. Intuitively, exactly one of (a1) and (a2) holds means either the \mh{aggregated} system or the transportation network formed by $\mathcal{P}^k$ is \textit{bad}. The effect of perturbing link $\ell_\mathrm{T}$ first affects the route bundle containing link $\ell_\mathrm{T}$, and then gets propagated to other route bundles through traffic flow relocation. If exactly one of the effects is BP-promoting, and the other effect is not inefficacious (i.e., strict non-occurrence), then TBP occurs. However, if both effects are BP-promoting, the effects get cancelled off, and TBP does not occur.

Theorem \ref{thm:NS_conditions_independent_network}-(b) and (c) generalize Theorem \ref{thm:fully_congested}-(b3) and (b4). All previous interpretations are valid, with route and bus understood as route bundle and subnetwork, respectively.

We provide examples demonstrating the usage of Theorem \ref{thm:NS_conditions_independent_network}-(a), the discussion on the extension of Corollary \ref{cor:BP_relations}, and generalizations without \mh{\textbf{A}2'} in \hyperref[sec:supp_NS_conditions_proofs]{Appendix~\S\ref{sec:supp_NS_conditions_proofs}}.

\subsubsection{\mh{Properties of Subnetwork and Route Bundle}}\label{sec:supp_reduced_transportation_network}
\mh{In this section, we provide proofs of results we obtained for coupled systems with general radial networks. Before we proceed, we first prove some properties:}
\begin{lemma}[Properties of Subnetworks and Route Bundles]\label{lemma:properties_subnetworks_route_bundles}
	Let $(\mathbf{x}^\star, \bm \lambda^\star)$ be a GUE, $\{(\mathcal{V}_\mathrm{P}^k,\mathcal{E}_\mathrm{P}^k)\}_{k=1}^K$ be subnetworks of the power network $(\mathcal{V}_\mathrm{P},\mathcal{E}_\mathrm{P})$, and $\{\mathcal{P}^k\}_{k=1}^K$ the corresponding route bundles. Then, the following statements hold\mh{:}
	\begin{enumerate}
		\item[(a)] \mh{F}or any $k$ and any $i_1, i_2 \in \mathcal{V}_\mathrm{P}^k$, $\lambda_{i_1}^\star = \lambda_{i_2}^\star \triangleq \hat \lambda_k \geq 0$;
		\item[(b)] \mh{A}ny routes in $\mathcal{P}^k$ have the same travel cost $\hat{c}_k \triangleq (\Alr)_r^\top \left[\mathrm{diag}(\bm \alpha)\Alr \mathbf{x}^\star + \bm \beta\right] > 0, \forall r \in \mathcal{P}^k, \forall k \in [K]$.
	\end{enumerate}
\end{lemma}
The proof of Lemma \ref{lemma:properties_subnetworks_route_bundles} is straightforward. Lemma \ref{lemma:properties_subnetworks_route_bundles}-(a) directly follows from Lemma \ref{lemma:relation_LMP_congestion}, and Lemma \ref{lemma:properties_subnetworks_route_bundles}-(b) is a direct consequence of Lemma \ref{lemma:properties_subnetworks_route_bundles}-(a) since any routes in the same route bundle have the same charging cost. 

\begin{lemma}[\mh{Properties of $\widehat{\bm \alpha}$ and $\hat{\bm \beta}$}]\label{lemma:properties_of_independent_networks}
    \mh{Under \textbf{A}2'}, the following statements hold:
    \begin{enumerate}
        \item[(a)] $\widehat{\bm \alpha}$ is diagonal, i.e. $\hat{\alpha}_{k,k'} = 0$ for all $k \neq k'$;
        \item[(b)] Provided that $\ell_\mathrm{T}$ is contained (only) in $\mathcal{P}^k$, $\partial \hat{\alpha}_{k_1,k_2}/\partial \alpha_{\ell_\mathrm{T}} = 0$ except that $k_1 = k_2 = k$;
        \item[(c)] Provided that $\ell_\mathrm{T}$ is contained (only) in $\mathcal{P}^k$, $\partial \hat{\beta}_{k'}/\partial \alpha_{\ell_\mathrm{T}} = 0$ except that $k' = k$.
    \end{enumerate}
\end{lemma}
\begin{proof}
    (a) We present the proof by analyzing the structure of $\hat{\alpha}_k$ for each $k$. Recall that $\hat{\alpha}_k \propto \widehat{\mathbf{C}}^\top (\Alr)^\top \mathrm{diag}(\bm \alpha) \Alr_k \mathbf{1}$. The $\ell$-th entry of the vector $\mathrm{diag}(\bm \alpha) \Alr_k \mathbf{1}$ is $n_\ell^k \alpha_\ell$, where $n_\ell^k$ is defined to be the number of times link $\ell$ is used by routes in $\mathcal{P}^k$. Note that the entry is $0$ if link $\ell$ does not appear in $\mathcal{P}^k$. Then, the $r$-th entry of $(\Alr)^\top \mathrm{diag}(\bm \alpha) \Alr_k \mathbf{1}$ is $\sum_{\ell:\Alr_{\ell,r}=1} n_\ell^k \alpha_\ell$. Note that only the $|\mathcal{P}^{k-1}|+1$-th to $|\mathcal{P}^k|$-th entries are nonzero, since other entries correspond to routes in other route bundles, and by \textbf{A}2', those routes do not contain any links used by $\mathcal{P}^k$. By Remark \ref{remark:structure_of_hat_C}, $\widehat{\mathbf{C}} \in \mathbb{R}^{R \times K}$ is block diagonal with $M$ blocks, and its $k$-th block is a $|\mathcal{P}^k|$-d vector. Therefore, the multiplication of $\widehat{\mathbf{C}}^\top$ and $(\Alr)^\top \mathrm{diag}(\bm \alpha) \Alr_k \mathbf{1}$ is a vector whose $k$-th entry is nonzero, while all other entries are $0$'s. Since this holds for all $k \in [K]$, we conclude that $\hat{\alpha}_{k,s} = 0$ for all $k \neq s$.

    (b) Let $\ell_\mathrm{T}$ be used only by $\mathcal{P}^k$. It can be seen from the analysis in Part (a) of the proof that $\hat{\alpha}_{k,s}$ does not depend on $\alpha_{\ell_\mathrm{T}}$ for all $s \neq k$. Moreover, for any $k' \neq s$, we know $\mathcal{P}^{k'}$ does not contain $\ell_\mathrm{T}$ and thus $\partial{\hat{\alpha}_{k',s}}/\partial \alpha_{\ell_\mathrm{T}} = 0$ for all $s \in [K]$.

    (c) According to the analyis of the structure of $(\Alr)^\top \mathrm{diag}(\bm \alpha) \Alr_k \mathbf{1}$ in Part (a), it can be similarly concluded that only $\hat{\beta}_k$ depends on $\alpha_{\ell_\mathrm{T}}$, provided that $\ell_\mathrm{T}$ is contained only in $\mathcal{P}^k$. Therefore, $\partial \hat{\beta}_{k'}/\partial \alpha_{\ell_\mathrm{T}} = 0$ for all $k' \neq k$.
\end{proof}



The proof of Lemma \ref{lemma:common_travel_cost_affine} is based on the following lemma, which states that there exists and one-to-one correspondence between GUE and its aggregate GUE.

\begin{lemma}[Relation between Flow and Aggregate Flow]\label{lemma:relation}
    Let $\mathbf{x} \in \mathbb{R}^{R}$ be the flow of a GUE, and the corresponding aggregate flow be $\hat{\mathbf{x}} \in \mathbb{R}^K$. Then,
    \begin{enumerate}
        \item[(a)] the flow $\mathbf{x}$ and the aggregate flow $\hat{\mathbf{x}}$ satisfy the following equations
        \begin{subequations}
            \begin{align}
            c^\mathrm{tr}_{i+\sum_{j=1}^k |\mathcal{P}^j|} &= c^\mathrm{tr}_{i+1+\sum_{j=1}^k |\mathcal{P}^j|}, \quad \mh{k = 0,...,K-1}, \quad i = 1,...,|\mathcal{P}^k|-1; \label{eq:equations_for_b_1}\\
            \mathbf{1}_k^\top \mathbf{x} &= \hat x_k, \quad k = 1,...,K,\label{eq:equations_for_b_2}
        \end{align}
        \mh{where $\mathbf{1}_k \in \{0,1\}^R$ is the vector whose entries corresponding to routes in $\mathcal{P}^k$ are $1$'s;}
        \end{subequations}
        \item[(b)] there exist a matrix $\widehat{\mathbf{C}} \in \mathbb{R}^{R \times K}$ and a vector $\hat{\mathbf{q}} \in \mathbb{R}^{R}$ such that
        \begin{align}\label{eq:relation_flow_aggregate_flow}
            \mathbf{x} = \widehat{\mathbf{C}} \hat{\mathbf{x}} + \hat{\mathbf{q}}.
        \end{align}
    \end{enumerate}
\end{lemma}

\begin{proof}
    (a) Since $\mathbf{x}$ is an equilibrium flow, the travel costs of routes in the same route bundle are the same according to Lemma \ref{lemma:properties_subnetworks_route_bundles}-(c), which is equivalent to \eqref{eq:equations_for_b_1}. Since $\hat{\mathbf{x}}$ is defined to be the corresponding aggregate flow, \eqref{eq:equations_for_b_2} automatically holds.

    (b) We can write the system of equations \eqref{eq:equations_for_b_1}-\eqref{eq:equations_for_b_2} in the following matrix form
    \begin{align}
        \mathbf{M}\mathbf{x} = \begin{bmatrix}
        \mathbf{b}\\
        \hat{\mathbf{x}}
        \end{bmatrix} \quad \Rightarrow \quad \mathbf{x} = \mathbf{M}^{-1} \begin{bmatrix}
        \mathbf{b}\\
        \hat{\mathbf{x}}
        \end{bmatrix} = \mh{\begin{bmatrix}
\mathbf{M}_{\leftarrow}^{-1} & \mathbf{M}_{\rightarrow}^{-1}	
\end{bmatrix} \begin{bmatrix}
        \mathbf{b}\\
        \hat{\mathbf{x}}
        \end{bmatrix} 
} = \mathbf{M}^{-1}_\rightarrow \hat{\mathbf{x}} + \mathbf{M}^{-1}_\leftarrow \mathbf{b},
    \end{align}
    where $\mathbf{M}\in \mathbb{R}^{R \times R}$ is the coefficient matrix, $\mathbf{b} \in \mathbb{R}^{R - K}$ depends only on $\bm \beta$, $\mathbf{M}^{-1}_\leftarrow$ is the left $R \times (R-K)$ block, and $\mathbf{M}^{-1}_\rightarrow$ is the right $R \times K$ block of $\mathbf{M}^{-1}$. In general, $\mathbf{M}$, as a function of $\bm \alpha$ is not always invertible over the parameter space. But it is invertible \mh{with the given parameter $(\bm \alpha, \bar{\mathbf{f}})$}, since in this case \mh{if a given $\hat{\mathbf{x}}$ corresponds to two different GUEs, it contradicts with the uniqueness of GUE}. In the lemma statement, we define $\widehat{\mathbf{C}}:=\mathbf{M}^{-1}_\rightarrow$ and $\mh{\hat{\mathbf{q}}}:=\mathbf{M}^{-1}_\leftarrow \mathbf{b}$.
\end{proof}
\begin{remark}\label{remark:structure_of_hat_C}
    In fact, the coefficient matrix $\mathbf{M}$ and $\mathbf{b}$ have the following analytical expression
    \begin{align}
        &\mathbf{M} = \begin{bmatrix}
            \mathbf{E}\widetilde{\mathbf{A}}\\
            \mathbf{1}_1^\top\\
            \vdots\\
            \mathbf{1}_K^\top
        \end{bmatrix}, \quad \mathbf{b} = -\mathbf{E}(\Alr)^\top \bm \beta,
    \end{align}
    where \mh{$\widetilde{\mathbf{A}}:=(\Alr)^\top \mathrm{diag}(\bm \alpha) \Alr$} and the matrix $\mathbf{E}$ is defined to be a block digonal matrix
    \begin{align}
        \mathbf{E} := \begin{bmatrix}
            \mathbf{E}_1 & & \\
            & \ddots & \\
            & & \mathbf{E}_K
        \end{bmatrix}, \quad \mathrm{with} \quad \mathbf{E}_k:=\begin{bmatrix}
            1 & -1 & 0 & \cdots & 0 & 0\\
            0 & 1 & -1 & \cdots & 0 & 0\\
            \vdots & \vdots & \ddots & \ddots & \vdots & \vdots \\
            0 & 0 & \cdots & \cdots & 1 & -1
        \end{bmatrix} \in \mathbb{R}^{(|\mathcal{P}^k|-1) \times |\mathcal{P}^k|}.
    \end{align}
    It can be shown that
    \begin{align}
        \widehat{\mathbf{C}} = \mathbf{M}_\rightarrow^{-1} = \left(\begin{bmatrix}
            \mathbf{1}\mathbf{1}^\top & & \\
            & \ddots & \\
            & & \mathbf{1}\mathbf{1}^\top
        \end{bmatrix} + \widetilde{\mathbf{A}}\begin{bmatrix}
            \mathbf{E}_1^\top \mathbf{E}_1 & & \\
            & \ddots & \\
            & & \mathbf{E}_K^\top \mathbf{E}_K
        \end{bmatrix}\widetilde{\mathbf{A}}\right)^\dagger \begin{bmatrix}
            \mathbf{1}_1 & \mathbf{1}_2 & \cdots & \mathbf{1}_K
        \end{bmatrix}, 
    \end{align}
    where the block matrices in the parenthesis both have $K$ blocks and the $k$-th block has dimension $|\mathcal{P}^k| \times |\mathcal{P}^k|$.

    One important property of $\widehat{\mathbf{C}}$ is when the underlying network satisfies \textbf{A}2', $\widehat{\mathbf{C}}$ becomes a block diagonal matrix with $K$ blockes, and its $k$-th block has dimension $|\mathcal{P}^k|$. The reason behind is in this case $\widetilde{\mathbf{A}}$ is a block diagonal matrix, and the Moore-Penrose inverse of a block diagonal matrix is block diagonal as well.
\end{remark}

\begin{proof}[Proof of Lemma \ref{lemma:common_travel_cost_affine}-(a)]
    Since the travel costs of routes in the same route bundle $\mathcal{P}^k$ are the same, the common travel cost $\hat{c}_k$ can be computed as the average of travel costs of routes in $\mathcal{P}^k$, i.e.,
    \begin{align}
        \hat{c}_k &= \frac{1}{|\mathcal{P}^k|} \mathbf{1}^\top (\Alr_k)^\top \left( \mathrm{diag}(\bm \alpha) \Alr \mathbf{x} + \bm \beta \right)\\
        &=\frac{1}{|\mathcal{P}^k|} \left(\widehat{\mathbf{C}}^\top (\Alr)^\top \mathrm{diag}(\bm \alpha) \Alr_k \mathbf{1} \right)^\top \hat{\mathbf{x}} + \frac{(\Alr_k \mathbf{1})^\top}{|\mathcal{P}^k|} \left( \mathrm{diag}(\bm \alpha) \Alr \mh{\hat{\mathbf{q}}} + \bm \beta \right) := \hat{\bm \alpha}_k^\top \hat{\mathbf{x}} + \hat{\beta}_k,
    \end{align}
    where the second line follows from \eqref{eq:relation_flow_aggregate_flow}.
\end{proof}

\begin{proof}[Proof of Lemma \ref{lemma:common_travel_cost_affine}-(b)]
    Let $\mathbf{g}^\star$ be the generation of GUE, it satisfies
    \begin{align}
        \mathbf{Q}_k \mathbf{g}_k + \bm \mu_k = \hat{\lambda}_k \mathbf{1}, \qquad \mathbf{1}^\top \mathbf{g}_k = \rho \hat{x}_k + \hat{f}_k, \quad \forall k \in [K].
    \end{align}
    Solving the system of equations gives
    \begin{subequations}
    	\begin{align}
        \mathbf{g}_k^\star &= \frac{\rho \mathbf{Q}_k^{-1} \mathbf{1}}{\mathbf{1}^\top \mathbf{Q}_k^{-1} \mathbf{1}} \hat{x}_k + \frac{\mathbf{Q}_k^{-1} \mathbf{1} \hat{f}_k}{\mathbf{1}^\top \mathbf{Q}_k^{-1} \mathbf{1}} - \mathbf{Q}_k^{-1}\left(\mathbf{Q}_k - \frac{\mathbf{1}\mathbf{1}^\top}{\mathbf{1}^\top \mathbf{Q}_k^{-1} \mathbf{1}}\right)\mathbf{Q}_k^{-1}\bm \mu_k\\
        &=\widehat{Q}_k \hat{g}_k \mathbf{Q}_k^{-1} \mathbf{1} - \mathbf{Q}_k^{-1}\bm \mu_k + \widehat{Q}_k \mathbf{Q}_k^{-1} \mathbf{1}\mathbf{1}^\top \mathbf{Q}_k^{-1} \bm \mu_k,\quad \forall k \in [K].
    \end{align}
    \end{subequations}
    
    It is not hard to check $\frac{1}{2} \mathbf{g}_k^\top \mathbf{Q}_k \mathbf{g}_k + \bm \mu_k^\top \mathbf{g}_k = \frac{1}{2} \widehat{Q}_k \hat{g}_k^2 + \hat{\mu}_k\hat{g}_k$ (up to constants independent of $\hat{g}_k$).
\end{proof}

\subsubsection{Two Examples Demonstrating The Usage of Theorem \ref{thm:NS_conditions_independent_network}-(a)} \begin{example}\label{eg:a_holds_b_fails}
We demonstrate that TBP occurs through a network with 6 links and 4 routes ($1\to 2$, $1 \to 3 \to 5$, $4\to 5$, and $6$) as shown in Figure \ref{fig:eg_fig}. Suppose that the corresponding power network is a simple 2-Bus network connecting by one power line with capacity $\bar{f}$. \mh{The network formed by three routes $1\to 2, 1 \to 3 \to 5, 4 \to 5$ (also known as Wheatstone network \citep{milchtaich2006network})} have charging stations connecting to bus 1, and the remaining lower route has a charging station connecting to bus 2. We set $\bm \alpha = [1,2,2,2,1,1]^\top$, $\bm \beta = \mathbf{0}$, $\mathbf{Q} = \mathrm{diag}([2,1]^\top)$, $\bm \mu = \mathbf{0}$, $\rho = 1$, and $\bar{f} = 0.01$. It can be checked that
\begin{align}
    \widehat{\bm \alpha} = \frac{1}{7}\begin{bmatrix}
        10 & 0\\
        0 & 7
    \end{bmatrix},
    \quad
    \hat{\bm \beta} = \mathbf{0}.
\end{align}
Moreover, solving for GUE yields $\hat{x}_1 \approx 0.37$, $\hat{x}_2 \approx 0.63$, $\hat{\lambda}_1 \approx 0.73$, $\hat{\lambda}_2 \approx 0.63$. \mh{Since the LMPs are different,} the upper Wheatstone network is itself a route bundle, and the remaining route consisting of link 6 forms the other route bundle, which we denote by $\mathcal{P}^1$ and $\mathcal{P}^2$, respectively. Obviously, \mh{\textbf{A}2' holds.}

\begin{figure}[ht]
\centering
\scalebox{0.8}{%
\begin{tikzpicture}[
    node distance=2cm,
    main/.style={circle, draw, thick, minimum size=0.7cm},
    intermediate/.style={circle, draw, thick, minimum size=0.7cm}
]

\node[circle, draw, thick, fill=white, minimum size=0.8cm] (1) at (0,0) {$o$};
\node[circle, draw, thick, fill=white, minimum size=0.8cm] (2) at (3,0) {$d$};
\node[circle, draw, thick, fill=white, minimum size=0.8cm] (3) at (1.5,1) {};
\node[circle, draw, thick, fill=white, minimum size=0.8cm] (4) at (1.5,-1) {};
\node at (1.5,-1.9) {6};
\node at (0.7,0.75) {1};
\node at (2.3,0.75) {2};
\node at (1.3,0) {3};
\node at (0.7,-0.75) {4};
\node at (2.3,-0.75) {5};

\draw[->,>=latex,thick,magenta!80!Black] (3) -- (2);
\draw[->,>=latex,thick,magenta!80!Black] (1) -- (3);
\draw[->,>=latex,thick,magenta!80!Black] (1) -- (4);
\draw[->,>=latex,thick,magenta!80!Black] (4) -- (2);
\draw[->,>=latex,thick,magenta!80!Black] (3) -- (4);
\draw[->,>=latex,thick,blue!80!Black] (1) .. controls (0,-2) and (3,-2) .. (2);  
\end{tikzpicture}
}%
\caption{An example 4-Route 2-Bus system where independence holds.}
\label{fig:eg_fig}
\end{figure}

It can be checked Theorem \ref{thm:theorem_independent_network}-(b1) does not hold \mh{for $\mathcal{P}^1$}, which implies type T-T \mh{A}TBP does not occur with respect to $\mathcal{P}^1$ (i.e., perturbing any link contained by $\mathcal{P}^1$), and thus Theorem \ref{thm:NS_conditions_independent_network}-(a1) fails for type T-T \mh{A}TBP. Note that $\hat{c}_1 = \hat{\alpha}_{1,1} \hat{x}_1$ and \mh{it can be checked that $\partial \hat{\alpha}_{1,1}/\partial \alpha_3 < 0$,} Theorem \ref{thm:NS_conditions_independent_network}-(a2) holds \mh{for $\mathcal{P}^1$}. Therefore, type T-T BP occurs according to the reverse direction of Theorem \ref{thm:NS_conditions_independent_network}-(a).

Both Theorem \ref{thm:theorem_independent_network}-(b1) and (b2) fail \mh{for $\mathcal{P}^2$}, and thus \mh{A}TBP does not occur with respect to $\mathcal{P}^2$. Note that $\hat{c}_2 = \alpha_6 \hat{x}_2$, and thus Theorem \ref{thm:NS_conditions_independent_network}-(a2) fails. Since both Theorem \ref{thm:NS_conditions_independent_network}-(a1) and (a2) fail, TBP \mh{does not} occur when perturbing $\alpha_6$.
\end{example}

\begin{example}
	We consider the same system as in Example \ref{eg:a_holds_b_fails}, except now setting \mh{$\mathbf{Q}=\mathrm{diag}([1,2]^\top)$}. The aggregate GUE becomes \mh{$\hat{x}_1 \approx 0.55$, $\hat{x}_2 \approx 0.45$, $\hat{\lambda}_1 \approx 0.55$, and $\hat{\lambda}_2 \approx 0.89$.} Note that $\hat{\lambda}_2 > \hat{\lambda}_1$ and thus  Theorem \ref{thm:theorem_independent_network}-(b2) holds for $\mathcal{P}^2$, which means type T-P \mh{A}TBP occurs with respect to $\mathcal{P}^2$. Therefore, Theorem \ref{thm:NS_conditions_independent_network}-(a1) holds for $\mathcal{P}^2$. It can be easily seen Theorem \ref{thm:NS_conditions_independent_network}-(a2) fails, as $\mathcal{P}^2$ consists of only one link, which is immune to \mh{classical} BP (cf. Theorem 1 in \cite{milchtaich2006network}). In conclusion, Theorem \ref{thm:NS_conditions_independent_network}-(a) guarantees that type T-P BP occurs when perturbing $\alpha_6$.
\end{example}

\subsubsection{Relations of BPs in General Systems with Radial Power Network} At the beginning of \S\ref{subsec:general_networks}, we assume each subnetwork is associated with an active route bundle (i.e., all routes in the route bundle have nonzero flows). If such an assumption is not imposed, it is possible to have power lines connecting two subnetworks where one is associated with an active route bundle, and the other is not. Then, similar to \eqref{eq:pt_fully_congested} and \eqref{eq:pp_fully_congested}, both \eqref{eq:hat_PT} and \eqref{eq:hat_PP} will include indicator functions specifying if subnetworks $i$ and $i'$ are associated with active routes.
	
	In this case, we have similar relations of BPs as Corollary \ref{cor:BP_relations}, but classical BP also plays a role. For example, if type P-T BP occurs when perturbing $\ell_\mathrm{P} = (i_k,i)$, where subnetwork $i_k$ is associated with active route bundle $k$, and $i$ is not associated with any active route bundle, then \eqref{eq:hat_PT} implies $\hat{\psi}_k < 0$, and Theorem \ref{thm:theorem_independent_network}-(b1) further implies type T-T ATBP does not occur with respect to $\mathcal{P}^k$ (in fact, a strict non-occurrence). Finally, the reverse direction of Theorem \ref{thm:NS_conditions_independent_network}-(a) indicates type T-T BP occurs if and only if classical BP occurs in $\mathcal{P}^k$, which gives a general version of Corollary \ref{cor:BP_relations}-(a). Note that in the special case considered in \S\ref{subsec:completely_congested}, each active route bundle contains only one route, so classical BP never occurs, and thus the conclusion reduces to Corollary \ref{cor:BP_relations}-(a). General versions of Corollary \ref{cor:BP_relations}-(b), (c), (d), and (e) can be similarly obtained.

\subsubsection{Proofs of Necessary and Sufficient Conditions}\label{sec:supp_NS_conditions_proofs}
\mh{We present in this section the most general necessary and sufficient conditions of BPs (cf. Theorem \ref{thm:NS_conditions_for_reduced_network} and \ref{thm:NS_conditions_power}), assuming still that the network has a radial power network, but removing \textbf{A}2'. Then, we show how those theorems can be specialized to prove Theorem \ref{thm:uncongested}, \ref{thm:fully_congested}, and all results in \S\ref{subsec:general_networks}.}

\mh{We start by characterizing the aggregate GUE.} It can be easily checked that the aggregate GUE $(\hat{\mathbf{x}}^\star,\hat{\bm \lambda}^\star)$ solves
\begin{align}\label{eq:characterize_aggregate_GUE}
\begin{cases}
    \widehat{\bm \alpha} \hat{\mathbf{x}} + \hat{\bm \beta} + \rho \hat{\bm \lambda} = \eta \mathbf{1},\\
    \mathrm{diag}\left(\widehat{\mathbf{Q}}\right) \left(\rho \hat{\mathbf{x}} + \hat{\mathbf{f}}\right) + \hat{\bm \mu} = \hat{\bm \lambda},\\
    \mathbf{1}^\top \hat{\mathbf{x}} = 1,
\end{cases}
\quad := \quad
\widehat{\mathbf{B}}
\begin{bmatrix}
    \hat{\mathbf{x}}\\
    \eta
\end{bmatrix}
=\hat{\mathbf{v}}.
\end{align}
We write the LHS equations in a compact form by defining
\begin{align}
    \widehat{\mathbf{B}}:=\begin{bmatrix}
        \widehat{\bm \alpha} + \rho^2 \mathrm{diag}\left(\widehat{\mathbf{Q}}\right) & -\mathbf{1}\\
        \mathbf{1}^\top & 0
    \end{bmatrix}, \quad \hat{\mathbf{v}}:=-\begin{bmatrix}
        \hat{\bm \beta} + \rho \mathrm{diag}\left(\widehat{\mathbf{Q}}\right)\hat{\mathbf{f}}+\mh{\rho}\hat{\bm \mu}\\
        1
    \end{bmatrix}.
\end{align}

\paragraph{Analytical Expression of $\widehat{\mathbf{B}}_\mathrm{ul}^{-1}$.} Define $\bm \Gamma := \widehat{\bm \alpha}+\rho^2\mathrm{diag}(\widehat{\mathbf{Q}})$ the upper left $K \times K$ block of $\widehat{\mathbf{B}}$. The inverse of $\widehat{\mathbf{B}}$ can be computed using its Schur complement $D:=\mathbf{1}^\top \bm \Gamma^{-1} \mathbf{1}$, 
\begin{align}
    \widehat{\mathbf{B}}^{-1} = \begin{bmatrix}
        \bm \Gamma^{-1} - \frac{1}{D}\bm \Gamma^{-1} \mathbf{1}\mathbf{1}^\top \bm \Gamma^{-1} & \frac{1}{D} \bm \Gamma^{-1} \mathbf{1}\\
        -\frac{1}{D} \mathbf{1}^\top \bm \Gamma^{-1} & \frac{1}{D}
    \end{bmatrix}.
\end{align}
Therefore, the upper left $K \times K$ block is
\begin{align}\label{eq:B_ul}
    \widehat{\mathbf{B}}_{\mathrm{ul}}^{-1} = \bm \Gamma^{-1} - \frac{1}{D}\bm \Gamma^{-1} \mathbf{1}\mathbf{1}^\top \bm \Gamma^{-1}.
\end{align}

\begin{lemma}
    The aggregate GUE flow $\hat{\mathbf{x}}$ is
    \begin{align}
        \hat{\mathbf{x}} = -\widehat{\mathbf{B}}_\mathrm{ul}^{-1}\left(\hat{\bm \beta} + \rho \mathrm{diag}(\widehat{\mathbf{Q}})\hat{\mathbf{f}} + \rho \hat{\bm \mu}\right) - \frac{1}{D} \bm \Gamma^{-1} \mathbf{1}.
    \end{align}
\end{lemma}
\begin{proof}
    This follows from directly solving \eqref{eq:characterize_aggregate_GUE} by inverting $\widehat{\mathbf{B}}$.
\end{proof}

\begin{lemma}
    The derivatives of $\hat{\mathbf{x}}$ with respect to $\hat{\alpha}_{k,k'}, k,k' \in [K]$ and $\hat{\beta}_k, k \in [K]$ are
    \begin{align}
        \frac{\partial \hat{\mathbf{x}}}{\partial \hat{\alpha}_{k,k'}} = -\hat{x}_{k'} \widehat{\mathbf{B}}_\mathrm{ul}^{-1} \mathbf{e}_k, \quad \frac{\partial \hat{\mathbf{x}}}{\partial \hat{\beta}_{k}} = -\widehat{\mathbf{B}}_\mathrm{ul}^{-1}\mathbf{e}_k, \quad \frac{\partial \hat{\mathbf{x}}}{\partial \bar{f}_{\ell_\mathrm{P}}} = -\rho \widehat{\mathbf{B}}_\mathrm{ul}^{-1} \mathrm{diag}\left(\widehat{\mathbf{Q}}\right) \widehat{\mathbf{S}} \mathbf{e}_{\ell_\mathrm{P}},
    \end{align}
    \mh{where $\mathbf{e}_k$ is the canonical vector whose $k$-th entry is $1$ and all other entries are $0$'s.}
\end{lemma}
\begin{proof}
    The derivative with respect to $\hat{\beta}_k$ and $\bar{f}_{\ell_\mathrm{P}}$ can be obtained by differentiate the analytical expression of $\hat{\mathbf{x}}$. To obtain the derivative with respect to $\hat{\alpha}_{k,k'}$, we differentiate both sides of equation \eqref{eq:characterize_aggregate_GUE} with respect to $\hat{\alpha}_{k,k'}$, 
    \begin{align}
        \widehat{\mathbf{B}} \begin{bmatrix}
            \partial \hat{\mathbf{x}}/\partial \hat{\alpha}_{k,k'}\\
            \partial \eta/\partial \hat{\alpha}_{k,k'}
        \end{bmatrix} = -\hat{x}_{k'} \mathbf{e}_k \quad \Rightarrow \quad \frac{\partial \hat{\mathbf{x}}}{\partial \hat{\alpha}_{k,k'}} = -\hat{x}_{k'} \widehat{\mathbf{B}}_\mathrm{ul}^{-1} \mathbf{e}_k.
    \end{align}
\end{proof}

\begin{theorem}[Necessary and Sufficient conditions of ATBP]\label{thm:NS_conditions_for_reduced_network}
    For \mh{an aggregated} system characterized by parameters $\{\hat{\bm \alpha}_k,\hat{\beta}_k\}_{k=1}^K$, let its aggregate GUE be $(\hat{\mathbf{x}}^\star, \hat{\bm \lambda}^\star)$, then \mh{A}TBP occurs if and only if any of the following holds
    \begin{enumerate}
        \item[(a)] $\partial \Phi_\mathrm{T}/\partial \hat{\alpha}_{k,k'} = (\hat{x}_{k'}^\star)^2 - \hat{x}_{k'}^\star \left(2\widehat{\bm \alpha} \hat{\mathbf{x}}^\star +\hat{\bm \beta}\right)^\top  \widehat{\mathbf{B}}^{-1}_\mathrm{ul} \mathbf{e}_k < 0$ for some $k,k' \in [K]$;
        \item[(b)] $\partial \Phi_\mathrm{T}/\partial \hat{\beta}_k = \hat{x}_k^\star - \left(2\widehat{\bm \alpha}\hat{\mathbf{x}}^\star+\hat{\bm \beta}\right)^\top \widehat{\mathbf{B}}^{-1}_\mathrm{ul} \mathbf{e}_k < 0$ for some $k \in [K]$;
        \item[(c)] $\partial \Phi_\mathrm{P}/\partial \hat{\alpha}_{k,k'} = -\rho \hat{x}_{k'}^\star (\hat{\bm \lambda}^\star)^\top \widehat{\mathbf{B}}^{-1}_\mathrm{ul} \mathbf{e}_k < 0$ for some $k,k' \in [K]$;
        \item[(d)] $\partial \Phi_\mathrm{P}/\partial \hat{\beta}_k = -\rho (\hat{\bm \lambda}^\star)^\top \widehat{\mathbf{B}}^{-1}_\mathrm{ul} \mathbf{e}_k < 0$ for some $k \in [K]$,
    \end{enumerate}
    where $\widehat{\mathbf{B}}_\mathrm{ul}^{-1}$ is defined in \eqref{eq:B_ul}.
\end{theorem}
\begin{proof}
    Note that $\Phi_\mathrm{T} = \hat{\mathbf{c}}^\top \hat{\mathbf{x}} = (\widehat{\bm \alpha} \hat{\mathbf{x}} + \hat{\bm \beta})^\top \hat{\mathbf{x}} = \hat{\mathbf{x}}^\top \widehat{\bm \alpha} \hat{\mathbf{x}} + \hat{\bm \beta}^\top \hat{\mathbf{x}}$. Hence,
    \begin{align}
        &\frac{\partial \Phi_\mathrm{T}}{\partial \hat{\alpha}_{k,k'}} = \hat{x}_{k'}^2 + \left(2\widehat{\bm \alpha}\hat{\mathbf{x}}+\hat{\bm \beta}\right)^\top \frac{\partial \hat{\mathbf{x}}}{\partial \hat{\alpha}_{k,k'}} = \hat{x}_{k'}^2 - \hat{x}_{k'}\left(\hat{\mathbf{c}} + \widehat{\bm \alpha}\hat{\mathbf{x}}\right)^\top \widehat{\mathbf{B}}_\mathrm{ul}^{-1} \mathbf{e}_k,\\
        &\frac{\partial \Phi_\mathrm{T}}{\partial \hat{\beta}_k} = \hat{x}_k + \left(2\widehat{\bm \alpha}\hat{\mathbf{x}}+\hat{\bm \beta}\right)^\top \frac{\partial \hat{\mathbf{x}}}{\partial \hat{\beta}_k} = \hat{x}_k - \left(\hat{\mathbf{c}} + \widehat{\bm \alpha}\hat{\mathbf{x}}\right)^\top \widehat{\mathbf{B}}_\mathrm{ul}^{-1} \mathbf{e}_k.
    \end{align}
    Note also that
    \begin{align}
        \frac{\partial{\Phi_\mathrm{P}}}{\partial \zeta} = (\mathbf{Q}\mathbf{g} + \bm \mu)^\top\frac{\partial \mathbf{g}}{\partial \zeta} = \bm \lambda^\top \frac{\partial \mathbf{g}}{\partial \zeta} = \hat{\bm \lambda}^\top \frac{\partial \hat{\mathbf{g}}}{\partial\zeta} = \hat{\bm \lambda}^\top \frac{\partial}{\partial \zeta}\left(\rho \hat{\mathbf{x}} + \hat{\mathbf{f}}\right) = \rho \hat{\bm \lambda}^\top \frac{\partial \hat{\mathbf{x}}}{\partial \zeta},
    \end{align}
    where $\zeta \in \{\hat{\alpha}_{k,k'},\hat{\beta}_k\}$. Therefore,
    \begin{align}
        &\frac{\partial \Phi_\mathrm{P}}{\partial \hat{\alpha}_{k,k'}} = -\rho \hat{x}_{k'} \hat{\bm \lambda}^\top \widehat{\mathbf{B}}_\mathrm{ul}^{-1} \mathbf{e}_k, \quad \frac{\partial \Phi_\mathrm{P}}{\partial \hat{\beta}_k} = -\rho \hat{\bm \lambda}^\top \widehat{\mathbf{B}}_\mathrm{ul}^{-1} \mathbf{e}_k.
    \end{align}
\end{proof}

Theorem \ref{thm:NS_conditions_for_reduced_network} provides criteria to check the occurrences of \mh{A}TBP, which, together with \eqref{eq:chain_rule}, leads to the following corollary --- a \textit{sufficient condition} of the non-occurrence of TBP.
\begin{corollary}[A Sufficient Condition of Non-occurrence of TBP]\label{cor:sufficient_condition_for_no_TSIBP}
    For any system, if \mh{A}TBP does not occur, and $\partial \hat{\alpha}_{k,k'}/\partial \alpha_{\ell_\mathrm{T}} \geq 0$ and $\partial \hat{\beta}_k/\partial \alpha_{\ell_\mathrm{T}} \geq 0$ hold for all $k,k' \in [K]$ and for all $\ell_\mathrm{T} \in [\mt]$, then, TBP \mh{does not} occur.
\end{corollary}
\begin{proof}
    This follows directly from the expression below
    \begin{align}
        \frac{\partial \Phi_s}{\partial \alpha_{\ell_\mathrm{T}}} = \sum_{k,k' \in [K]} \frac{\partial \Phi_s}{\partial \hat{\alpha}_{k,k'}}\frac{\partial \hat{\alpha}_{k,k'}}{\partial \alpha_{\ell_\mathrm{T}}} + \sum_{k \in [K]}\frac{\partial \Phi_s}{\partial \hat{\beta}_k}\frac{\partial \hat{\beta}_k}{\partial \alpha_{\ell_\mathrm{T}}}, \quad s \in \{\mathrm{T,P}\}.
    \end{align}
    which is obtained through the chain rule. If the condition of Corollary \ref{cor:sufficient_condition_for_no_TSIBP} holds, then $\partial \Phi_s/\partial \alpha_{\ell_\mathrm{T}}$ is non-negative, which implies TBPs never occur.
\end{proof}

Finally, Theorem \ref{thm:NS_conditions_power} provides characterizations of PBP without \mh{\textbf{A}2'}.

\begin{theorem}[Necessary and Sufficient conditions for PBP]\label{thm:NS_conditions_power}
    For any general coupled system,
    \begin{enumerate}
        \item[(a)] Type P-T BP occurs if and only if there exists $\ell_\mathrm{P}$ such that
        \begin{align}
            \frac{\partial \Phi_\mathrm{T}}{\partial \bar{f}_{\ell_\mathrm{P}}} = -\rho \left(2 \widehat{\bm \alpha} \hat{\mathbf{x}}^\star + \hat{\bm \beta}\right)^\top \widehat{\mathbf{B}}^{-1}_\mathrm{ul} \mathrm{diag}(\widehat{\mathbf{Q}}) \widehat{\mathbf{S}} \mathbf{e}_{\ell_\mathrm{P}} > 0;
        \end{align}
        \item[(b)] Type P-P BP occurs if and only if there exists $\ell_\mathrm{P}$ such that
        \begin{align}
            \frac{\partial \Phi_\mathrm{P}}{\partial \bar{f}_{\ell_\mathrm{P}}} = (\hat{\bm \lambda}^\star)^\top \left(\mathbf{I} - \rho^2 \widehat{\mathbf{B}}_\mathrm{ul}^{-1} \mathrm{diag}(\widehat{\mathbf{Q}}) \right) \widehat{\mathbf{S}} \mathbf{e}_{\ell_\mathrm{P}} > 0.
        \end{align}
    \end{enumerate}
\end{theorem}
\begin{proof}
    \begin{align}
        &\frac{\partial \Phi_\mathrm{T}}{\partial \bar{f}_{\ell_\mathrm{P}}} = \left(2\widehat{\bm \alpha}\hat{\bm \alpha} + \hat{\bm \beta}\right)^\top \frac{\partial \hat{\mathbf{x}}}{\partial \bar{f}_{\ell_\mathrm{P}}} = -\rho \left(\hat{\mathbf{c}} + \widehat{\bm \alpha}\hat{\mathbf{x}}\right)^\top \widehat{\mathbf{B}}_\mathrm{ul}^{-1} \mathrm{diag}\left(\widehat{\mathbf{Q}}\right) \widehat{\mathbf{S}} \mathbf{e}_{\ell_\mathrm{P}},\\
        &\frac{\partial \Phi_\mathrm{P}}{\partial \bar{f}_{\ell_\mathrm{P}}} = \rho \hat{\bm \lambda}^\top \frac{\partial \hat{\mathbf{x}}}{\partial \bar{f}_{\ell_\mathrm{P}}} + \hat{\bm \lambda}^\top \hat{\mathbf{S}} \mathbf{e}_{\ell_\mathrm{P}} = \hat{\bm \lambda}^\top \left(\mathbf{I} - \rho^2 \widehat{\mathbf{B}}_\mathrm{ul}^{-1} \mathrm{diag}\left(\widehat{\mathbf{Q}}\right)\right) \widehat{\mathbf{S}} \mathbf{e}_{\ell_\mathrm{P}}.
    \end{align}
\end{proof}

\mh{Theorem \ref{thm:theorem_independent_network} and \ref{thm:NS_conditions_independent_network} in \S\ref{subsec:general_networks} are direct corollaries of Theorem \ref{thm:NS_conditions_for_reduced_network} and \ref{thm:NS_conditions_power}, since under Assumption \textbf{A}2', Lemma \ref{lemma:properties_of_independent_networks} guarantees that $\widehat{\bm \alpha}$ is diagonal and $\hat{\bm \beta}$ has at most one non-zero entry.}

\begin{proof}[Proof of Theorem \ref{thm:theorem_independent_network}]
	(a) If $K=1$, $\widehat{\mathbf{B}}_\mathrm{ul}^{-1}$ is a scalar and equals $0$.

    (b) By Lemma \ref{lemma:properties_of_independent_networks}, $\widehat{\bm \alpha}$ is diagonal, and thus the upper left $K \times K$ block of $\widehat{\mathbf{B}}$ is,
    \begin{align}
        \begin{bmatrix}
            \hat{\alpha}_{1,1} + \rho^2 \widehat{Q}_1 & \cdots & 0\\
            \vdots & \ddots & \vdots\\
            0 & \cdots & \hat{\alpha}_{K,K} + \rho^2 \widehat{Q}_K &\\
        \end{bmatrix}.
    \end{align}
    Then it can be easily checked that the $k$-th column of $\widehat{\mathbf{B}}_\mathrm{ul}^{-1}$ is
    \begin{align}\label{eq:column_k_Bul}
        \widehat{\mathbf{B}}_\mathrm{ul}^{-1} \mathbf{e}_k = \gamma_k \left(\mathbf{e}_k-\tilde{\bm \gamma}\right).
    \end{align}
    Substitute the expression into the results in Theorem \ref{thm:NS_conditions_for_reduced_network} completes the proof.

    (c) By Lemma \ref{lemma:properties_of_independent_networks} and the chain rule we know
    \begin{align}
        \frac{\partial \Phi_s}{\partial \alpha_{\ell_\mathrm{T}}} = \frac{\partial \Phi_s}{\partial \hat{\alpha}_{k,k}}\frac{\partial \hat{\alpha}_{k,k}}{\partial \alpha_{\ell_\mathrm{T}}} + \frac{\partial \Phi_s}{\partial \hat{\beta}_k}\frac{\partial \hat{\beta}_k}{\partial \alpha_{\ell_\mathrm{T}}}.
    \end{align}
    If the condition in the theorem statement holds, $\partial \Phi_s/\partial \alpha_{\ell_\mathrm{T}} \geq 0$ and thus TBP do not occur.
\end{proof}

\begin{proof}[Proof of Theorem \ref{thm:NS_conditions_independent_network}]
	(a) follows from \eqref{eq:chain_rule} and Lemma \ref{lemma:properties_of_independent_networks}; Both (b) and (c) are direct consequences of the more general Theorem \ref{thm:NS_conditions_power}, as when \textbf{A}2' holds, $\widehat{B}_\mathrm{ul}^{-1}$ is simplified \mh{and the analytical expression of its $k$-th column is given in \eqref{eq:column_k_Bul}.}
\end{proof}

\begin{proof}[Proof of Theorem \ref{thm:uncongested}]
	Theorem \ref{thm:uncongested} is a corollary of both Theorem \ref{thm:NS_conditions_independent_network} and \ref{thm:NS_conditions_power} since under the setting assumed in \S\ref{subsec:uncongested}, there is only one subnetwork (route bundle). Even though we do not assume the power network is radial in \S\ref{subsec:uncongested}, the results in \S\ref{subsec:general_networks} uses the tree topology of the power network to only conclude that routes in the same route bundle see the same LMP, which automatically holds under the setting of \S\ref{subsec:uncongested} \citep{oren1998transmission}.
\end{proof}

\begin{proof}[Proof of Theorem \ref{thm:fully_congested}]
	Theorem \ref{thm:fully_congested} can also be viewed as a corollary of the previously mentioned two theorems. It is merely an extreme case in which each bus is a subnetwork and an active route bundle contains at most one active route. However, in \S\ref{subsec:completely_congested}, we allow the possibility that subnetworks are not associated with any active route bundle. In fact, our proofs in this section still holds by without loss of generality assuming all subnetworks are associated with route bundles (with some route bundles being empty), which only introduces slightly more bookkeeping burdens. Further simplifications of the results in this section are possible because each route bundle contains only one route.
\end{proof}

\newpage
\subsection{Supplementary Materials for \S\ref{sec:mitigation}}\label{sec:proofs}
\begin{lemma}[Transportation UE Given Pricing Policy]\label{lemma:trans_UE_spp}
    An admissible $\mathbf{x^\star}$ is a transportation UE for a given pricing policy $\bm \Pi$ if and only if it solves
    \begin{subequations}\label{eq:solving_equilibrium_general_pricing}
    \begin{align}
        \min_{\mathbf{x}} ~& \frac{1}{2} \mathbf{(\Alr x)^\top \mathrm{diag}(\bm \alpha) \Alr x + (\bm \beta^\top \Alr + \mh{\bm \Pi^\top}) x},\\
        \mathrm{s.t.} ~& \nu \in \mathbb{R}: \mathbf{1^\top x} = 1;\quad \bm \xi \geq \mathbf{0}: \mathbf{x} \geq \mathbf{0}.
    \end{align}
    \end{subequations}
\end{lemma}
\begin{proof}
The proof follows from the proof of Lemma \ref{lemma:tse} by replacing $\bm \lambda$ by $\bm \Pi$.
%
\end{proof}

\begin{lemma}[Existence of Derivatives (Theorem 1 in \cite{amos2017optnet})]\label{lemma:existence_derivatives}
	Social cost metric $\Phi_s, s \in \{\mathrm{T},\mathrm{P},\mathrm{C}\}$ is almost everywhere\footnote{Social cost metrics $\Phi_\mathrm{s}, s \in \{\mathrm{T,P,C}\}$ fail to be differentiable only at boundaries of critical regions, which aggregately has (Lebesgue) measure zero.} differentiable with respect to $\alpha_{\ell_\mathrm{T}},\ell_\mathrm{T} \in [\mt]$ and $\bar f_{\ell_\mathrm{P}}, \ell_\mathrm{P}\in [m_\mathrm{P}]$ over $(\bm \alpha,  \bar{\mathbf{f}}) \in \mathbb{R}_{++}^{\mt} \times  \mathbb{R}_{++}^{m_\mathrm{p}}$. 
\end{lemma}
\begin{proof}[Proof of Lemma \ref{lemma:existence_derivatives}]
	Consider an optimization problem equivalent to Problem \eqref{eq:auxiliary_opt}:
	\begin{subequations}\label{eq:equivalent_RHS_formulation}
		\begin{align}
			\min_{\mathbf{x},\,\mathbf{g},\,\mathbf{p},\mathbf{s}} &\quad  \frac{1}{2} \mathbf g^\top \mathbf Q \mathbf g + \bm \mu^\top \mathbf g + \frac{1}{2} \mathbf x^\top (\Alr)^\top \mathrm{diag}(\mathbf{s}) \Alr \mathbf x + \bm \beta^\top \Alr\mathbf x,  \\
            \mathrm{s.t.} &\quad \bm \lambda: \mathbf{p =  g-d(x)};
            \quad \gamma: \mathbf{1^\top p} = 0;
            \quad  \bm \eta: \mathbf{Hp \leq \bar f};
            \quad \nu: \mathbf{1^\top x} = 1;
            \quad \bm \xi: \mathbf{x \geq 0},\\
            &\quad \mathbf{s} = \bm \alpha.
		\end{align}
	\end{subequations}
	The advantage of Problem \eqref{eq:equivalent_RHS_formulation} over the original problem \eqref{eq:auxiliary_opt} is all parameters appear on the right hand side of constraints, which is a standard structure studied in sensitivity analysis of nonlinear program. Define $(\bm \alpha_1, \mathbf{\bar f}_1) \sim (\bm \alpha_2, \mathbf{\bar f}_2)$ if constraint binding patterns of \eqref{eq:equivalent_RHS_formulation} under them are the same. `$\sim$' can be proved to be an equivalence relation and thus partitions $\mathbb{R}_+^{\mt} \times \mathbb{R}^{\mp}_{++}$ into subsets called \textit{critical regions} for \eqref{eq:equivalent_RHS_formulation}, and also \eqref{eq:auxiliary_opt}. Since there are only finitely many constraint binding patterns (because we have finitely many inequality constraints), the number of critical regions is finite. We assign indices to $\mp+\nr$ inequality constraints by set $[\mp+\nr]$, denote collectively optimal dual variables, and the associated inequality constraints when parameter is given as $(\bm \alpha,\mathbf{\bar f})$ by $\mathbf{u}(\bm \alpha,\mathbf{\bar f}):=[
		\bm \eta~ \bm \xi]^\top$, and $\mathbf{h}(\bm \alpha, \mathbf{\bar f}):=[\mathbf{Hp-\bar f} ~ -\mathbf{x}]^\top$, respectively. Let $\mathcal{B}' \subseteq [\mp+\nr]$ be a constraint binding pattern\footnote{For this proof, we choose a different way to label binding constraints compared to Definition \ref{def:constraint_binding} for simplicity.} and \mh{$\mathcal{C}_{\mathcal{B}'}:=\{(\bm \alpha, \bar{\mathbf{f}}): h_i(\bm \alpha, \mathbf{\bar f}) = 0 , \forall i \in \mathcal{B}', h_j(\bm \alpha,\mathbf{\bar f}) < 0, \forall j \notin \mathcal{B}'\}$} be the corresponding critical region characterized by $\mathcal{B}'$. Define 
		\begin{subequations}
		\begin{align}
			&\mathcal{C}_{\mathcal{B}'}^0 := \{(\bm \alpha, \mathbf{\bar f}):u_i(\bm \alpha,\mathbf{\bar f}) > 0, \forall i \in \mathcal{B}', h_j(\bm \alpha,\mathbf{\bar f}) < 0, \forall j \notin \mathcal{B}'\}, \label{eq:c_b_0}\\
			&\overline{\mathcal{C}_{\mathcal{B}'}} := \{(\bm \alpha,\mathbf{\bar f}):u_i(\bm \alpha,\mathbf{\bar f}) = 0, \exists i \in \mathcal{B}', h_i(\bm \alpha, \mathbf{\bar f}) = 0 , \forall i \in \mathcal{B}', h_j(\bm \alpha,\mathbf{\bar f}) < 0, \forall j \notin \mathcal{B}'\}. \label{eq:c_b_bar}
		\end{align}
		\end{subequations}
	Obviously the sets defined in \eqref{eq:c_b_0} and \eqref{eq:c_b_bar} satisfy $\mathcal{C}_{\mathcal{B}'}^0 \subseteq \mathcal{C}_{\mathcal{B}'}$ and $\overline{\mathcal{C}_{\mathcal{B}'}} \subseteq \mathcal{C}_{\mathcal{B}'}$ since we assume the binding pattern of \eqref{eq:equivalent_RHS_formulation} given $(\bm \alpha,\mathbf{\bar f})$ is ${\mathcal{B}'}$, and $\mathcal{C}_{\mathcal{B}'}^0 \cup \overline{\mathcal{C}_{\mathcal{B}'}} = \mathcal{C}_{\mathcal{B}'}$ by definition. Moreover, we have:
	
	\textit{Claim 1. The interior of critical region $\mathcal{C}_{\mathcal{B}'}$ for any ${\mathcal{B}'}$ contains $\mathcal{C}_{\mathcal{B}'}^0$.} 
	
	By Berge's theorem (also known as maximum theorem) \citep{berge1877topological}, optimal solutions to \eqref{eq:equivalent_RHS_formulation} are continuous functions of $(\bm \alpha, \mathbf{\bar f})$. Therefore, all optimal dual variables are continuous functions of $(\bm \alpha, \mathbf{\bar f})$ since they are all continuous functions of optimal primal variables. For any $(\bm \alpha_0, \mathbf{\bar f}_0) \in \mathcal{C}_{\mathcal{B}'}^0$, by continuity, there is a sufficiently small $\epsilon_0 > 0$ such that for all $(\bm \alpha, \mathbf{\bar f}) \in B_{\epsilon_0}(\bm \alpha_0, \mathbf{\bar f}_0)$ still satisfy $u_i((\bm \alpha, \mathbf{\bar f})) > 0, \forall i \in {\mathcal{B}'}$ and $h_j(\bm \alpha, \mathbf{\bar f}) < 0, \forall j \notin {\mathcal{B}'}$, which implies $B_{\epsilon_0}((\bm \alpha_0, \mathbf{\bar f}_0))\subseteq \mathcal{C}^0_{\mathcal{B}'} \subseteq \mathcal{C}_{\mathcal{B}'}$.
	
	\textit{Claim 2. For any $(\bm \alpha, \mathbf{\bar f}) \in \mathcal{C}^0_{\mathcal{B}'}$, strict complementarity slackness (SCC) holds for  \eqref{eq:equivalent_RHS_formulation} given $(\bm \alpha, \mathbf{\bar f})$.} 
	
	For any $(\bm \alpha, \mathbf{\bar f}) \in \mathcal{C}^0_{\mathcal{B}'}$, by definition the associated optimal dual variables for binding inequality constraints are strictly greater than zero. Therefore, (SCS) holds.
	
	\textit{Claim 3. Optimal dual variables of \eqref{eq:equivalent_RHS_formulation} are unique for any $(\bm \alpha, \mathbf{\bar f})$ and (MFCQ) holds at all feasible points of \eqref{eq:equivalent_RHS_formulation}.} 
	
		We analyze the equivalent problem \eqref{eq:auxiliary_opt}. Any optimal primal-dual tuple of \eqref{eq:auxiliary_opt} satisfies KKT conditions of both \eqref{eq:solving_transportation_equilibrium} and \eqref{eq:economic_dispatch}. It can be easily shown that linearly independent constraint qualification (LICQ) holds at all feasible points of \eqref{eq:solving_transportation_equilibrium}. Together with the fact that \eqref{eq:solving_transportation_equilibrium} is strictly convex, optimal primal-dual tuple is unique. Moreover, solving the dual of \eqref{eq:economic_dispatch} shows optimal dual variables are unique. If optimal dual variables of \eqref{eq:auxiliary_opt} are not unique, it contradicts the fact that optimal dual variables of \eqref{eq:solving_transportation_equilibrium} and \eqref{eq:economic_dispatch} are unique. Uniqueness of optimal dual variables immediately imply (MFCQ) holds at all feasible points for \eqref{eq:auxiliary_opt}, and thus for \eqref{eq:equivalent_RHS_formulation} \citep{wachsmuth2013licq}.
	
	\textit{Claim 4. Optimal primal variables of \eqref{eq:equivalent_RHS_formulation} given any $(\bm \alpha_0, \mathbf{\bar f}_0) \in \mathcal{C}_{\mathcal{B}'}^0$ are all differentiable with respect to $(\bm \alpha, \mathbf{\bar f})$.}
	
	Since (KKT condition), (second-order sufficient condition), (SCS), and (MFCQ) holds for the problem \eqref{eq:equivalent_RHS_formulation} for any given $(\bm \alpha_0, \mathbf{\bar f}_0) \in \mathcal{C}^0_{\mathcal{B}'}$, optimal primal variables are differentiable with respect to $(\bm \alpha, \mathbf{\bar f})$ \citep{giorgi2018tutorial}.  
	
	\textit{Claim 5. The set $\overline{\mathcal{C}_{\mathcal{B}'}}$ has Lebesgue measure zero.}
	
	First note that set $\overline{\mathcal{C}_{\mathcal{B}'}}$ has no intersection with the interior of $\mathcal{C}_{\mathcal{B}'}$, which means it is a subset of boundary of $\mathcal{C}_{\mathcal{B}'}$, $\partial \mathcal{C}_{\mathcal{B}'}$. To see this, suppose for contradiction there exists $(\bm \alpha_0, \mathbf{\bar f}_0) \in \overline{\mathcal{C}_{\mathcal{B}'}} \cap \mathrm{int}(\mathcal{C}_{\mathcal{B}'})$. Then, there is a sufficiently small $\epsilon > 0$ such that the ball centered at $(\bm \alpha_0,\mathbf{\bar f}_0)$ with radius $\epsilon$, $B_\epsilon(\bm \alpha_0,\mathbf{\bar f}_0) \subseteq \mathcal{C}_{\mathcal{B}'}$. Let $i_0$ be the index of the binding constraint whose associated optimal dual variable is zero. We can always change model parameter such that the $i_0$-th inequality constraint becomes unbinding while the model parameter still remains in $B_{\epsilon}(\bm \alpha_0,\mathbf{\bar f}_0)$. Now, constraint binding pattern is no longer ${\mathcal{B}'}$, which yields a contradiction. Since $\partial \mathcal{C}_{\mathcal{B}'}$ is the boundary of a convex set (it is known critical regions of a nonlinear program whose constraints are linear in parameters, are convex polytope \citep{pistikopoulos2007multi}), it has measure zero \citep{lang1986note}. Together with the fact $\overline{\mathcal{C}_{\mathcal{B}'}} \subseteq \partial \mathcal{C}_{\mathcal{B}'}$, we conclude $\overline{\mathcal{C}_{\mathcal{B}'}}$ has measure zero.

	By Claim 5, we know $\bigcup_{{\mathcal{B}'}} \overline{\mathcal{C}_{\mathcal{B}'}}$ has  Lebesgue measure zero since it is a finite union of measure zero sets. Claim 4 implies optimal primal variables of \eqref{eq:equivalent_RHS_formulation} are differentiable with respect to $(\bm \alpha, \mathbf{\bar f})$ at any $(\bm \alpha_0, \mathbf{\bar f}_0) \in \bigcup_{\mathcal{B}'}\mathcal{C}^0_{\mathcal{B}'}$. Together with the fact that the parameter space is $\bigcup_{{\mathcal{B}'}} \mathcal{C}_{\mathcal{B}'} = \bigcup_{{\mathcal{B}'}} \left(\mathcal{C}^0_{\mathcal{B}'}\cup \overline{\mathcal{C}_{\mathcal{B}'}}\right) = \left[\bigcup_{\mathcal{B}'}\mathcal{C}^0_{\mathcal{B}'}\right] \bigcup \left[\bigcup_{\mathcal{B}'} \overline{\mathcal{C}_{\mathcal{B}'}}\right]$, we know except the set $\bigcup_{\mathcal{B}'} \overline{\mathcal{C}_{\mathcal{B}'}}$, which is of measure zero, optimal primal variables are differentiable with respect to $(\bm \alpha, \bar{\mathbf{f}})$. The proof is thus complete.
\end{proof}

\begin{lemma}[Derivative Computation]\label{lemma:derivatives}
The derivatives of the social cost metrics with respect of capacity parameters satisfy, for each $\lt \in [\mt]$ and $\lp \in [\mp]$,
\begin{subequations}
\begin{align}
\frac{\partial \Phi_\mathrm{T}}{\partial \alpha_{\ell_\mathrm{T}}} = \frac{\partial (\mathbf{x^\top [c(x,\bm \lambda)-\bm \pi(\bm \lambda)])}}{\partial \alpha_{\ell_\mathrm{T}}}, \quad &
\frac{\partial \Phi_\mathrm{P}}{\partial \alpha_{\ell_\mathrm{T}}} = \bm \lambda^\top \frac{ \partial \mathbf{d(x)}}{\partial \alpha_{\ell_\mathrm{T}}},  \\
\frac{\partial \Phi_\mathrm{T}}{\partial \bar f_{\ell_\mathrm{P}}} = \frac{\partial (\mathbf{x^\top [c(x,\bm \lambda)-\bm \pi(\bm \lambda)])}}{\partial \bar f_{\ell_\mathrm{P}}}, \quad &
\frac{\partial \Phi_\mathrm{P}}{\partial \bar f_{\ell_\mathrm{P}}} = \bm \lambda^\top \frac{ \partial \mathbf{d(x)}}{\partial \bar f_{\ell_\mathrm{P}}} - \eta_{\ell_\mathrm{P}}.
\end{align}
\end{subequations}
\end{lemma}
\begin{proof}[Proof of Lemma \ref{lemma:derivatives}]
	Derivatives $\partial \Phi_\mathrm{T}/\partial \alpha_{\ell_\mathrm{T}}$ and $\partial \Phi_\mathrm{T}/\partial \bar f_{\ell_\mathrm{P}}$ follow directly from the definition that 
	\begin{equation}
	\Phi_\mathrm{T}:=\mathbf{x}^\top \left[\mathbf{c(x,\bm \lambda)-\bm \pi(\bm \lambda)}\right].
	\end{equation}
	To derive $\partial \Phi_\mathrm{P}/\partial \alpha_{\ell_\mathrm{T}}$ and $\partial \Phi_\mathrm{P}/\partial \bar f_{\ell_\mathrm{P}}$, note first that:
    \begin{subequations}
    	\begin{align}
        \frac{\partial (\bm \lambda^\top \mathbf{d})}{\partial \alpha_{\ell_\mathrm{T}}} &= \frac{\partial \bm \lambda^\top (\mathbf{g}-\mathbf{p})}{\partial \alpha_{\ell_\mathrm{T}}}=2\mathbf{g^\top Q} \frac{\partial \mathbf{g}}{\partial \alpha_{\ell_\mathrm{T}}} + \bm \mu^\top \frac{\partial \mathbf{g}}{\partial \alpha_{\ell_\mathrm{T}}} -\frac{\partial \bm \lambda^\top \mathbf{p}}{\partial \alpha_{\ell_\mathrm{T}}}\\
        &=  2\mathbf{g^\top Q} \frac{\partial \mathbf{g}}{\partial \alpha_{\ell_\mathrm{T}}} + \bm \mu^\top \frac{\partial \mathbf{g}}{\partial \alpha_{\ell_\mathrm{T}}} + \frac{\partial \bm \eta^\top \mathbf{Hp}}{\partial \alpha_{\ell_\mathrm{T}}}\\
        &=\frac{\partial \Phi_\mathrm{P}}{\partial \alpha_{\ell_\mathrm{T}}} + \mathbf{g^\top Q}\frac{\partial \mathbf{g}}{\partial \alpha_{\ell_\mathrm{T}}}+ \frac{\partial \bm \eta^\top \mathbf{Hp}}{\partial \alpha_{\ell_\mathrm{T}}},
    \end{align}
    \end{subequations}
    where $\mathbf{x},\mathbf{g},\mathbf{p},\bm \lambda,\gamma, \bm \eta$ are optimal solutions to \eqref{eq:auxiliary_opt} viewed as functions of $(\bm \alpha, \mathbf{\bar f})$. We apply complementary slackness $\bm \lambda^\top (\mathbf{g}-\mathbf{p}-\mathbf{d}) = 0$ to obtain the first equaility, $\bm \lambda = \mathbf{Qg+\bm \mu}$ to obtain the second equality, $\bm \lambda = \gamma \mathbf{1} - \mathbf{H^\top \bm \eta}$ to obtain the third equaility, and $\Phi_\mathrm{P} = \frac{1}{2}\mathbf{g^\top Q g + \bm \mu^\top g}$ to obtain the last equaility. Therefore,
    \begin{subequations}
        \begin{align}\label{eq:de1_1}
            \bm \lambda^\top \frac{\partial \mathbf{d}}{\partial \alpha_{\ell_\mathrm{T}}} &= \frac{\partial \Phi_\mathrm{P}}{\partial \alpha_{\ell_\mathrm{T}}} + \mathbf{g^\top Q}\frac{\partial \mathbf{g}}{\partial \alpha_{\ell_\mathrm{T}}}+ \frac{\partial \bm \eta^\top \mathbf{Hp}}{\partial \alpha_{\ell_\mathrm{T}}} - \mathbf{d}^\top \frac{\partial \bm \lambda}{\partial \alpha_{\ell_\mathrm{T}}}\\
            &=\frac{\partial \Phi_\mathrm{P}}{\partial \alpha_{\ell_\mathrm{T}}} + \mathbf{g}^\top \frac{\partial \bm \lambda}{\partial \alpha_{\ell_\mathrm{T}}} -\mathbf{p}^\top \frac{\partial \bm \lambda}{\partial \alpha_{\ell_\mathrm{T}}} - \bm \lambda^\top \frac{\partial \mathbf{p}}{\partial \alpha_{\ell_\mathrm{T}}} - \mathbf{d}^\top \frac{\partial \bm \lambda}{\partial \alpha_{\ell_\mathrm{T}}}\\
            &=\frac{\partial \Phi_\mathrm{P}}{\partial \alpha_{\ell_\mathrm{T}}}-\bm \lambda^\top \frac{\partial \mathbf{p}}{\partial \alpha_{\ell_\mathrm{T}}} = \frac{\partial \Phi_\mathrm{P}}{\partial \alpha_{\ell_\mathrm{T}}},
        \end{align}
    \end{subequations}
    where we use $\bm \lambda^\top \partial \mathbf{p}/\partial \alpha_{\ell_\mathrm{T}} = (\gamma\mathbf{1}^\top - \bm \eta^\top \mathbf{H})\partial \mathbf{p}/\partial \alpha_{\ell_\mathrm{T}} = -\bm \eta^\top \partial \mathbf{Hp}/\partial \alpha_{\ell_\mathrm{T}}= -(\mathbf{\bar f}-\mathbf{f})^\top \partial \bm \eta/\partial \alpha_\mathrm{\ell_\mathrm{T}} = 0$. This is because the last inner product can be written as a summation whose summand takes the form $(\bar f_\ell - f_\ell) \cdot \partial \eta_\ell/\partial \alpha_{\ell_\mathrm{T}}$. If $f_\ell < \bar f_\ell$, then $\eta_l = 0$ according to complementary slackness. Since optimal solutions to \eqref{eq:auxiliary_opt} are continuous in $(\bm \alpha, \mathbf{\bar f})$ according to the proof of Lemma \ref{lemma:existence_derivatives}, changing $\alpha_{\ell_\mathrm{T}}$ by a small amount would not make the constraint $(\mathbf{Hp})_\ell \leq \bar f_\ell$ binding. Therefore, $\partial \eta_\ell/\partial \alpha_{\ell_\mathrm{P}} = 0$. If $f_\ell = \bar f_\ell$, then the summand is trivially zero. The proof for $\bm \lambda^\top \partial \mathbf{d}/\partial \bar f_l$ follows from a similar argument. Note that:
    \begin{subequations}
    	\begin{align}
        \frac{\partial (\bm \lambda^\top \mathbf{d})}{\partial \bar f_{\ell_\mathrm{P}}} &= \frac{\partial \bm \lambda^\top (\mathbf{g}-\mathbf{p})}{\partial \bar f_{\ell_\mathrm{P}}}=2\mathbf{g^\top Q} \frac{\partial \mathbf{g}}{\partial \bar f_{\ell_\mathrm{P}}} + \bm \mu^\top \frac{\partial \mathbf{g}}{\partial \bar f_{\ell_\mathrm{P}}} -\frac{\partial \bm \lambda^\top \mathbf{p}}{\partial \bar f_{\ell_\mathrm{P}}}\\
        &=\frac{\partial \Phi_\mathrm{P}}{\partial \bar f_{\ell_\mathrm{P}}} + \mathbf{g^\top Q}\frac{\partial \mathbf{g}}{\partial \bar f_{\ell_\mathrm{P}}}-\bm \lambda^\top \frac{\partial \mathbf{p}}{\partial \bar f_{\ell_\mathrm{P}}} - \mathbf{p}^\top \frac{\partial \bm \lambda}{\partial \bar f_{\ell_\mathrm{P}}}\\
        &=\frac{\partial \Phi_\mathrm{P}}{\partial \bar f_{\ell_\mathrm{P}}} + \eta_{\ell_\mathrm{P}}+ \mathbf{g^\top Q}\frac{\partial \mathbf{g}}{\partial \bar f_{\ell_\mathrm{P}}} - \mathbf{p}^\top \frac{\partial \bm \lambda}{\partial \bar f_{\ell_\mathrm{P}}},
    \end{align}
    \end{subequations}
    where we apply $-\bm \lambda^\top \partial \mathbf{p}/\partial \bar f_{\ell_\mathrm{P}} = (\mathbf{\bar f - f})^\top \partial \bm \eta/\partial \bar f_{\ell_\mathrm{P}} + \eta_{\ell_\mathrm{P}} = \eta_{\ell_\mathrm{P}}$ to obtain the last equality. Rearranging terms we have:    \begin{equation}
    	\frac{\partial \Phi_\mathrm{P}}{\partial \bar f_{\ell_\mathrm{P}}} = \bm \lambda^\top\frac{\partial \mathbf{d(x)}}{\partial \bar f_{\ell_\mathrm{P}}} - \eta_{\ell_\mathrm{P}}. \label{eq:de1_2}
    \end{equation}
Finally, we prove $\mathbf{c(x,\bm \lambda)}^\top (\partial \mathbf{x}/\partial \alpha_{\ell_\mathrm{T}}) = \mathbf{c(x,\bm \lambda)}^\top (\partial \mathbf{x}/\partial \bar f_{\ell_\mathrm{P}}) = 0$. Denote the optimal value of Problem \eqref{eq:auxiliary_opt} as a function of $(\bm \alpha, \mathbf{\bar f})$, $\Phi^\star(\bm \alpha,\mathbf{\bar f})$, and apply sensitivity analysis techniques on Problem \eqref{eq:auxiliary_opt} we have:
\begin{equation}
	\frac{\partial \Phi^\star}{\partial \alpha_{\ell_\mathrm{T}}} = \frac{1}{2}(\Alr\mathbf{x})_{\ell_\mathrm{T}}^2, \quad \frac{\partial \Phi^\star}{\partial \bar f_{\ell_\mathrm{P}}} = -\eta_{\ell_\mathrm{P}}.
\end{equation}
Note also that $\Phi_\mathrm{C} = \Phi^\star + \frac{1}{2}(\Alr \mathbf{x})^\top \mathrm{diag}(\bm \alpha) \Alr \mathbf{x}$ when it is evaluated at a GUE. Therefore, 
\begin{subequations}\label{eq:prove_sufficient_condition_lemma}
\begin{align}
		&\frac{\partial \Phi_\mathrm{C}}{\partial \alpha_{\ell_\mathrm{T}}} = (\Alr \mathbf{x})_{\ell_\mathrm{T}}^2 + (\Alr \mathbf{x})^\top \mathrm{diag}(\bm \alpha)\Alr \frac{\partial \mathbf{x}}{\partial \alpha_{\ell_\mathrm{T}}} = \frac{\partial \Phi_\mathrm{T}}{\partial \alpha_{\ell_\mathrm{T}}} - \left[\mathbf{c(x,\bm \lambda) - \bm \pi}\right]^\top \frac{\partial \mathbf{x}}{\partial \alpha_{\ell_\mathrm{T}}}, \label{eq:de2_1}\\
		&\frac{\partial \Phi_\mathrm{C}}{\partial \bar f_{\ell_\mathrm{P}}} =
(\Alr \mathbf{x})^\top \mathrm{diag}(\bm \alpha)\Alr \frac{\partial \mathbf{x}}{\partial \bar f_{\ell_\mathrm{P}}} - \eta_{\ell_\mathrm{P}} = \frac{\partial \Phi_\mathrm{T}}{\partial \bar f_{\ell_\mathrm{P}}}-\left[\mathbf{c(x,\bm \lambda) - \bm \pi}\right]^\top \frac{\partial \mathbf{x}}{\partial \bar f_{\ell_\mathrm{P}}} - \eta_{\ell_\mathrm{P}}. \label{eq:de2_2}
\end{align}	
\end{subequations}
Comparing Equations \eqref{eq:de1_1} and \eqref{eq:de1_2} with Equations \eqref{eq:de2_1} and \eqref{eq:de2_2} we have:
    \begin{subequations}
        \begin{align}
        &\bm \lambda^\top \frac{\partial \mathbf{d(x)}}{\partial \alpha_{\ell_\mathrm{T}}} = \frac{\partial \Phi_\mathrm{P}}{\partial \alpha_{\ell_\mathrm{T}}} = \bm \lambda^\top \frac{\partial \mathbf{d(x)}}{\partial \alpha_{\ell_\mathrm{T}}}-\mathbf{c(x, \bm \lambda)}^\top \frac{\partial \mathbf{x}}{\partial \alpha_{\ell_\mathrm{T}}}, \\
        &\bm \lambda^\top \frac{\partial \mathbf{d(x)}}{\partial \bar f_{\ell_\mathrm{P}}} - \eta_{\ell_\mathrm{P}}= \frac{\partial \Phi_\mathrm{P}}{\partial \bar f_{\ell_\mathrm{P}}} = \bm \lambda^\top \frac{\partial \mathbf{d(x)}}{\partial \bar f_{\ell_\mathrm{P}}} - \mathbf{c(x,\bm \lambda)}^\top\frac{\partial \mathbf{x}}{\partial \bar f_{\ell_\mathrm{P}}} - \eta_{\ell_\mathrm{P}}.
        \end{align}
    \end{subequations}
    Rearranging terms we conclude that $\mathbf{c(x,\bm \lambda)}^\top (\partial \mathbf{x}/\partial \alpha_{\ell_\mathrm{T}}) = \mathbf{c(x,\bm \lambda)}^\top (\partial \mathbf{x}/\partial \bar f_{\ell_\mathrm{P}}) = 0$.
\end{proof}

\begin{proof}[Proof of Theorem \ref{thm:sysoptPi}]
The main idea for the proof is constructing optimization problems whose optimal solutions are equivalent to GUE induced by a specific charging pricing policy $\bm \Pi$. Then, analyzing the properties of those auxiliary problems.

	($\bm \Pi^\star_\mathrm{T}$): Consider an auxiliary optimization problem:
	\begin{subequations}\label{eq:trans_social_optimum}
	\begin{align}
		\min_{\mathbf{x}} ~& \frac{1}{2}(\Alr \mathbf{x})^\top \mathrm{diag}(\bm \alpha)\Alr \mathbf{x} + \bm \beta^\top \Alr \mathbf{x} + \frac{1}{2}(\bm \Pi_\mathrm{T}^\star)^\top \mathbf{x},\\
		\mathrm{s.t.} ~& \nu: \mathbf{1^\top x} = 1, \quad \bm \xi: \mathbf{x} \geq \mathbf{0}.
	\end{align}
	\end{subequations}
	The stationarity condition is:
	\begin{equation}
		\underbrace{(\Alr)^\top \mathrm{diag}(\bm \alpha)\Alr \mathbf{x} + (\Alr)^\top \bm \beta}_{\text{Travel cost}} + \underbrace{(\Alr)^\top \mathrm{diag}(\bm \alpha)\Alr \mathbf{x}}_{\bm \Pi_\mathrm{T}^\star} = \nu \mathbf{1} + \bm \xi.
			\end{equation}
Imitate the proof of Lemma \ref{lemma:tse} one can show GUE induced by $\bm \Pi^\star_\mathrm{T}$ is equivalent to optimal solutions to \eqref{eq:trans_social_optimum}. Note that the objective function of \eqref{eq:trans_social_optimum} is in fact $\Phi_\mathrm{T}$ so the induced GUE achieves minimal $\Phi_\mathrm{T}$ for the given model parameter. Apply sensitivity analysis and it is not hard to show optimal value of \eqref{eq:trans_social_optimum} would never decrease if $\alpha_{\ell_\mathrm{T}}$ increases for any $\ell_\mathrm{T} \in [\mt]$, which means type T-T BP would never occur. Moreover, as Problem \eqref{eq:trans_social_optimum} does not depend on transmission line capacity $\mathbf{\bar f}$, changing $\mathbf{\bar f}$ has no impact on the GUE and thus, has no impact on $\Phi_\mathrm{T}$. Since power load $\mathbf{d}$ does not change with $\mathbf{\bar f}$. The feasible region of economic dispatch \eqref{eq:economic_dispatch} enlarges as one increases $\bar f_{\ell_\mathrm{P}}$ for any $\ell_\mathrm{P} \in [\mp]$. Therefore, PBP never occurs.
	
	($\bm \Pi^\star_\mathrm{P}$): Consider an auxiliary optimization problem:
	\begin{subequations}\label{eq:power_social_optimum}
		\begin{align}
		\min_{\mathbf{x,g,p}}~& \frac{1}{2}\mathbf{g^\top Q g+\bm \mu^\top g},\\
		\mathrm{s.t.}~& \mathbf{d(\mathbf{x}) + p = g}; \quad \mathbf{1^\top p} = 0; \quad \mathbf{Hp \leq \bar f};\quad \nu:\mathbf{1^\top x} = 1; \quad \bm \xi:\mathbf{x \geq 0}.	
		\end{align}
	\end{subequations}
	The objective function is $\Phi_\mathrm{P}$ and the stationarity condition in terms of $\mathbf{x}$ is:
	\begin{equation}
		\underbrace{(\Alr)^\top \mathrm{diag}(\bm \alpha)\Alr \mathbf{x} + (\Alr)^\top \bm \beta}_{\text{Travel cost}} + \underbrace{\bm \pi(\bm \lambda) - (\Alr)^\top \mathrm{diag}(\bm \alpha)\Alr \mathbf{x} - (\Alr)^\top \bm \beta}_{\bm \Pi_\mathrm{P}^\star} = \nu \mathbf{1} + \bm \xi.
	\end{equation}
	One can similarly prove GUE induced by $\bm \Pi_\mathrm{P}^\star$ is equivalent to optimal solutions to \eqref{eq:power_social_optimum}. Since Problem \eqref{eq:power_social_optimum} does not depend on $\bm \alpha$, changing $\bm \alpha$ does not affect GUE and thus TBP never occurs. Moreover, for any $\ell_\mathrm{P}\in [\mp]$, increasing $\bar f_{\ell_\mathrm{P}}$ enlarges feasible region of \eqref{eq:power_social_optimum} and thus would never increase power system social cost.
	
	($\bm \Pi^\star_\mathrm{C}$): Consider the auxiliary optimization problem:
		\begin{subequations}\label{eq:total_social_optimum}
		\begin{align}
		\min_{\mathbf{x,g,p}}~& (\Alr \mathbf{x})^\top \mathrm{diag}(\bm \alpha) \Alr \mathbf{x} + \bm \beta^\top \Alr \mathbf{x}+\frac{1}{2}\mathbf{g^\top Q g+\bm \mu^\top g},\\
		\mathrm{s.t.}~& \mathbf{d(\mathbf{x}) + p = g}; \quad \mathbf{1^\top p} = 0; \quad \mathbf{Hp \leq \bar f};\quad \nu:\mathbf{1^\top x} = 1; \quad \bm \xi:\mathbf{x \geq 0}.	
		\end{align}
	\end{subequations}
	Following the similar idea as the previous two cases, one can prove optimal solutions to \eqref{eq:total_social_optimum} are equivalent to GUE induced by $\bm \Pi_\mathrm{C}^\star$. The objective function is total social cost so the induced GUE achieves mimimum total social cost. Through sensitivity analysis, one can show increasing $\alpha_{\ell_\mathrm{T}}$ for any $\ell_\mathrm{T}$ would never decrease the optimal objective value and increasing $\bar f_{\ell_\mathrm{P}}$ for any $\ell_\mathrm{P}$ would never increase optimal objective function. Hence, neither type T-C nor P-C BP occur.
\end{proof}

\begin{proposition}\label{prop:derivatives_TT_TP}
    Let $(\mathbf{x},\bm \lambda)$ be a GUE induced by static pricing policy $\bm \Pi$. The derivatives of social costs with respect to $\alpha_{\ell_\mathrm{T}}, \ell_\mathrm{T} \in [m_\mathrm{T}]$, are 
    \begin{subequations}
        \begin{align}
             &\frac{\partial \Phi_\mathrm{T}}{\partial \alpha_{\ell_\mathrm{T}}} = (\Alr \mathbf{x})_{\lt}^2 + \left(2(\Alr\mathbf{x})^\top \mathrm{diag}(\bm\alpha) + \bm \beta^\top\right) \Alr \frac{\partial \mathbf{x}}{\partial \alpha_{\ell_\mathrm{T}}}, 
          \qquad \frac{\partial \Phi_\mathrm{P}}{\partial \alpha_{\ell_\mathrm{T}}} = \bm \pi (\bm \lambda)^\top \frac{\partial \mathbf{x}}{\partial \alpha_{\ell_\mathrm{T}}}. 
        \end{align}
    \end{subequations}
\end{proposition}

\begin{proof}[Proof of Proposition \ref{prop:derivatives_TT_TP}]
Directly differentiate $\Phi_\mathrm{T}:=(\Alr \mathbf{x})^\top \Alr \mathbf{x} + \bm \beta^\top \Alr \mathbf{x}$ with respect to $\alpha_{\ell_\mathrm{T}}$ gives the desired result. Formula for $\partial \Phi_\mathrm{P}/\partial \alpha_{\ell_\mathrm{T}}$ follows from Lemma \ref{lemma:derivatives}. 
\end{proof}

\begin{proof}[Proof of Theorem \ref{thm:GUE_affine}]
	Let $\mathcal{B} = (\mathcal{R},\mathcal{L}_\mathrm{P})$ be a given constraint binding pattern and let $\bm \Pi \in \mathcal{P}_\mathcal{B}^\mathrm{cr}$ be fixed.

	(a) We first prove $\mathbf{x}$ is piecewise affine in $\bm \Pi$.  Consider a reduced version of \eqref{eq:solving_transportation_equilibrium} with $\bm \pi$ replaced by $\bm \Pi$, for which we drop all binding constraints. Let $\mathbf{x}^\star$ be one optimal solution to the original optimization problem and $\mathbf{S}$ be the selection matrix defined based on $\mathcal{R}$ that selects all nonzero entries from $\mathbf{x}^\star$. A shortened vector $\mathbf{Sx^\star}$ optimizes the reduced problem. The Lagrangian function of the reduced optimization problem is:
    \begin{subequations}
        \begin{align}
           L &= \frac{1}{2} \mathbf{x^\top (\Alr S^\top S)^\top \mathrm{diag}(\bm \alpha) (\Alr S^\top S) x} + \left( (\Alr \mathbf{S}^\top)^\top \bm \beta + \mathbf{S} \bm \Pi - \nu \mathbf{S 1}\right)^\top \mathbf{S x} + \nu,
        \end{align}
    \end{subequations}
    in which $\mathbf{S} \bm \xi$ is eliminated since it is equal to $\mathbf{0}$. The optimal $\mathbf{x}^\star$ that optimizes $L$
    \mh{satisfies:
    \begin{equation}
   (\Alr \mathbf{S}^\top \mathbf{S})^\top  \mathrm{diag}(\bm \alpha) (\Alr \mathbf{S}^\top \mathbf{S}) \mathbf{x}^\star(\nu) = \mathbf{S^\top} (\nu \mathbf{S 1} - (\Alr \mathbf{S}^\top)^\top \bm \beta - \mathbf{S} \bm \Pi),
    \end{equation}
    where $(\Alr \mathbf{S}^\top \mathbf{S})^\top  \mathrm{diag}(\bm \alpha) (\Alr \mathbf{S}^\top \mathbf{S})$ is bijective from $\{\mathbf{x}:\mathbf{(I-S^\top S)x = 0}\}$ to $\{\mathbf{x}:\mathbf{(I-S^\top S)x = 0}\}$. Therefore, there is a unique $\mathbf{x^\star}(\nu)$ such that $(\mathbf{I-S^\top S})\mathbf{x}^\star(\nu) = \mathbf{0}$ that solves the linear equations. Left multiplying $\mathbf{S}$ on both sides and use the fact $\mathbf{S^\top S x^\star}(\nu) = \mathbf{x}^\star(\nu)$ yields: 
    \begin{equation}
    	\underbrace{\mathbf{S}\left[ (\Alr \mathbf{S}^\top \mathbf{S})^\top  \mathrm{diag}(\bm \alpha) (\Alr \mathbf{S}^\top \mathbf{S}) \right] \mathbf{S^\top}}_{\text{invertible}} \mathbf{Sx}^\star(\nu) = \mathbf{S}\mathbf{S^\top} (\nu \mathbf{S 1} - (\Alr \mathbf{S}^\top)^\top \bm \beta - \mathbf{S} \bm \Pi).
    \end{equation}
    The coefficient matrix of $\mathbf{Sx^\star}(\nu)$ is invertible since if there exists $\mathbf{y \neq 0}$ such that the product of the matrix and $\mathbf{y}$ is $\mathbf{0}$, then it contradicts the fact that $(\Alr)^\top \mathrm{diag}(\bm \alpha) \Alr$ is positive definite since $\mathbf{S^\top S S^\top y \neq 0}$. Therefore,
    \begin{subequations}
    \begin{align}
    	&\mathbf{Sx}^\star(\nu) = \left[\mathbf{S}\left[ (\Alr \mathbf{S}^\top \mathbf{S})^\top  \mathrm{diag}(\bm \alpha) (\Alr \mathbf{S}^\top \mathbf{S}) \right] \mathbf{S^\top}\right]^{-1} \mathbf{S}\mathbf{S^\top} (\nu \mathbf{S 1} - (\Alr \mathbf{S}^\top)^\top \bm \beta - \mathbf{S} \bm \Pi),\\
    	\Rightarrow & \mathbf{x^\star} = \underbrace{\mathbf{S^\top} \left[\mathbf{S}\mathbf{S^\top S} (\Alr)^\top  \mathrm{diag}(\bm \alpha) \Alr \mathbf{S}^\top \mathbf{S}  \mathbf{S^\top}\right]^{-1} \mathbf{S}}_{\mathbf{M}(\bm \alpha) \in \mathbb{R}^{\nr\times \nr}}\mathbf{S^\top} (\nu \mathbf{S 1} - (\Alr \mathbf{S}^\top)^\top \bm \beta - \mathbf{S} \bm \Pi),
    	\end{align}
    	\label{eq:M}
    \end{subequations}
    where we again use $\mathbf{S^\top S x^\star}(\nu) = \mathbf{x}^\star(\nu)$.}
Therefore, the dual problem is:
    \begin{subequations}
        \begin{align}
            \max_{\nu} ~& -\frac{1}{2}(\nu \mathbf{S 1} - (\Alr \mathbf{S}^\top)^\top \bm \beta - \mathbf{S} \bm \Pi)^\top \mathbf{S} \mathbf{M} \mathbf{S}^\top (\nu \mathbf{S 1} - (\Alr \mathbf{S}^\top)^\top \bm \beta - \mathbf{S} \bm \Pi) + \nu.
        \end{align}
    \end{subequations}
    The optimal $\nu^\star$ is:
    \begin{subequations}
        \begin{align}
        \nu^\star &= \frac{1+\mathbf{1^\top S^\top S \left[ (\Alr \mathbf{S}^\top \mathbf{S})^\top \mathrm{diag}(\bm \alpha) (\Alr \mathbf{S}^\top \mathbf{S}) \right]^{-1} S^\top (\mathbf{S} \bm \Pi + (\Alr \mathbf{S}^\top)^\top \bm \beta)}}{ \mathbf{1^\top S^\top S} \left[ (\Alr \mathbf{S}^\top \mathbf{S})^\top \mathrm{diag}(\bm \alpha) (\Alr \mathbf{S}^\top \mathbf{S}) \right]^{-1} \mathbf{S^\top S 1}} \\
        &= \frac{1+\mathbf{1^\top S^\top S \mathbf{M} S^\top (\mathbf{S} \bm \Pi + (\Alr \mathbf{S}^\top)^\top \bm \beta)}}{ \mathbf{1^\top S^\top S} \mathbf{M} \mathbf{S^\top S 1}}.
        \end{align}
    \end{subequations}
    Since $\mathbf{x}^\star$ is affine in $\nu^\star$ and $\bm \Pi$, and $\nu^\star$ is affine in $\bm \Pi$, $\mathbf{x}^\star$ is affine in $\bm \Pi$ upon substituition. In particular, we have:
    \begin{equation}
        \mathbf{x}^\star = \mathbf{K} \bm \Pi + \mathbf{v}, 
    \end{equation}
    where
    \begin{subequations}
        \begin{align}
            &\mathbf{K} := -\mathbf{S^\top S}\left( \mathbf{M} - \frac{\mathbf{MS^\top S 1 1^\top S^\top SM}}{\mathbf{1^\top S^\top S M S^\top S 1}} \right)\mathbf{ S^\top S}, \\
            &\mathbf{v} := \mathbf{S^\top S}\left(\frac{\mathbf{M S^\top S 1 1^\top S^\top S}}{\mathbf{1^\top S^\top S M S^\top S 1}} - \mathbf{M} \right)\mathbf{S^\top S (\Alr)^\top \bm \beta} + \frac{\mathbf{S^\top SMS^\top S 1}}{\mathbf{1^\top S^\top S M S^\top S 1}},
        \end{align}
    \end{subequations}
    and $\mathbf{x}^\star$ depends on $\bm \alpha$ only through $\mathbf{M}$.
    
    (b) The idea of proving $\bm \lambda$ is also affine in $\bm \Pi$ is the same as part (a). We assume a constraint binding pattern $\mathcal{L}_\mathrm{P}$ and consider a reduced version of economic dispatch \eqref{eq:economic_dispatch}. Then, consider the corresponding dual problem (which is easier since unbinding constraints contribute zero dual variables). The derivation is teadious but if one goes through the whole process would see optimal dual variables of \eqref{eq:economic_dispatch} are unique and 
    \begin{equation}
    	\bm \lambda = \mathbf{C\bm \Pi + w},
    \end{equation}
    where
    \begin{subequations}
    \begin{align}
    	&\mathbf{C} := \rho \left( \mathbf{\frac{1 1^\top}{1^\top Q^{-1} 1} + \left(\frac{H^\mathrm{net} Q^{-1} 1 1^\top}{1^\top Q^{-1}1}-H^\mathrm{net}\right)^\top R^{-1} \left(\frac{H^\mathrm{net} Q^{-1} 1 1^\top}{1^\top Q^{-1}1}-H^\mathrm{net}\right)}\right) \mathbf{(\Acb)^\top \Acr},\\
    	&\mathbf{R := \left( H^\mathrm{net}Q^{-1}H^\mathrm{net}-\frac{H^\mathrm{net}Q^{-1}11^\top Q^{-1} (H^\mathrm{net})^\top}{1^\top Q^{-1} 1} \right)},\\
        	&\begin{aligned}
\mathbf{w}:=&\left(\frac{\mathbf{1 1^\top Q^{-1}}}{\mathbf{1^\top Q^{-1}1^\top}}-\mathbf{I}\right) \mathbf{H}^\mathrm{net} \left(\mathbf{H}^\mathrm{net} \mathbf{Q^{-1}}(\mathbf{H}^\mathrm{net})^\top-\frac{\mathbf{H}^\mathrm{net}\mathbf{Q^{-1} 1 1^\top Q^{-1} (H^\mathrm{net})^\top}}{\mathbf{1^\top Q^{-1} 1}}\right)^{-1}\\
    		\cdot &\left(\frac{\mathbf{H}^\mathrm{net} \mathbf{Q^{-1} 1 1^\top Q^{-1} \bm \mu}}{\mathbf{1^\top Q^{-1} 1}}-\mathbf{H}^\mathrm{net} \mathbf{Q^{-1}} \bm \mu-\mathbf{\bar f}^\mathrm{net}\right) +\frac{\mathbf{1 1^\top Q^{-1}}}{\mathbf{1^\top Q^{-1} 1}} \bm \mu.
    	\end{aligned}
    	\end{align}
    \end{subequations}
    and $\mathbf{H}^\mathrm{net}$ is a reduced version of $\mathbf{H}$ in which only rows associated with binding constraints are preserved.
    
    (c) Let $\bm \Pi_1, \bm \Pi_2 \in \mathcal{P}_\mathcal{B}^\mathrm{cr}$, $t \in (0,1)$, and $\bm \Pi_t:=t\bm \Pi_1 + (1-t) \bm \Pi_2$. We first show that  $\mathbf{x}(\bm \Pi_t)$ has constraint binding pattern $\mathcal{R}$. By Part (a), there exist $\mathbf{K}$ and $\mathbf{v}$ such that $\mathbf{x}(\bm \Pi_i) = \mathbf{K} \bm \Pi_i + \mathbf{v}, i = 1,2$. 
    
    \textit{Claim 1. $\mathbf{x}_t := t\mathbf{x}(\bm \Pi_1) + (1-t)\mathbf{x}(\bm \Pi_2)$ is the unique optimal solution to Problem \eqref{eq:solving_equilibrium_general_pricing} given $\bm \Pi_c$.} Both $\mathbf{x}(\bm \Pi_1)$ and $\mathbf{x}({\bm \Pi_2})$ satisfy KKT conditions of Problem \eqref{eq:solving_equilibrium_general_pricing} with optimal dual variables, say, $(\nu_1, \bm \xi_1)$ and $(\nu_2,\bm \xi_2)$, respectively. Define $\nu := t\nu_1 + (1-t)\nu_2$ and $\bm \xi:=t\bm \xi_1 + (1-t)\bm \xi_2$. First, it is easy to check that $\mathbf{x}_t$ satisfies primal feasibility since it is a convex combination of $\mathbf{x}(\bm \Pi_1)$ and $\mathbf{x}(\bm \Pi_2)$. Multiply the stationarity condition for $\mathbf{x}(\bm \Pi_1)$ by $t$ and that of $\mathbf{x}(\bm \Pi_2)$ by $(1-t)$, and add them together, we obtain:
    \begin{equation}
    	(\Alr)^\top \mathrm{diag}(\bm \alpha) \Alr \mathbf{x}_t + \Alr \bm \beta + \bm \Pi_t = \nu \mathbf{1} + \bm \xi.
    \end{equation}
    Finally, we prove complementary slackness:
    \begin{subequations}
    \begin{align}
    	\bm \xi^\top \mathbf{x}_t &= (t\bm \xi_1 + (1-t)\bm \xi_2)^\top (t\mathbf{x}(\bm \Pi_1) + (1-t)\mathbf{x}(\bm \Pi_2))\\
    	&=t(1-t)\left[\bm \xi_1^\top \mathbf{x}(\bm \Pi_2) + \bm \xi_2^\top \mathbf{x}(\bm \Pi_1)\right] = 0,
    \end{align}
    \end{subequations}
    where the second equality is due to complementary slackness $\bm \xi_i^\top \mathbf{x}(\bm \Pi_i) = 0, i=1,2$ and the last inequality is because $\mathbf{x}(\bm \Pi_1)$ and $\mathbf{x}(\bm \Pi_2)$ have the same constraint binding pattern. Therefore, the tuple $(\mathbf{x}_t,\nu,\bm \xi)$ satisfies KKT conditions of Problem \eqref{eq:solving_equilibrium_general_pricing} given $\bm \Pi_t$. Since the problem is strictly convex, $\mathbf{x}_t$ is the unique optimal solution.
    
    By \textit{Claim 1}, we conclude $\mathbf{x}_t = \mathbf{x}(\bm \Pi_t)$ so the transportation UE induced by $\bm \Pi_t$ has constraint binding pattern $\mathcal{R}$ as it is a convex combination of $\mathbf{x}(\bm \Pi_i),i=1,2$.
    
    By Part (b), we know there exist $\mathbf{C}$ and $\mathbf{w}$ such that $\bm \lambda(\bm \Pi_i) = \mathbf{C}\bm \Pi_i + \mathbf{w},i=1,2$. Stationarity condition of economic dispatch implies $\mathbf{g}(\bm \Pi_i) = \mathbf{Q}^{-1}(\bm \lambda(\bm \Pi_i)-\bm \mu),i=1,2$, where $\mathbf{g}(\bm \Pi_i)$ is the optimal solution to \eqref{eq:economic_dispatch} given $\mathbf{x}(\bm \Pi_i)$. The optimal power injection induced by $\bm \Pi_t$ can be computed by:
    \begin{subequations}
    \begin{align}
    	\mathbf{p}(\bm \Pi_t) &= \mathbf{g}(\bm \Pi_t) - \mathbf{d}(\mathbf{x}(\bm \Pi_t)) = \mathbf{Q}^{-1}(\mathbf{C}\bm \Pi_t+\mathbf{w}-\bm \mu)-\rho \Acb (\Acr)^\top \mathbf{x}(\bm \Pi_t)\\
    	&=t\left[\mathbf{Q}^{-1}(\mathbf{C}\bm \Pi_1+\mathbf{w}-\bm \mu)-\rho \Acb (\Acr)^\top \mathbf{x}(\bm \Pi_1)\right] \\
    	&+ (1-t)\left[ \mathbf{Q}^{-1}(\mathbf{C}\bm \Pi_2+\mathbf{w}-\bm \mu)-\rho \Acb (\Acr)^\top \mathbf{x}(\bm \Pi_2)\right]\\
    	&=t\mathbf{p}({\bm \Pi_1})+(1-t)\mathbf{p}(\bm \Pi_2).
    \end{align}
    \end{subequations}
    Hence, $\mathbf{p}(
    \bm \Pi_t)$, as a convex combination of $\mathbf{p}(\bm \Pi_1)$ and $\mathbf{p}(\bm \Pi_2)$, has constraint binding pattern $\mathcal{L}_\mathrm{P}$. Therefore, the critical region $\mathcal{P}_\mathcal{B}^\mathrm{cr}$ is convex.
\end{proof}

\begin{proof}[Proof of Theorem \ref{thm:BP_convex}]
Fix a constraint binding pattern $\mathcal{B}=(\mathcal{R},\mathcal{L}_\mathrm{P})$ and consider the corresponding critical region $\mathcal{P}_\mathcal{B}^\mathrm{cr}$. We prove that the feasible set of \eqref{eq:BP_elimination} over this fixed critical region is convex. Let $\mathcal{I}:=[\nr]\setminus\mathcal{R}$ be the active route set, and let $\mathbf{S}\in\mathbb{R}^{|\mathcal{I}|\times \nr}$ be the selection matrix that keeps the entries of a route-flow vector indexed by $\mathcal{I}$. Define the reduced active-route flow $\mathbf{z}:=\mathbf{Sx}\in\mathbb{R}^{|\mathcal{I}|}$ and the active-route link-incidence matrix $\mathbf{B}:=\Alr\mathbf{S}^\top\in\mathbb{R}^{\mt\times|\mathcal{I}|}$. The reduced Hessian of the transportation UE problem is
\begin{equation}
    \mathbf{G}:=\mathbf{B}^\top\mathrm{diag}(\bm\alpha)\mathbf{B}\in\mathbb{R}^{|\mathcal{I}|\times|\mathcal{I}|}.
\end{equation}
By the same nonsingularity condition used in the proof of Theorem \ref{thm:GUE_affine}, $\mathbf{G}$ is positive definite on the active routes. For any $\bm\Pi\in\mathcal{P}_\mathcal{B}^\mathrm{cr}$, the KKT conditions of the reduced transportation UE problem are
\begin{subequations}\label{eq:static_price_active_kkt}
\begin{align}
    \mathbf{Gz}+\mathbf{B}^\top\bm\beta+\mathbf{S}\bm\Pi &= \nu\mathbf{1},\\
    \mathbf{1}^\top\mathbf{z} &= 1,
\end{align}
\end{subequations}
where $\nu\in\mathbb{R}$ is the multiplier for the unit-demand constraint and $\mathbf{1}$ has dimension $|\mathcal{I}|$. Define
\begin{equation}
    \mathbf{N}:=\mathbf{G}^{-1}-\frac{\mathbf{G}^{-1}\mathbf{1}\mathbf{1}^\top\mathbf{G}^{-1}}{\mathbf{1}^\top\mathbf{G}^{-1}\mathbf{1}}.
\end{equation}
Solving \eqref{eq:static_price_active_kkt} gives $\mathbf{z}(\bm\Pi)=-\mathbf{N}\mathbf{S}\bm\Pi+\mathbf{z}^0$ for a vector $\mathbf{z}^0$ independent of $\bm\Pi$. Hence $\mathbf{x}(\bm\Pi)=\mathbf{K}\bm\Pi+\mathbf{v}$ with $\mathbf{K}=-\mathbf{S}^\top\mathbf{N}\mathbf{S}$, consistent with Theorem \ref{thm:GUE_affine}. We first note that $\mathbf{N}$ is positive semidefinite. Indeed, for any $\mathbf{q}\in\mathbb{R}^{|\mathcal{I}|}$, the Cauchy--Schwarz inequality under the inner product induced by $\mathbf{G}^{-1}$ gives
\begin{equation}
    \mathbf{q}^\top\mathbf{Nq}
    =\mathbf{q}^\top\mathbf{G}^{-1}\mathbf{q}
    -\frac{(\mathbf{q}^\top\mathbf{G}^{-1}\mathbf{1})^2}{\mathbf{1}^\top\mathbf{G}^{-1}\mathbf{1}}
    \geq 0.
\end{equation}
Therefore, $\mathbf{K}=-\mathbf{S}^\top\mathbf{N}\mathbf{S}$ is negative semidefinite.

We next characterize the two BP-sign constraints in \eqref{eq:constraint_BP}. Fix a road link $\ell\in[\mt]$, let $\mathbf{e}_\ell\in\mathbb{R}^{\mt}$ be the $\ell$-th unit vector, and define
\begin{equation}
    \mathbf{b}_\ell:=\mathbf{B}^\top\mathbf{e}_\ell,\qquad
    y_\ell(\bm\Pi):=\mathbf{b}_\ell^\top\mathbf{z}(\bm\Pi)=(\Alr\mathbf{x}(\bm\Pi))_\ell .
\end{equation}
If $\mathbf{b}_\ell=\mathbf{0}$, then no active route uses link $\ell$, so $y_\ell(\bm\Pi)=0$ throughout $\mathcal{P}_\mathcal{B}^\mathrm{cr}$. In this case, differentiating \eqref{eq:static_price_active_kkt} with respect to $\alpha_\ell$ gives $\partial\mathbf{z}/\partial\alpha_\ell=\mathbf{0}$, and both derivative constraints in \eqref{eq:constraint_BP} hold with equality. Thus link $\ell$ imposes no restriction on $\bm\Pi$. We henceforth consider the case $\mathbf{b}_\ell\neq\mathbf{0}$. Since every active route has positive flow in the chosen critical region, $y_\ell(\bm\Pi)>0$ for every $\bm\Pi\in\mathcal{P}_\mathcal{B}^\mathrm{cr}$.

Differentiate \eqref{eq:static_price_active_kkt} with respect to $\alpha_\ell$. Since $\partial\mathbf{G}/\partial\alpha_\ell=\mathbf{b}_\ell\mathbf{b}_\ell^\top$ and $\mathbf{b}_\ell^\top\mathbf{z}=y_\ell$, the derivative $\mathbf{z}_\ell':=\partial\mathbf{z}/\partial\alpha_\ell$ satisfies
\begin{subequations}\label{eq:z_alpha_sensitivity}
\begin{align}
    \mathbf{G}\mathbf{z}_\ell' + y_\ell\mathbf{b}_\ell &= \nu_\ell'\mathbf{1},\\
    \mathbf{1}^\top\mathbf{z}_\ell' &=0,
\end{align}
\end{subequations}
where $\nu_\ell':=\partial\nu/\partial\alpha_\ell$. Solving \eqref{eq:z_alpha_sensitivity} gives
\begin{equation}
    \mathbf{z}_\ell'=-y_\ell\mathbf{N}\mathbf{b}_\ell.
\end{equation}
Let $\mathbf{a}_\ell:=\mathbf{N}\mathbf{b}_\ell$ and define the scalar
\begin{equation}
    \sigma_\ell:=2\frac{\mathbf{1}^\top\mathbf{G}^{-1}\mathbf{b}_\ell}{\mathbf{1}^\top\mathbf{G}^{-1}\mathbf{1}}
    -(\mathbf{B}^\top\bm\beta)^\top\mathbf{a}_\ell .
\end{equation}
Using Proposition \ref{prop:derivatives_TT_TP}, $\partial\mathbf{x}/\partial\alpha_\ell=\mathbf{S}^\top\mathbf{z}_\ell'$, and $\Alr\partial\mathbf{x}/\partial\alpha_\ell=\mathbf{B}\mathbf{z}_\ell'$, we obtain
\begin{align}
    \frac{\partial \Phi_\mathrm{T}}{\partial\alpha_\ell}
    &=y_\ell^2+(2\mathbf{Gz}+\mathbf{B}^\top\bm\beta)^\top\mathbf{z}_\ell' \notag\\
    &=y_\ell^2-y_\ell(2\mathbf{Gz}+\mathbf{B}^\top\bm\beta)^\top\mathbf{a}_\ell . \label{eq:dt_static_pricing_proof}
\end{align}
Moreover, since $\mathbf{G}\mathbf{N}=\mathbf{I}-\mathbf{1}\mathbf{1}^\top\mathbf{G}^{-1}/(\mathbf{1}^\top\mathbf{G}^{-1}\mathbf{1})$ and $\mathbf{1}^\top\mathbf{z}=1$,
\begin{equation}
    \mathbf{z}^\top\mathbf{G}\mathbf{a}_\ell
    =\mathbf{z}^\top\mathbf{G}\mathbf{N}\mathbf{b}_\ell
    =y_\ell-\frac{\mathbf{1}^\top\mathbf{G}^{-1}\mathbf{b}_\ell}{\mathbf{1}^\top\mathbf{G}^{-1}\mathbf{1}}.
\end{equation}
Substituting this identity into \eqref{eq:dt_static_pricing_proof} yields
\begin{equation}
    \frac{\partial \Phi_\mathrm{T}}{\partial\alpha_\ell}
    =y_\ell(\sigma_\ell-y_\ell).
\end{equation}
Because $y_\ell(\bm\Pi)>0$ and $y_\ell(\bm\Pi)$ is affine in $\bm\Pi$, the constraint $\partial\Phi_\mathrm{T}/\partial\alpha_\ell\geq0$ is equivalent, on this critical region, to the affine inequality $y_\ell(\bm\Pi)\leq\sigma_\ell$. If $\sigma_\ell\leq0$, this affine inequality has no solution in the present case $y_\ell>0$, and the corresponding feasible set is empty, which is still convex.

For the power-system derivative, Proposition \ref{prop:derivatives_TT_TP} gives
\begin{equation}
    \frac{\partial \Phi_\mathrm{P}}{\partial\alpha_\ell}
    =\bm\pi(\bm\lambda(\bm\Pi))^\top\mathbf{S}^\top\mathbf{z}_\ell'
    =-y_\ell\left[\mathbf{S}\bm\pi(\bm\lambda(\bm\Pi))\right]^\top\mathbf{a}_\ell .
\end{equation}
Since $y_\ell(\bm\Pi)>0$, the constraint $\partial\Phi_\mathrm{P}/\partial\alpha_\ell\geq0$ is equivalent to
\begin{equation}
    \left[\mathbf{S}\bm\pi(\bm\lambda(\bm\Pi))\right]^\top\mathbf{a}_\ell\leq0.
\end{equation}
The map $\bm\lambda(\bm\Pi)=\mathbf{C}\bm\Pi+\mathbf{w}$ is affine by Theorem \ref{thm:GUE_affine}, and $\bm\pi(\bm\lambda)$ is linear in $\bm\lambda$ by definition of the route-specific charging cost. Therefore, the above constraint is an affine inequality in $\bm\Pi$. We have shown that every derivative-sign constraint in \eqref{eq:constraint_BP} is either vacuous, empty, or equivalent to an affine inequality on the fixed critical region.

It remains to handle the revenue constraint \eqref{eq:constraint_revenue}. Since $\mathbf{x}(\bm\Pi)=\mathbf{K}\bm\Pi+\mathbf{v}$ and $\mathbf{K}$ is negative semidefinite, the revenue function
\begin{equation}
    r(\bm\Pi):=\bm\Pi^\top\mathbf{x}(\bm\Pi)=\bm\Pi^\top\mathbf{K}\bm\Pi+\bm\Pi^\top\mathbf{v}
\end{equation}
is concave in $\bm\Pi$. Hence its superlevel set $\{\bm\Pi:r(\bm\Pi)\geq\theta\}$ is convex. By Theorem \ref{thm:GUE_affine}, $\mathcal{P}_\mathcal{B}^\mathrm{cr}$ is convex. Consequently, the feasible set of \eqref{eq:BP_elimination} is the intersection of the convex critical region $\mathcal{P}_\mathcal{B}^\mathrm{cr}$, finitely many affine derivative-sign constraints, and the convex superlevel set of the concave revenue function. This intersection is convex, and therefore \eqref{eq:BP_elimination} is a convex feasibility problem over the fixed critical region.
\end{proof}

\newpage
\subsection{Supplementary Materials for \S\ref{sec:numerical}}\label{sec:additional_numerical_study}
\subsubsection{The coupled network setting.} The base parameters for the coupled system considerd in our numeircal study is summarized in Table \ref{table:power_system_setting}.
\begin{table*}[h]
  \centering
  \renewcommand{\arraystretch}{1}
  \setlength{\tabcolsep}{3pt}
  \begin{tabularx}{\textwidth}{@{}X r@{}}
    \toprule
    \multicolumn{2}{l}{\textit{Power System}}\\
    \midrule
    Quadratic coefficients & $\mathbf{Q}=\mathrm{diag}([0.11,0.085,0.1225,0,0,0,0,0,0]^\top)$ (\$/MW$^2$)\\
    Linear coefficients & $\bm \mu = [5,1.2,1,0,0,0,0,0,0]^\top$ (\$/MW)\\
    Base load & $\mathbf{d}_0 = [0,480,0,10,160,80,0,40,120]^\top$ (MW)\\
    Transmission line capacities & $\mathbf{\bar f} = [250,250,25,300,10,250,250,250,250]^\top$ (MW)\\
    \midrule
    \multicolumn{2}{l}{\textit{Transportation System}}\\
    \midrule
    OD Pair & (\textit{Davis}, \textit{San Jose})\\
    Number of EVs & $N = 15,000$ (vehicles)\\
    Common charging demand & $\rho = 20\mbox{ (kW)} = 0.02$ (MW)\\
    Travel cost linear coefficient & $\bm \alpha=10^{-3} \times [3.2,3.2,3.2,6.4,9.6,6.4,9.6]^\top$ (\$/vehicle)\\
    Fixed travel cost & $\bm \beta=[1.6, 20.8, 22.4, 19.2, 12.8, 12.8, 19.2]^\top$ (\$)\\
    \bottomrule
  \end{tabularx}
  \caption{Parameters of the coupled system. 
  }
  \label{table:power_system_setting}
\end{table*}

\subsubsection{Sensitivities of BP}
In \S\ref{sec:numerical} we fix network setting\jqe{s} and study the BPs induced by perturbations of $\bm \alpha$ and $\bar{\mathbf{f}}$. It is also interesting to study \textit{sensitivities of BPs}, i.e., under a setting where some type of BP occurs, \mh{how parameters other than $\bm \alpha$ and $\bar{\mathbf{f}}$ can impact the extent of BPs.} We are particularly interested in changes of two aspects: (i) \textit{BP strength}, i.e., the magnitude of the derivatives $\partial \Phi_\mh{s}/\partial \zeta$, where $s \in \{\mathrm{T,P}\}$ and $\zeta$ can be either $\alpha_{\ell_\mathrm{T}}$ or $\bar{f}_{\ell_\mathrm{P}}$; (ii) \textit{overall social cost increment caused by BP}, i.e., the maximal possible increment in $\Phi_s$ cased by the BP \mh{over certain range of the varying parameter}. \mh{BP strength} and \mh{overall social cost increment} reflect how bad BP could be if one road/line gets expanded by a \textit{tiny}/\textit{large} amount. \jqe{Both are important for coupled infrastructure system planning.}

We use the same setting as that where type P-T and P-P BPs occur (see  Figure \ref{fig:my_2by3} (c)-(f)). We fix the number of EVs $N$ and \mh{consider two} per-vehicle charging demand values $\rho \in \{0.002,0.01\}$MW \mh{(the two numbers correspond to the typical Level 1 and Level 2 charger capacities)}. This allows us to isolate the impact of charging intensity. A higher $\rho$ represents situations \mh{with chargers that have higher power ratings.}

\begin{table}[h]
  \centering
  \begin{tabularx}{\columnwidth}{X|X|X|l}
    \toprule
    $\rho$ & $\partial\Phi_\mathrm{T}/\partial \bar{f}_{6,7}$ & $\partial \Phi_\mathrm{P}/\partial \bar{f}_{6,7}$ & Total Increment of $\Phi_\mathrm{T} \& \Phi_\mathrm{P}$ \\
    \midrule
    $0.002$MW & $99.35$ & $46.97$ & $43923(8.4\%)~|~1302(6.3\%)$\\
    $0.01$MW & $5.30$ & $42.62$ & $20349(3.9\%)~|~4855(19.6\%)$\\
    \bottomrule
  \end{tabularx}
  \caption{Sensitivities of type P-T and P-P BPs to the charging demand $\rho$.}
  \label{table:sensitivity_rho}
\end{table}

Figure \ref{fig:sensitivity_rho} illustrates \mh{BP strength} and \mh{overall social cost increments} when varying $\rho$, and Table \ref{table:sensitivity_rho} includes specific numerical values of BP strength and overall social cost increments. The derivatives in the second and third columns are evaluated at $\bar{f}_{6,7} = 10$MW (i.e., the slopes at $10$MW of lines in Figure \ref{fig:sensitivity_rho}), and the social cost increments in the fourth column are taken as the maximal possible increments when expanding line $(6,7)$ from $10$MW to $80$MW (i.e., the total increments of lines in Figure \ref{fig:sensitivity_rho}). The numbers in parentheses are the percentages of increases compared to the social cost values when $\bar{f}_{6,7} = 10$MW.

The second and third  columns of Table \ref{table:sensitivity_rho} show \mh{lower charging demand} corresponds to stronger strengths of type P-T and P-P BPs. Type P-T BP strength \mh{increases} because reducing $\rho$ \mh{encourages traffic flow to relocate to route $11$, which as discussed in the previous section, is crucial for the increase of $\Phi_\mathrm{T}$. Meanwhile, the traffic flow relocation induces load relocation adding generation burden to the expensive generator $2$ (see also Figure \ref{fig:my_2by3}-(c)), which explains the increase in type P-P BP strength.} However, the \mh{type P-T BP} strength in fact drops to $0$ when $\rho$ decreases to $0$, as then flow relocation does not induce load relocation, \mh{and type P-P BP strength becomes negative when $\rho$ approaches $0$, since in this case transportation and power systems decouple, and expanding a congested line improves $\Phi_\mathrm{P}$.}

Both Figure \ref{fig:sensitivity_rho} and the fourth column of Table \ref{table:sensitivity_rho} demonstrate \mh{lower} charging demand corresponds to \mh{higher} total increment of $\Phi_\mathrm{T}$, and \mh{lower} total increment of $\Phi_\mathrm{P}$. \mh{This is because lower $\rho$ corresponds to smaller charging price differentials, which results in traffic flow relocation to route 11, worsening the transportation system performance. Although the type P-P BP strength increases as $\rho$ decreases, the interval of $\bar{f}_{6,7}$ within which $\Phi_\mathrm{P}$ increases is getting narrower, since it is easier to make line $(6,7)$ uncongested with smaller $\rho$. This explains why the increment in $\Phi_\mathrm{P}$ decreases as $\rho$ decreases.}

The occurrences of the phenomena shown in Figure \ref{fig:sensitivity_rho} and Table \ref{table:sensitivity_rho} can be attributed to the short-cut route $11$. \mh{Without route $11$ (i.e., removing (\textit{Fremont}, \textit{Mtn.View}))}, increasing \mh{charging demand} no longer induces the same changes, as shown in Figure \ref{fig:sensitivity_rho_2}. \mh{In this case, both $\Phi_\mathrm{T}$ and $\Phi_\mathrm{P}$ remain constant.}

\begin{figure}[!htbp]
  \centering
  \begin{subfigure}[t]{0.20\textwidth}
    \centering
    \includegraphics[width=\textwidth]{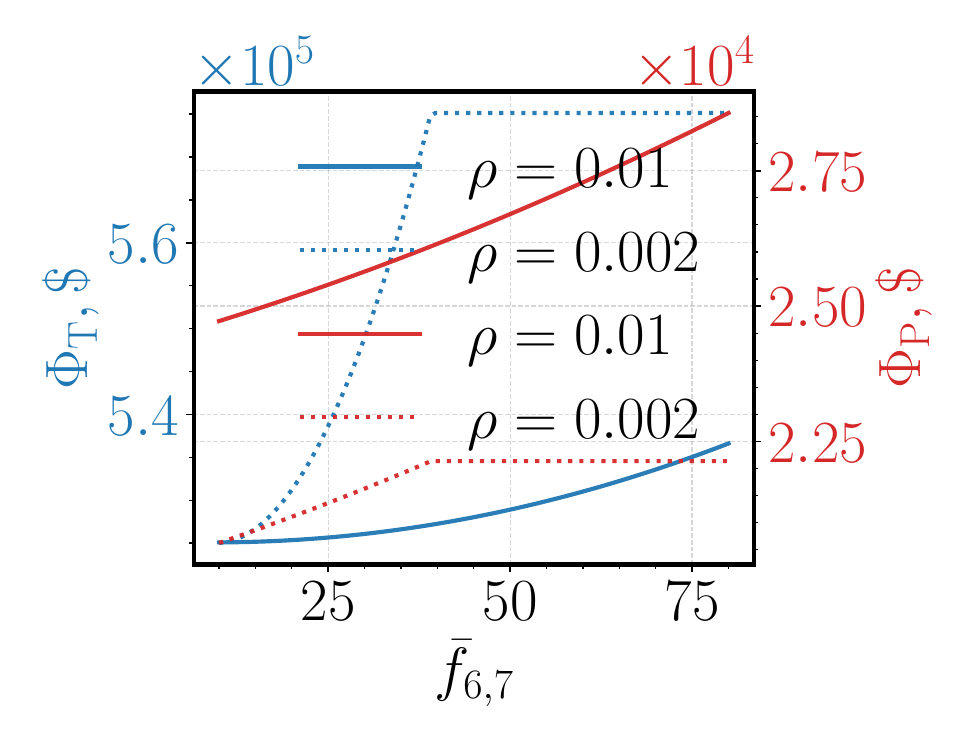}
    \caption{\mh{With route $11$}.}
    \label{fig:sensitivity_rho}
  \end{subfigure}
  \hfill
  \begin{subfigure}[t]{0.20\textwidth}
    \centering
    \includegraphics[width=\textwidth]{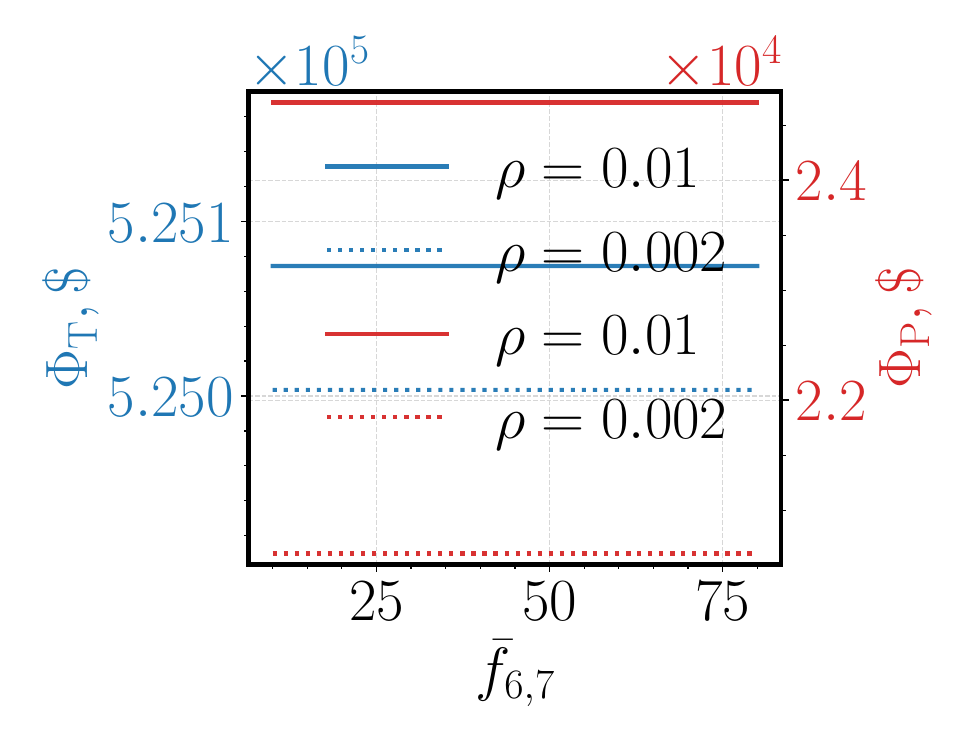}
    \caption{\mh{Without route $11$ (i.e., setting $\alpha_{\mathrm{Fr,M}} = 10^{10}$}).}
    \label{fig:sensitivity_rho_2}
  \end{subfigure}
  \hfill
  \begin{subfigure}[t]{0.20\textwidth}
    \centering
    \includegraphics[width=\textwidth]{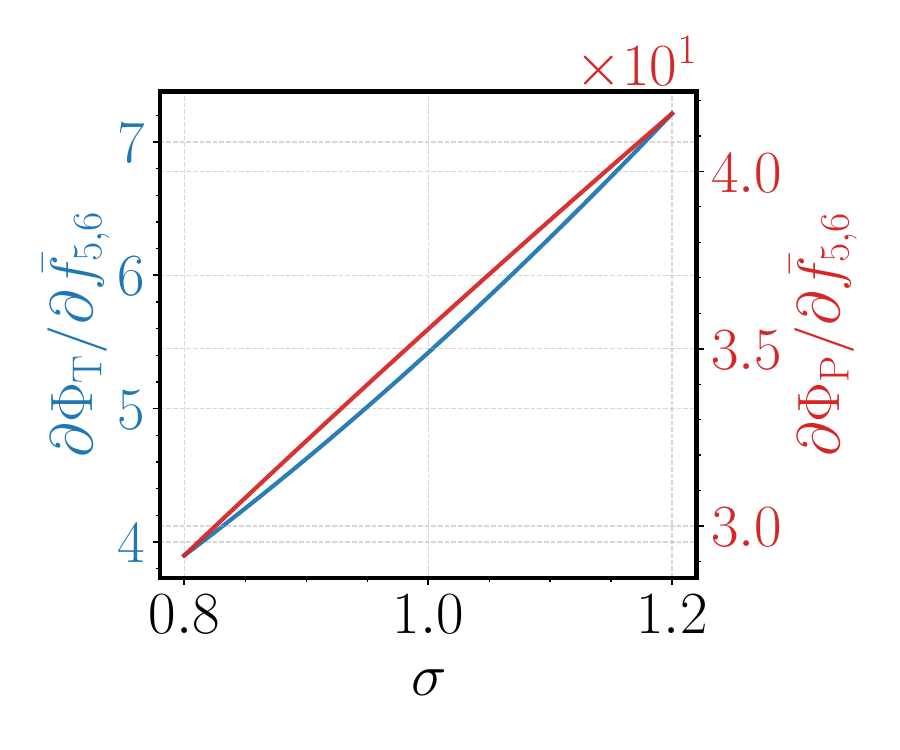}
    \caption{Derivatives versus $\sigma$.}
    \label{fig:derivatives_trend}
  \end{subfigure}
  \hfill
  \begin{subfigure}[t]{0.20\textwidth}
    \centering
    \includegraphics[width=\textwidth]{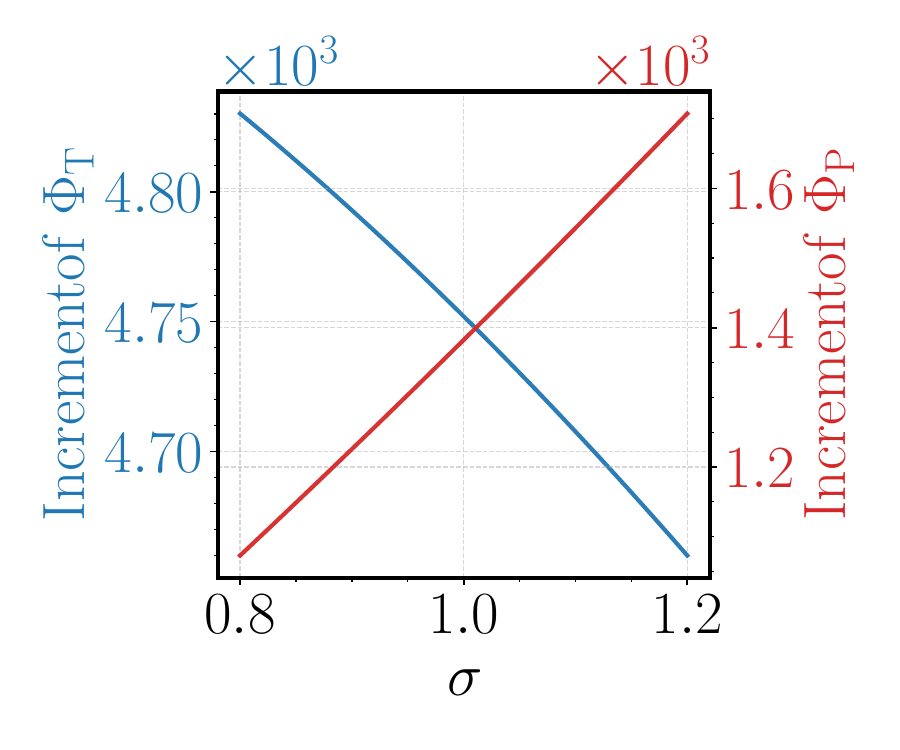}
    \caption{Increments versus $\sigma$.}
    \label{fig:increments}
  \end{subfigure}
  \caption{Sensitivities of BP metrics to changes in $\rho$ and $\mathbf{Q}$. Panels (a)--(b) use the same setting as Figure \ref{fig:my_2by3}-(a)\&(b), while panels (c)--(d) use the same setting as Figure \ref{fig:my_2by3}-(c)--(f).}
  \label{fig:sensitivity_rho_overall}
\end{figure}

%
%
%
%

\begin{table}[!h]
  \centering
  \begin{tabularx}{\columnwidth}{l|X|X|l}
    \toprule
    Scaling Factor & $\partial\Phi_\mathrm{T}/\partial \bar{f}_{6,7}$ & $\partial \Phi_\mathrm{P}/\partial \bar{f}_{6,7}$ & Total Increment of $\Phi_\mathrm{T} \& \Phi_\mathrm{P}$ \\
    \midrule
    $\sigma = 0.8$ & $3.90$ & $29.17$ & $4830~|~1073$\\
    $\sigma = 1.0$ & $5.42$ & $35.55$ & $4752~|~1382$\\
    $\sigma = 1.2$ & $7.21$ & $41.63$ & $4660~|~1707$\\
    \bottomrule
  \end{tabularx}
  \caption{Sensitivities of type P-T and P-P BPs to $\mathbf{Q}$.}
  \label{table:sensitivity_Q}
\end{table}

We also study the sensitivities of type P-T and P-P BPs to the quadratic coefficients $\mathbf{Q}$ of generation costs, under the same setting generating Figure \ref{fig:my_2by3}-(c)-(f). We choose to scale all quadratic coefficients simultaneously by a factor $\sigma > 0$. Table \ref{table:sensitivity_Q} summarizes derivatives and increments of $\Phi_\mathrm{T}$ and $\Phi_\mathrm{P}$. The results show that upscaling $\mathbf{Q}$ enhances the strengths of both BPs. In panels (c)--(d) of Figure \ref{fig:sensitivity_rho_overall}, we plot both derivatives as functions of the scaling factor $\sigma$, which corroborates our finding that larger curvature (i.e., $\mathbf{Q}$) of the generation costs could lead to stronger BP strengths that could be detrimental if a slight amount of line expansion is deployed. Figure \ref{fig:increments} shows how the increments of $\Phi_\mathrm{T}$ and $\Phi_\mathrm{P}$ vary with $\sigma$. Larger scaling factor $\sigma$ leads to larger increase $\Phi_\mathrm{P}$, since the generators are made more expensive. However, larger $\sigma$ leads to smaller increase in $\Phi_\mathrm{T}$, which can be attributed to the larger charging price differentials prevent travelers from shifting to routes worsening the transportation system.

\subsubsection{BP Mitigation}\label{subsec:mitigation}
We study the mitigation effectiveness of pricing policies proposed in \S\ref{sec:mitigation}. Figure \ref{fig:system_optimal_pricing} presents the results of applying transportation/power system-optimal pricing policies $\bm \Pi_\mathrm{T}^\star$ and $\bm \Pi_\mathrm{P}^\star$ on previous examples.

Figure \ref{fig:eliminate_a} shows type T-P BP (see Figure \ref{fig:new_tp}) is eliminated by $\bm \Pi_\mathrm{P}^\star$. Power system-optimal pricing, as explained in \S\ref{sec:mitigation}, reimburses route travel costs, so effectively incentivizes route selection based on LMPs. Expanding (\textit{Fremont, San Jose}) then has no effect  on route choices, and thus does not lead to load relocation to expensive generators.

Figure \ref{fig:eliminate_b} adopts the same setting as that induces type T-T BP (see Figure \ref{fig:my_2by3}-(a)\&(b)), except that we expand line capacity of $(1,4)$ to $280$MW, which does not affect the occurrence of type T-T BP as the power system is uncongested even when $(1,4)$ has capacity $250$MW. The expansion makes $\bm \Pi_\mathrm{T}^\star$ applicable since under $\bm \Pi_\mathrm{T}^\star$, route choices are independent of LMPs, which could lead to a charging load spatial distribution infeasible to the power system operation. As explained in \S\ref{sec:mitigation}, under $\bm \Pi_\mathrm{T}^\star$, the traffic UE is optimal in the sense that it minimizes $\Phi_\mathrm{T}$. Expanding any road could only lead to improvement of $\Phi_\mathrm{T}$ by simple sensitivity analysis.

Figure \ref{fig:eliminate_c} shows type P-P BP is eliminated by $\bm \Pi_\mathrm{P}^\star$, while type P-T BP still occurs. Under $\bm \Pi_\mathrm{P}^\star$, load would be distributed to minimize $\Phi_\mathrm{P}$, and thus expanding any line would never make $\Phi_\mathrm{P}$ worse off. Failure to eliminate type P-T BP is consistent with Theorem \ref{thm:sysoptPi}-(b), since the LMP-only dependent route choices could lead to traffic flow distribution that worsens  the traffic condition (i.e., travelers may accumulate at cheap chargers and congest routes).

What is interesting in Figure \ref{fig:system_optimal_pricing} is for eliminated BPs, the corresponding social costs are lower compared to those when BPs occur, which is also consistent with our discussion in \S\ref{sec:mitigation}. However, it seems that in Figure \ref{fig:eliminate_c}, the improvement in $\Phi_\mathrm{P}$ \mh{(see the red solid and dashed lines)} comes with the cost of degradation of $\Phi_\mathrm{T}$ (see the blue solid and dashed lines). Type P-T BP becomes even stronger in this case (i.e., the blue dashed line is much steeper than the blue solid line), which is obviously not what system planners want to see.


\begin{figure}[!htbp] 
\centering
\begin{subfigure}[t]{0.19\textwidth}
\centering
\includegraphics[width=\textwidth]{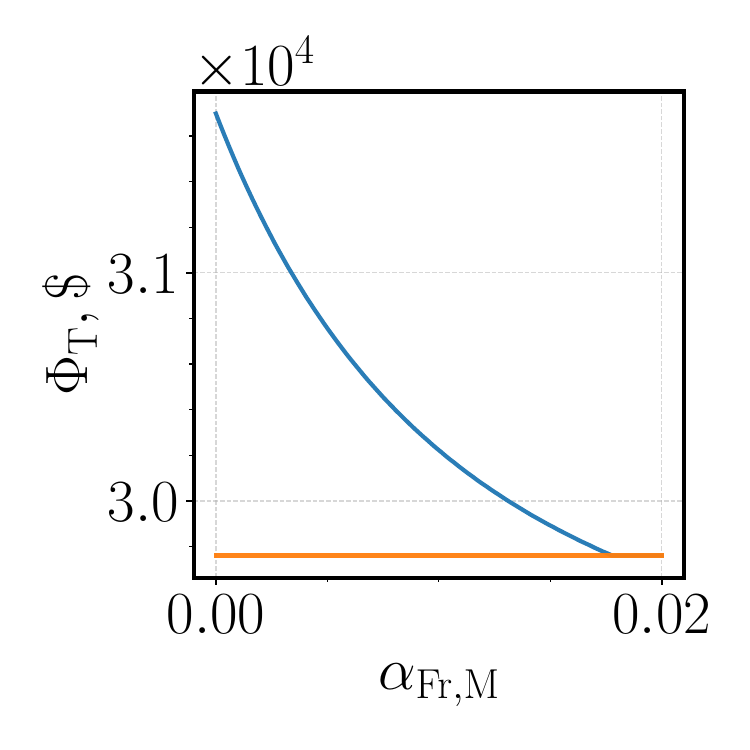}
\caption{Type T-P BP eliminated.}
\label{fig:eliminate_a}
\end{subfigure}
\hfill
\begin{subfigure}[t]{0.19\textwidth}
\centering
\includegraphics[width=\textwidth]{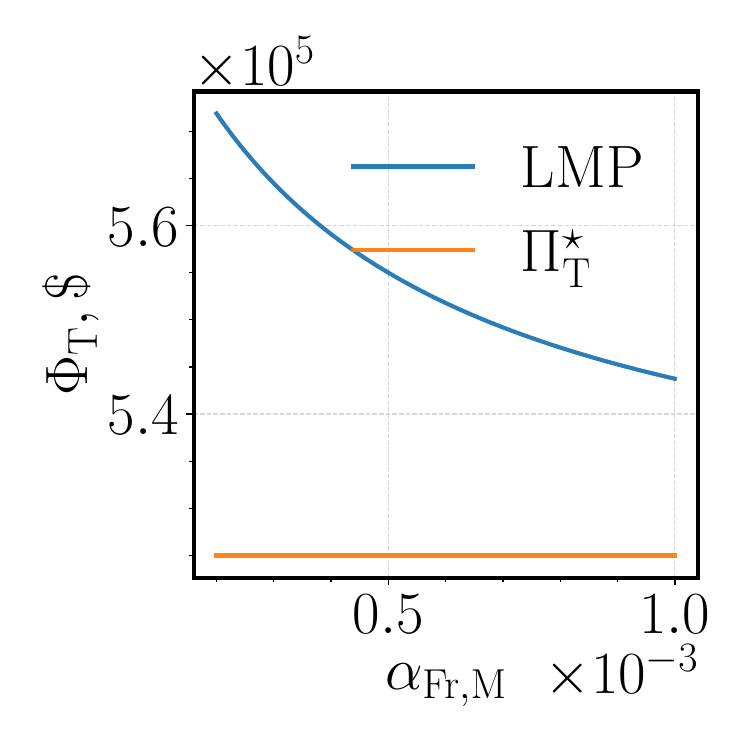}
\caption{Type T-T BP eliminated.}
\label{fig:eliminate_b}
\end{subfigure}
\hfill
\begin{subfigure}[t]{0.19\textwidth}
\centering
\includegraphics[width=\textwidth]{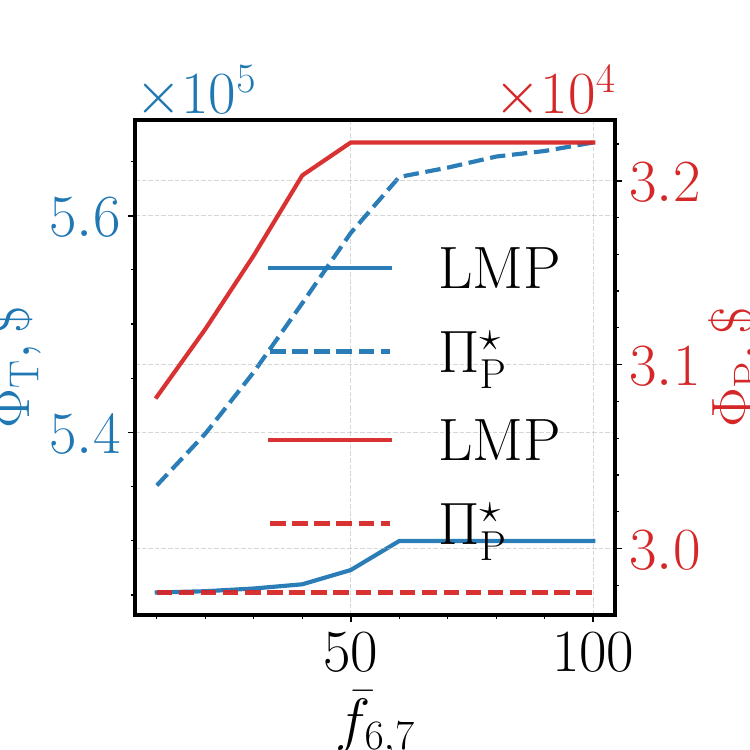}
\caption{Type P-P BP eliminated.}
\label{fig:eliminate_c}
\end{subfigure}
\hfill
\begin{subfigure}[t]{0.19\textwidth}
\centering
\includegraphics[width=\textwidth]{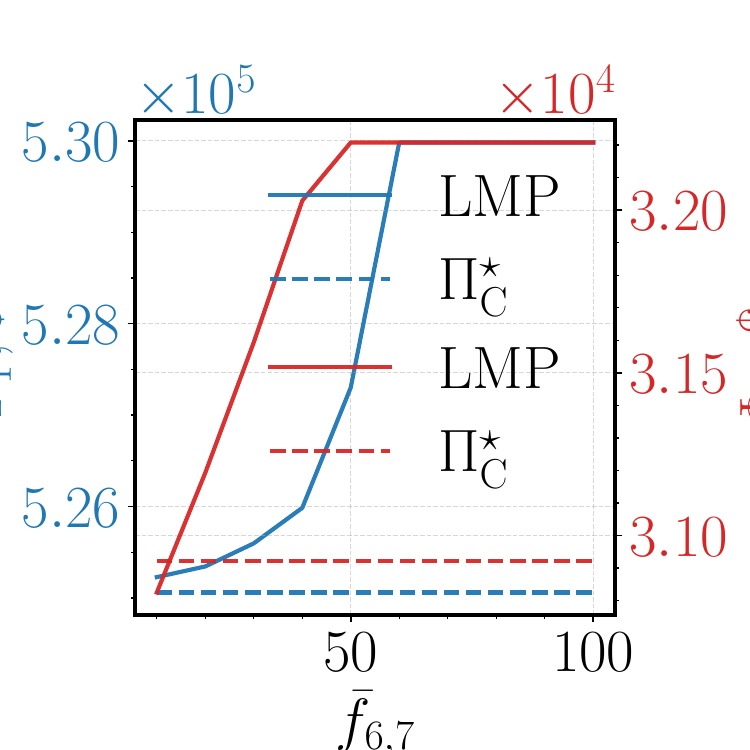}
\caption{Type P-T BP eliminated.}
\label{fig:eliminate_d}
\end{subfigure}
\caption{BPs are eliminated by system-optimal pricing policies $\bm \Pi_\mathrm{T}^\star$, $\bm \Pi_\mathrm{P}^\star$, and $\bm \Pi_\mathrm{C}^\star$. Figure \ref{fig:eliminate_a}, \ref{fig:eliminate_b}, \ref{fig:eliminate_c}, and \ref{fig:eliminate_d} use the same setting generating Figure \ref{fig:new_tp}, \ref{fig:my_2by3}-(a)\&(b), \ref{fig:my_2by3}-(c)-(f), and \ref{fig:my_2by3}-(c)-(f), respectively.}
\label{fig:system_optimal_pricing}
\end{figure}

\mh{
For the same setting as Figure \ref{fig:eliminate_c}, we attempted to eliminate the type P-T BP using $\bm \Pi_\mathrm{T}^\star$ (see Theorem \ref{thm:sysoptPi}-(a)). However, $\bm \Pi_\mathrm{T}^\star$ induces a power load distribution that is infeasible for power system operation, which is not a surprise since transportation system-optimal pricing policy optimizes only the transportation system and may therefore cause power system failures. In contrast, the combined system-optimal pricing $\bm \Pi_\mathrm{C}^\star$ successfully eliminates both type P-T and P-P BPs, as shown in Figure \ref{fig:eliminate_d}, and achieves an even lower $\Phi_\mathrm{T}$ than that in Figure \ref{fig:eliminate_c}. Although Theorem \ref{thm:sysoptPi}-(c) does not gurantee such an elimination, the results show that a system-optimal pricing policy may still eliminate them in practice.
}

\textbf{Computational complexity of static pricing.} We further examine the computational effort required to search for a static pricing policy over critical regions. We choose the admissible price set as $\mathcal{P}=[0,200]^{n_R}$, and the feasibility problem \eqref{eq:BP_elimination} is solved for each discovered critical region. We consider three coupled systems with increasing size: IEEE 9-bus/Bay Area (\texttt{small}), IEEE 14-bus/Bay Area (\texttt{medium}), and IEEE 57-bus/Bay Area (\texttt{large}). The \texttt{small} and \texttt{medium} coupled systems use the same Bay Area transportation network as the numerical study in the main text (cf. Section~\ref{sec:numerical}), coupled with the IEEE 9-bus and IEEE 14-bus power systems, respectively. For the \texttt{large} coupled system, we couple the IEEE 57-bus power system with the Anaheim transportation network \cite{transportationNetworksResearch}. To obtain a computationally tractable route/charging representation and isolate the computational effort investigation of static pricing, we only retain the five highest-demand OD pairs, set the remaining OD-pair demands to zero, and use the resulting 50 route/charging plans and 9 charging nodes over 914 transportation links.

We use the active-set adjacency procedure based on the critical-region construction in MPT \cite{kvasnica2004multi}. Region discovery and static-price feasibility checks are interleaved: after a critical region is discovered through the active-set adjacency procedure, we solve \eqref{eq:BP_elimination} and continue to an unexplored adjacent critical region only when \eqref{eq:BP_elimination} is infeasible for the current critical region. For comparison, we also use a sampling-based procedure that tests newly observed critical regions before drawing another batch of prices. Each experiment uses five independent random seeds, and the wall-clock time measures critical-region discovery, GUE solves, and feasibility checks. Table~\ref{tab:static_pricing_complexity_appendix} summarizes the results.

\begin{table}[!htbp]
  \centering
  \renewcommand{\arraystretch}{1.3}
  \setlength{\tabcolsep}{2pt}
  \begin{tabularx}{\linewidth}{@{}>{\raggedright\arraybackslash}p{0.27\linewidth}>{\raggedright\arraybackslash}p{0.14\linewidth}>{\raggedright\arraybackslash}p{0.18\linewidth}>{\raggedright\arraybackslash}p{0.14\linewidth}X@{}}
    \hline
    Coupled system & Method & Regions enumerated & Early stop & Total wall-clock time \\
    \hline
    \multirow{3}{=}{IEEE 9-bus/Bay Area (\texttt{small})} & Active-set & $5.4\,\textcolor{gray!65}{\pm\,0.5}$ & \textcolor{green!60!black}{\(\checkmark\,\checkmark\,\checkmark\,\checkmark\,\checkmark\)} & $1.63\,\textcolor{gray!65}{\pm\,0.23}$ s \\
     & Sampling (20,000) & $12.2\,\textcolor{gray!65}{\pm\,1.1}$ & \textcolor{green!60!black}{\(\checkmark\,\checkmark\,\checkmark\,\checkmark\)}\,\textcolor{red}{\(\times\)} & $9.22\,\textcolor{gray!65}{\pm\,6.06}$ s \\
     & Sampling (100,000) & $12.6\,\textcolor{gray!65}{\pm\,1.67}$ & \textcolor{green!60!black}{\(\checkmark\,\checkmark\,\checkmark\,\checkmark\,\checkmark\)} & $10.20\,\textcolor{gray!65}{\pm\,7.62}$ s \\
    \hline
    \multirow{2}{=}{IEEE 14-bus/Bay Area (\texttt{medium})} & Active-set & $3.8\,\textcolor{gray!65}{\pm\,1.8}$ & \textcolor{green!60!black}{\(\checkmark\,\checkmark\,\checkmark\,\checkmark\,\checkmark\)} & $3.08\,\textcolor{gray!65}{\pm\,1.69}$ s \\
     & Sampling (1,000) & $1.0\,\textcolor{gray!65}{\pm\,0.0}$ & \textcolor{green!60!black}{\(\checkmark\,\checkmark\,\checkmark\,\checkmark\,\checkmark\)} & $0.027\,\textcolor{gray!65}{\pm\,0.004}$ s \\
    \hline
    \multirow{2}{=}{IEEE 57-bus/Anaheim (\texttt{large})} & Active-set & $1.2\,\textcolor{gray!65}{\pm\,0.4}$ & \textcolor{green!60!black}{\(\checkmark\,\checkmark\,\checkmark\,\checkmark\,\checkmark\)} & $22.36\,\textcolor{gray!65}{\pm\,14.53}$ s \\
     & Sampling (1,000) & $2.2\,\textcolor{gray!65}{\pm\,0.4}$ & \textcolor{green!60!black}{\(\checkmark\,\checkmark\,\checkmark\,\checkmark\,\checkmark\)} & $27.34\,\textcolor{gray!65}{\pm\,8.75}$ s \\
    \hline
  \end{tabularx}
  \caption{\textit{Each number shown is mean$_{\pm\text{std}}$ over five random seeds. Each early-stop marker corresponds to one seed; \textcolor{green!60!black}{$\checkmark$} indicates that the run found a feasible static price, whereas \textcolor{red}{$\times$} indicates failure to find one. The number in parentheses after ``Sampling'' is the number of random prices drawn, in addition to the LMP and zero-price baselines. The active-set procedure tests each newly enumerated critical region immediately, and the sampling procedure tests each newly observed batch of critical regions before drawing another batch.}}
  \label{tab:static_pricing_complexity_appendix}
\end{table}

The sampling-based procedure can miss critical regions intersecting $\mathcal{P}$, and that is the reason why one of five seeds of the sampling method run for the \texttt{small} system failed to find a feasible static price with 20,000 random prices. As a higher-budget diagnostic, we repeated the five-seed experiment with 100,000 random prices in addition to the same baselines. All five seeds then found a feasible static price after processing $12.6\,{\pm\,1.67}$ regions and in $10.20\,{\pm\,7.62}$ seconds. Increasing the sampling budget can improve coverage, but exhaustive enumeration remains necessary for a nonexistence certificate over $\mathcal{P}$. These experiments demonstrate the possibility of early termination, but do not remove the worst-case exponential dependence on the number of critical regions.

\end{document}